\documentclass[12pt]{article}
\usepackage[utf8]{inputenc}
\usepackage[margin=1in]{geometry}
\usepackage{multirow,multicol}
\usepackage{amsfonts}
\usepackage{amsmath,amssymb,amsthm}
\usepackage[mathscr]{euscript}
\usepackage{comment}
\usepackage{graphicx,color}
\graphicspath{{image/}}
\usepackage{mathtools}
\usepackage{mathabx} 
\usepackage[linesnumbered,ruled,vlined]{algorithm2e}
\usepackage{setspace}
\usepackage{sectsty}
\usepackage{makecell}
\usepackage{chngcntr}
\usepackage{caption}
\usepackage{subcaption}
\usepackage{booktabs}
\usepackage{etoolbox}
\usepackage{tabularx}
\newcolumntype{Y}{>{\raggedleft\arraybackslash}X}
\newcolumntype{Z}{>{\centering\arraybackslash}X}
\usepackage{enumitem}
\usepackage{underscore} 
\usepackage{array} 
\usepackage[english]{babel}
\usepackage[colorlinks=true,linkcolor=blue,urlcolor=black,citecolor=blue,bookmarksopen=true]{hyperref}
\usepackage{bookmark}
\usepackage{multirow} 

\usepackage{lipsum}
\usepackage{pdflscape}
\usepackage{calrsfs} 

\newcommand{\B}{\fontseries{b}\selectfont}

\DeclareMathAlphabet{\pazocal}{OMS}{zplm}{m}{n} 

\DeclareFontFamily{OT1}{pzc}{}
\DeclareFontShape{OT1}{pzc}{m}{it}{<-> s * [1.10] pzcmi7t}{}
\DeclareMathAlphabet{\mathpzc}{OT1}{pzc}{m}{it}

\usepackage{natbib}
\usepackage{bibunits}

\newtheorem{assumption}{Assumption}

\newtheorem{lemma}{Lemma}

\newtheorem{theorem}{Theorem}
\newtheorem{corollary}{Corollary}

\newtheorem{remark}{Remark}
\theoremstyle{definition}

\usepackage{epstopdf}

\newcolumntype{P}[1]{>{\centering\arraybackslash}p{#1}}

\newcommand{\bm}{\boldsymbol}
\newcommand{\cm}[1]{\mbox{\boldmath$\mathscr{#1}$}}
\newcommand{\cmt}[1]{\mbox{\boldmath\scriptsize$\mathscr{#1}$}}

\newcommand{\Fr}{{\mathrm{F}}}

\newcommand{\ar}{{\mathrm{AR}}}

\newcommand{\HW}{{\mathrm{HW}}}
\def\HH{{\mathrm{\scriptscriptstyle\mathsf{H}}}}

\DeclareMathOperator*{\vect}{vec}

\DeclareMathOperator*{\argmin}{arg\,min}

\DeclareMathOperator*{\var}{var}

\numberwithin{equation}{section}

\title{\vspace{-20pt} Tensor Stochastic Regression for High-dimensional Time Series via CP Decomposition}

\author{Shibo Li
	and 
	Yao Zheng\footnote{Corresponding Author. Assistant Professor, Department of Statistics, University of Connecticut, Storrs, CT 06269. (Email: yao.zheng@uconn.edu).} \\
	\textit{University of Connecticut}}
\date{}

\begin{document}
\begin{bibunit}[apalike]	
\setlength{\parindent}{16pt}
	
\maketitle
	
\begin{abstract}
As tensor-valued data become increasingly common in time series analysis, there is a growing need for flexible and interpretable models that can handle high-dimensional predictors and responses across multiple modes. We propose a unified framework for high-dimensional tensor stochastic regression based on CANDECOMP/PARAFAC (CP) decomposition, which encompasses vector, matrix, and tensor responses and predictors as special cases. Tensor autoregression naturally arises as a special case within this framework. By leveraging CP decomposition, the proposed models interpret the interactive roles of any two distinct tensor modes, enabling dynamic modeling of input-output mechanisms. We develop both CP low-rank and sparse CP low-rank estimators, establish their non-asymptotic error bounds, and propose an efficient alternating minimization algorithm for estimation. Simulation studies confirm the theoretical properties and demonstrate the computational advantage. Applications to mixed-frequency macroeconomic data and spatio-temporal air pollution data reveal interpretable low-dimensional structures and meaningful dynamic dependencies. 
\end{abstract}

\textit{Keywords}:  CANDECOMP/PARAFAC decomposition; High-dimensional time series; Low-rank estimation; Tensor autoregression; Sparse estimation; 

\textit{MSC2020 subject classifications}: Primary 62M10; secondary 62H12, 60G10

\newpage
\section{Introduction}\label{section:intro}

Driven by the need for analyzing complex dynamic systems and advancements in data collection capability, tensor-valued time series data have become increasingly common in many areas, such as economics \citep[e.g.,][]{Chen2022}, finance \citep[e.g.][]{rogers2013multilinear},  transportation engineering \citep[e.g.,][]{Han2024}, medical imaging analyses \citep[e.g.,][]{Bi2021}. Tensor-valued time series  $\cm{Y}_t\in\mathbb{R}^{p_1\times \dots\times p_m}$ are typically  formed when more granular data are available, where $t$ denotes time. This granularity may arise from  the cross-sectional dimension, i.e., when $p_d$ outcome variables are observed at disaggregated levels, such as by industry or geographic region, or across different categories. It may also result from the temporal dimension, which could be due to the incorporation of multiple time lags or the availability of higher-frequency observations. In particular, when modeling mixed-frequency time series data (e.g., yearly and monthly), it is natural to treat the  additional temporal granularity (e.g, $p_d=12$ months per year) as another dimension. This allows  higher-frequency information to be easily leveraged within a modeling framework indexed by a lower-frequency time index $t$; e.g., higher-frequency variables can be  used to predict lower-frequency variables, or vice versa.

In this paper, we consider the tensor stochastic regression for time series data, where the response $\cm{Y}_t\in\mathbb{R}^{p_1\times\cdots\times p_m}$ and the predictor $\cm{X}_t\in\mathbb{R}^{q_1\times\cdots\times q_n}$ can be vectors, matrices or high-order tensors of different dimensions. The tensor forms of the response or predictor can result from more detailed data in variables and temporal dimensions. It is noteworthy that our framework includes vector, matrix and tensor autoregressions \citep[e.g.,][]{Chen2021, wang2022high, wang2024high} as special cases, when $\cm{X}_t$ consists of lagged values of  $\cm{Y}_t$. However, it also allows for mixed-frequency tensor regression, which generalizes the mixed-data sampling (MIDAS) regression \citep{midas2006} to tensor time series. As  emphasized in the rapidly growing literature on matrix- and tensor-valued time series, preserving the tensor form of the data---rather than vectorizing them and applying conventional multivariate modeling---is essential for maintaining multi-dimensional structural information.  This preservation not only enables the extraction of distinct insights across each dimension of the data but also improves estimation efficiency, as dimension reduction can be achieved simultaneously across all directions.

There is an extensive body of literature on tensor regression for $i.i.d.$ data, where  low-rank assumptions are imposed on  coefficient tensors typically via CANDECOMP/PARAFAC (CP) or Tucker decomposition; see \cite{Liu2021} and \cite{Bi2021} for recent reviews. Our paper is most closely related to the literature on tensor regression via CP decomposition. For example,  \cite{zhou2013tensor, wang2017generalized, feng2021brain} focus on models with scalar responses, whereas \cite{Lock2018} and \cite{raskutti2019convex}, among others, study tensor-on-tensor regression, which includes scalar-on-tensor, vector-on-vector, vector-on-tensor regression, etc., as special cases. Our setup aligns with the latter, as we similarly aim for generality. For tensor-on-tensor regression via Tucker decomposition, see the recent work by \cite{LZ2024} and the references therein. Our paper differs considerably from the above works as we focus on time series rather than independent data. From a theoretical perspective, our analysis is complicated by the temporal dependencies in the response and predictor series. From an empirical perspective, the temporal dimensions, whether arising from time lags or mixed frequencies, carry distinct interpretations for tensor time series compared to models for $i.i.d.$ data. 

Recently, tensor decomposition methods  have also been introduced to the modeling of vector-, matrix- and tensor-valued time series. For example, Tucker decomposition  has been  applied to high-dimensional vector \citep{wang2022high} and tensor \citep{wang2024high} autoregressions. Moreover, various matrix autoregressive (MAR) and tensor autoregressive (TAR) models with low-rank assumptions are intrinsically connected to the Tucker decomposition \citep[e.g.,][]{Chen2021, Samadi2024}; see also  \cite{Tsay2024} for a recent survey. \cite{Chang2023} proposes regression with tensor response and vector predictors via CP decomposition for time series data. In contrast, CP decomposition has been rarely explored in autoregression and  tensor-on-tensor regression for time series; an exception is the recent work by \cite{Chan2024}  where CP decomposition is applied to Bayesian vector autoregression (VAR). The popularity of  Tucker decomposition  may be partly due to its more natural connection with the classical reduced-rank VAR \citep{Velu1986, wang2022high}. Nevertheless, CP decomposition has  unique advantages over Tucker decomposition. Indeed, beyond the regression setting, there has been a recent surge of interest in applying  CP decomposition to factor models. Compared to factor models based on  Tucker decomposition \citep[e.g.,][]{Chen2022,Barigozzi2022,Chen2024a}, those based on CP decomposition are attractive \citep[e.g.][]{Babii2024,Han2024,Chen2024} as their identification does not rely on the orthogonality constraints typically required for the identifiability of Tucker decomposition based models. As discussed in these latter works, such constraints can be difficult to interpret and may limit the model's flexibility. This drawback motivates the growing interest in incorporating CP decomposition rather than Tucker decomposition in the construction of factor models.
However, tensor factor models cannot be directly used for forecasting as they are fundamentally unsupervised dimension reduction tools. To fill this gap, we consider predictive modeling for tensor-valued time series via CP decomposition, which simultaneously accomplishes dimension reduction and forecasting.

Furthermore, an important insight revealed in this paper is that, in the context of (auto)regression, CP decomposition can provide meaningful interpretations of \textit{interactive} patterns across various dimensions, as well as input-output dynamic mechanisms. In contrast, methods based on Tucker decomposition, such as the TAR model in \cite{wang2024high},  can only extract patterns separately for each dimension of the tensor-variate data. This can be a significant limitation in empirical analysis.  For instance, applying our method to mixed-frequency macroeconomic data in Section \ref{subsec:macro} reveals the dynamic mechanism between quarterly responses and monthly predictors.
For the air pollution data in Section \ref{subsec:taiwan}, we are able to not only discover patterns within individual dimensions such as pollutants, regions, or months, but also  explore interactions among these dimensions: e.g., do different pollutants concentrate in specific regions or during certain months? Additionally,  we can examine input-output dynamics: how  do past levels of different pollutants influence current concentrations, and how do past pollution in various regions affect the present distribution of pollutants across these regions? 
The main contributions of this paper are summarized as follows:
\begin{itemize}
\item For the general tensor stochastic  regression, we propose dimension reduction via the CP low-rank assumption on the coefficient tensor. This framework encompasses  vector, matrix and tensor autoregressions (TAR) via CP decomposition as special cases, which have not been studied in the literature.  

\item To improve the estimation efficiency and interpretability when the tensor dimensions are large, we further impose entrywise sparsity  on the component vectors. This leads to the sparse CP low-rank stochastic tensor regression and TAR. 
 For both the non-sparse and sparse models, we establish non-asymptotic statistical error bounds for the corresponding estimators. An efficient algorithm is also developed, and its computational advantage is verified through simulation studies.

\item As demonstrated by our empirical analysis, unlike models based on Tucker low-rank structures, our approach  provides meaningful insights into interactions across different dimensions, rather than only extracting patterns from each individual dimension.
\end{itemize}

The remainder of this paper is organized as follows. Section \ref{sec:model} introduces the proposed models. Section \ref{sec:est}  presents  non-asymptotic properties for the proposed non-sparse and sparse CP low-rank estimators. Simulation  and empirical studies are presented in Sections \ref{sec:sim} and \ref{sec:real}, respectively. Section \ref{sec:conclude} concludes with a brief discussion. A detailed discussion of the connections between the proposed model and Tucker low-rank models, the proposed algorithm, additional simulation and empirical results, as well as all technical proofs, is provided in a separate supplementary file.

\section{Tensor stochastic regression}\label{sec:model}
\subsection{Notation and tensor algebra}\label{subsec:notation}
Tensors, a.k.a. multi-dimensional arrays, are higher-order extensions of matrices. A multidimensional array $\cm{A} \in \mathbb{R}^{p_1 \times p_2 \times \cdots \times p_m}$ is called an $m$th order tensor \citep{Kolda2009}.
In this paper, we use bold lowercase letters to represent vectors, e.g., $\bm{x}$, $\bm{y}$; bold capital letters to represent matrices, e.g., $\bm{X}$, $\bm{Y}$; and bold Euler capital letters to represent tensors, e.g., $\cm{A}$, $\cm{B}$. For any $\bm{x}\in\mathbb{R}^p$, denote its $\ell_q$-norm by $\|\bm{x}\|_q=(\sum_{j=1}^p|x_j|^q)^{1/q}$ for $q>0$, its $\ell_0$-norm by $\|\bm{x}\|_0=|\textrm{supp}(\bm{x})|=\sum_{j=1}^p \mathbb{I}_{\{x_j\neq 0\}}$, and its normalization by  $\textrm{Norm}(\bm{x})=\bm{x}/\|\bm{x}\|_2$. For a positive integer $s$, we let $\textrm{Truncate}(\bm{x}, s)$ denote  the operation that sets all but the $s$ largest (in absolute value)  entries of $\bm{x}$ to zero.
For any $\bm{X}\in\mathbb{R}^{p_1\times p_2}$, let $\bm{X}^\top$ and  $\lambda_{\max}(\bm{X})$ (or $\lambda_{\min}(\bm{X})$) denote its transpose and largest (or smallest)  eigenvalue, respectively.
For any sequences $x_n$ and $y_n$, denote $x_n\lesssim y_n$ (or $x_n\gtrsim y_n$) if there is an absolute constant $C>0$ such that $x_n\leq C y_n$ (or $x_n\geq C y_n$). 

The outer product of a set of vectors $\{\bm{z}_d\}_{k=1}^K$, $\bm{z}_k\in\mathbb{R}^{p_k}$, denoted by $\bm{z}_1\circ\cdots\circ\bm{z}_K$, equals $\cm{Z}\in\mathbb{R}^{p_1\times\cdots\times p_K}$, where $\cm{Z}_{i_1\cdots i_K}=\prod_{k=1}^K \bm{z}_{k, i_k}$. The tensor-matrix multiplication of $\cm{Z}\in\mathbb{R}^{p_1\times p_2\times\cdots\times p_K}$ and $\bm{Y}\in\mathbb{R}^{q_k\times p_k}$ on the $k$th mode with $1\leq k\leq K$ is denoted by $\cm{Z}\times_k\bm{Y}\in\mathbb{R}^{p_1\times  p_2\times \cdots\times p_{k-1}\times q_k\times p_{k+1}\cdots\times p_K}$, where $(\cm{Z}\times_k\bm{Y})_{i_1\cdots i_{k-1}j i_{k+1}\cdots i_K}=\sum_{i_k=1}^{p_k}\cm{Z}_{i_1\cdots i_K}\bm{Y}_{ji_k}$. For $m,n\geq1$, the generalized inner product of two tensors $\cm{B}\in\mathbb{R}^{p_1\times \cdots\times p_m\times q_{1}\times\cdots\times q_n}$ and $\cm{X}\in\mathbb{R}^{q_{1}\times\cdots\times q_n}$ on the last $n$ modes is denoted by $\langle\cm{B}, \cm{X} \rangle \in\mathbb{R}^{p_{1}\times\cdots\times p_m}$, where $\langle\cm{B}, \cm{X} \rangle_{i_{1}\cdots i_m} = \sum_{j_{1},\ldots,j_n}\cm{B}_{i_1\cdots i_m j_1\cdots j_n}\cm{X}_{j_{1}\cdots j_n}$. The outer product between two tensors $\cm{A}\in\mathbb{R}^{p_1\times \cdots\times p_m}$ and $\cm{X}\in\mathbb{R}^{q_1 \times \cdots\times q_n}$ is denoted by $\cm{A}\circ\cm{X}\in\mathbb{R}^{p_1\times\cdots\times p_m\times q_1\times \cdots \times q_n}$, where $(\cm{A}\circ\cm{X})_{i_1\cdots i_m j_1\cdots j_n}=\cm{A}_{i_1\cdots i_m}\cm{X}_{j_{1}\cdots j_n}$. The Frobenius norm of a tensor $\cm{Z}\in\mathbb{R}^{p_1\times\cdots\times p_K}$ is $\|\cm{Z}\|_{\Fr} = \sqrt{\sum_{i_1,\dots, i_K} \cm{Z}_{i_1\cdots i_K}^2}$. The vectorization of a tensor, $\vect(\cm{Z})$, is defined analogously to that of a matrix.

For a tensor $\cm{B}\in\mathbb{R}^{p_1\times p_2\times\cdots\times p_N}$, the CANDECOMP/PARAFAC (CP) decomposition 
factorizes it into a sum of component rank-one
tensors as
\begin{align}\label{eq:CP}
	\cm{B} =[\![\bm{\omega};\bm{B}^{(1)},\bm{B}^{(2)},\dots, \bm{B}^{(N)}]\!]= \sum_{r=1}^R\omega_r\bm{\beta}_{r,1}\circ\bm{\beta}_{r,2}\circ\cdots\circ\bm{\beta}_{r,N},
\end{align}
where $\omega_r\in(0,\infty)$ and $\lVert\bm{\beta}_{r,d}\rVert_2=1$, for $1\leq r\leq R$ and $1\leq d\leq  N$, with  $R$ being a positive integer.
The rank of a tensor $\cm{B}$ is defined as the
smallest number of rank-one tensors that generate $\cm{B}$ as their sum. In other words, this is the smallest $R$ such that the CP decomposition in \eqref{eq:CP} holds for $\cm{B}$. A sufficient condition for the uniqueness of the CP decomposition for a general $n$th-order tensor $\cm{B}$, up to a permutation, is 
$\sum_{d=1}^{N} k_{\bm{B}^{(d)}} \geq 2R + (N-1)$, where  $k_{\bm{B}^{(d)}}$ denotes
the Kruskal rank of the matrix $\bm{B}^{(d)}= (\bm{\beta}_{1, d},\ldots, \bm{\beta}_{R, d})$; see \cite{Sidiropoulos2000}. Here the Kruskal rank of a matrix is defined as the maximum value $k$ such that any $k$ columns of the matrix are linearly independent \citep{kruskal1976more}.

\begin{remark}\label{remark:CPunique}
The sufficient condition for the uniqueness of the CP decomposition, up to a permutation, is a typically mild requirement in practice \citep{ten2002uniqueness, sun2017store, Zeng2021}: when $R$ is fixed, it implies that the identifiability is guaranteed when the components $\bm{\beta}_{r,d}$ are not highly dependent across $r$, and it is reasonable to assume that this condition is fulfilled in most empirical applications.
\end{remark}

\subsection{Model formulation}\label{subsec:model}

We consider the regression of an $m$th-order tensor response series $\cm{Y}_t\in\mathbb{R}^{p_1\times\cdots\times p_m}$ on an  $n$th-order tensor predictor series $\cm{X}_t\in\mathbb{R}^{q_{1}\times\cdots\times q_n}$ as follows:
\begin{align}\label{eq:model}
	\cm{Y}_t = \langle \cm{B},\cm{X}_{t}\rangle+\cm{E}_t, \quad 1\leq t \leq T,
\end{align}
where $m,n\geq1$,  $\cm{E}_t\in\mathbb{R}^{p_1\times\cdots\times p_m}$ is the innovation at time $t$, and $\cm{B}\in\mathbb{R}^{p_1\times\cdots\times p_{m}\times q_{1}\times\cdots\times q_n}$ is the $N$th-order coefficient tensor, with $N=m+n$. Here, without loss of generality, by centering both series $\{\cm{Y}_t\}$ and $\{\cm{X}_t\}$, we set the intercept to zero. 
Let $p_y=\prod_{d=1}^m p_d$ and $q_x=\prod_{d=1}^n q_d$. 
Note that even when each individual dimension $p_d$ or $q_d$ is relatively small, the ambient dimension of $\cm{B}$, i.e., $p_y q_x$,  can be very large  due to a moderately large $N$. For example,   $\cm{B}$ is a sixth-order tensor with more than $10^6$ entries for the air pollution data in Section \ref{subsec:taiwan}.
Thus, efficient dimension reduction is required for model \eqref{eq:model}.

To reduce the dimensionality of $\cm{B}$ while achieving the interpretability of the tensor dynamic structure, we assume that $\cm{B}$  admits the  rank-$R$ CP decomposition in \eqref{eq:CP} with $N=m+n$. This assumption reduces the number of parameters from $p_y q_x$ to $R(1+\sum_{d=1}^m p_d + \sum_{d=1}^n q_d)$.
Moreover, for an even more substantial dimension reduction, we can further impose entrywise sparsity assumptions  $\|\bm{\beta}_{r,d}\rVert_0\leq s_d$, where $s_d$ represents the sparsity level for the $d$th mode, for $1\leq d\leq N$.

Combining \eqref{eq:CP} and \eqref{eq:model}, the model can be written as 
\begin{equation}\label{eq:generalmodel}
\cm{Y}_t = \sum_{r=1}^R\omega_r\bm{\beta}_{r,1}\circ\cdots\circ\bm{\beta}_{r,m}\big( \underbrace{\cm{X}_t\times_{d=1}^{n}\bm{\beta}_{r,m+d}}_{f_{r,t}} \big) +\cm{E}_t,
\end{equation}
where  $f_{r,t} = \cm{X}_t\times_{d=1}^{n}\bm{\beta}_{r,m+d}$  for $1\leq r\leq R$ are univariate factor processes summarizing the information in the predictor series. Model \eqref{eq:generalmodel} shows that we can separate $\{\bm{\beta}_{r,d}\}_{d=1}^{N}$ into two parts, $\{\bm{\beta}_{r,d}\}_{d=1}^{m}$ and $\{\bm{\beta}_{r,m+d}\}_{d=1}^{n}$, which we call the \textit{response} and \textit{predictor} loading vectors, respectively. The  latter encodes  the high-dimensional predictor  $\cm{X}_t$ into univariate factors $f_{r,t}$ for $1\leq r\leq R$. The \textit{predictor} loading vector $\bm{\beta}_{r,m+d}$, which consists of relative weights across variables  along the $d$th mode of the predictor, can vary across $r$. Then the information in each factor $f_{r,t}$ is  disseminated across $m$ modes of the response  $\cm{Y}_t$, where the relative weights across variables in the $d$th  mode of the response  are given by the \textit{response} loading vector $\bm{\beta}_{r,d}$ for $1\leq d\leq m$.  In addition,  $\omega_r$ controls the relative importance of different factors. 

As discussed above, for every $r$, the loading vectors $\bm{\beta}_{r,1}, \dots, \bm{\beta}_{r,N}$  provide insights into each  mode of the response and predictor tensors individually. More interestingly, CP decomposition enables us to interpret the interactive roles of any two distinct modes $1\leq d_1\neq d_2\leq N$  via the loading matrix $\bm{\beta}_{r,d_1}\circ \bm{\beta}_{r,d_2}$. There are three types of interactive relationships implied by \eqref{eq:generalmodel}:
\begin{itemize}
\item [(i)] response-response ($1\leq d_1\neq d_2 \leq m$):  Since the outer product is permutation-invariant, the full tensor $\bm{\beta}_{r,1} \circ \cdots \circ \bm{\beta}_{r,m}$ can be viewed as an extension of $\bm{\beta}_{r,d_1} \circ \bm{\beta}_{r,d_2}$ along the remaining $m-2$ modes.

\item [(ii)] predictor-predictor ($d_1=m+d^\prime$ and $d_2=m+d^{\prime\prime}$, with $1\leq d^{\prime}\neq d^{\prime\prime}\leq n$): We can separate out $\bm{\beta}_{r,d_1}$ and $\bm{\beta}_{r,d_2}$ by rewriting $f_{r,t} = \langle \bm{\beta}_{r,d_1} \circ \bm{\beta}_{r,d_2}, \cm{X}_t\times_{d=1, d\notin \{d_1-m, d_2-m\}}^{n}\bm{\beta}_{r,m+d} \rangle$. 

\item [(iii)] response-predictor ($1\leq d_1\leq m$ and $d_2=m+d^{\prime\prime}$, with $1\leq d^{\prime\prime}\leq n$): It can be shown that, up to a permutation of modes, $\bm{\beta}_{r,1}\circ\cdots\circ\bm{\beta}_{r,m} f_{r,t}$ is equivalent to 
\begin{equation}\label{eq:interact}
	\bm{\beta}_{r,1}\circ \cdots \circ\bm{\beta}_{r,d_1-1}\circ \bm{\beta}_{r,d_1+1}\circ \cdots\circ\bm{\beta}_{r,m}  \circ 
	\langle \bm{\beta}_{r,d_1} \circ \bm{\beta}_{r,d_2}, 
	\cm{X}_t\times_{d=1, d\neq d_2-m}^{n}\bm{\beta}_{r,m+d}
	\rangle.
\end{equation}
\end{itemize}
Thus, in all three cases, the loading matrix $\bm{\beta}_{r,d_1} \circ \bm{\beta}_{r,d_2}$ fully encapsulates the effects associated with two specific modes, which may belong to the response or predictor tensor.

Interestingly, model \eqref{eq:generalmodel} offers a simple data-driven framework to accommodate  high-dimensional response and predictor series with mismatched sampling frequencies. 
In Section \ref{subsec:macro}, we demonstrate a macroeconomic analysis based on the model  $\bm{y}_t = \langle \cm{B},\bm{X}_{t}\rangle+\bm{\varepsilon}_t$. Here, $\bm{y}_t\in\mathbb{R}^{179}$ is a low-frequency response, containing 179 economic variables observed each quarter $t$. The predictor $\bm{X}_{t}\in\mathbb{R}^{112\times 3}$ is high-frequency,  containing 112 variables, with three monthly observations available  in  each quarter $t$. Based on the decomposition $\cm{B} = \sum_{r=1}^R\omega_r\bm{\beta}_{r,1}\circ\bm{\beta}_{r,2}\circ\bm{\beta}_{r,3}$,  group patterns within quarterly  and monthly variables can be identified  from $\bm{\beta}_{r,1}\in\mathbb{R}^{179}$ and $\bm{\beta}_{r,2}\in\mathbb{R}^{112}$, respectively. More interestingly,  we can understand the response-predictor (i.e., output-input) relationship by visualizing the heatmap of the matrix $\bm{\beta}_{r,1}\circ \bm{\beta}_{r,2}$, despite the different sampling frequencies of predictors and responses. Furthermore, $\bm{\beta}_{r,3}\in\mathbb{R}^{3}$ captures the  temporal aggregation pattern, i.e., how the high-frequency signals are summarized into the low-frequency scale. Importantly, since $\bm{\beta}_{r,3}$ is an unknown parameter, it allows the data to determine the aggregation weights rather than relying on predetermined ones; see Remark \ref{remark:midas} in section for more discussions. In short, as a useful by-product, our CP decomposition-based approach offers a simple, data-driven tool for disentangling  variables and sampling frequencies in mixed-frequency time series modeling. 

In fact, it is straightforward to further incorporate lags of predictor variables. For instance, to include data from the $L>1$ most recent quarters up to quarter $t$, we can simply form a tensor $\cm{X}_{t}\in\mathbb{R}^{112\times 3\times L}$. Moreover, it is noteworthy that our approach is equally applicable in settings with high-frequency responses and low-frequency predictors. In such cases, the  frequency difference will be treated as a mode of the response tensor $\cm{Y}_t$, while the corresponding loading vector $\bm{\beta}_{r,d}$ indicates how the low-frequency information is exploited at the high-frequency scale.

\begin{remark}[Supervised factor modeling]\label{remark:tfm}
	While we refer to $\{f_{r,t}\}$ in \eqref{eq:generalmodel} as factor processes, they are different from conventional factors in unsupervised settings \citep{Bai2016, wang2019factor, Han2024, Babii2024}, where the goal is to identify latent drivers of co-movement in high-dimensional time series. In contrast, our factors are extracted from the observed predictors $\cm{X}_t$ and are supervised, i.e., chosen to be most predictive of the response $\cm{Y}_t$. This supervised nature arises directly from the regression framework, where $\cm{X}_t$ serves as input and $\cm{Y}_t$ as output. Therefore, although \eqref{eq:generalmodel} resembles tensor factor models in \cite{Han2024, Babii2024} in form, the meaning of $\{f_{r,t}\}$ differs. For instance, \cite{Han2024} assume uncorrelated factors, and \cite{Babii2024} impose orthogonality assumptions. These assumptions are not applicable to our model due to the different nature of those factors.
\end{remark}

\begin{remark}[Connection with MIDAS]\label{remark:midas}
There is extensive literature on mixed-frequency regression, particularly Mixed Data Sampling (MIDAS) models \citep{midas2006, Foroni2013, Babii2022}, which typically forecast a univariate low-frequency variable using high-frequency predictors. 
However, the case involving high-dimensional responses and predictors remains largely underexplored. While a naive extension is to apply univariate MIDAS entrywise,  it would inevitably ignore potential low-dimensional structures in the response space. Traditional MIDAS methods often impose  parametric lag structures, which can be restrictive when applied uniformly across responses. While U-MIDAS variants \citep{Foroni2015} avoid these structures, it becomes overparameterized in high dimensions. 
Our method can be viewed as a structured U-MIDAS approach using CP decomposition: it tensorizes the coefficient array across frequency, predictors, and responses, enabling flexible modeling of lag effects and response interactions. The CP low-rank structure ensures both scalability and interpretability in high-dimensional mixed-frequency settings. Moreover, our approach applies to high-frequency responses and low-frequency predictors.
\end{remark}

\subsection{Tensor autoregression}\label{subsec:tar}
This section introduces a new tensor autoregressive (TAR) model as an important special case of model \eqref{eq:model}.

In the simple case where $\cm{X}_t = \cm{Y}_{t-1}$, the TAR(1) model is formed, where the coefficient tensor $\cm{B}\in\mathbb{R}^{p_1\times\cdots\times p_m\times p_1\times\cdots\times p_m}$ has the CP decomposition in \eqref{eq:CP}, with $N=2m=2n$, and $(q_1,\dots, q_{n})=(p_1,\dots, p_m)$.  The interpretations regarding \eqref{eq:generalmodel}, particularly those regarding individual predictor or response loading vectors and various loading matrices  $\bm{\beta}_{r,d_1} \circ \bm{\beta}_{r,d_2}$, apply immediately.
Consider the Taiwan air pollution time series data  in  Section \ref{subsec:taiwan} as a concrete example. For  each year $t$, the observations can be arranged into a third-order tensor  $\cm{Y}_t\in\mathbb{R}^{12\times7\times 12}$, for  12 monitoring stations, 7 air pollutants, and 12 months. Then $\cm{B}$ is a tensor with $N=6$ modes satisfying \eqref{eq:CP}.  In this case, for each $1\leq r\leq R$, the \textit{predictor} loading vectors  $\{\bm{\beta}_{r,d}\}_{d=4}^6$ capture important signals   from the previous year across various monitoring stations, air pollutants, and months, which are  predictive of the following year. Analogously, the \textit{response} loading vectors  $\{\bm{\beta}_{r,d}\}_{d=1}^3$ summarize how the past influence manifests in the following year across   stations,  pollutants, and months. Moreover, we can interpret interactions between any two modes. For example,  $\bm{\beta}_{r,4} \circ \bm{\beta}_{r,5}$ reveals the station-pollutant interaction in  signals from the past year, while $\bm{\beta}_{r,1} \circ \bm{\beta}_{r,2}$ captures this interaction in terms of the response that manifests in the following year. 
In addition, as demonstrated in \eqref{eq:interact}, $\bm{\beta}_{r,1} \circ \bm{\beta}_{r,4}$ summarizes dynamic interactions across stations from one year to the next; see Section \ref{subsec:taiwan} for details.

More generally, for an $m$th-order tensor time series $\cm{Y}_t\in\mathbb{R}^{p_1\times\cdots\times p_m}$, we can take $\cm{X}_t\in\mathbb{R}^{p_1\times\cdots\times p_m\times L}$ with $\cm{X}_{\cdots \ell}=\cm{Y}_{t-\ell}$ for $1\leq \ell\leq L$, given that $L>1$; i.e., $n=m+1$ and $(q_1,\dots, q_{n})=(p_1,\dots, p_m, L)$. Thus, if $L>1$, we can define the TAR($L$) model,
\begin{align}\label{eq:TensorAR}
	\cm{Y}_t=\sum_{\ell=1}^L\langle\cm{B}_\ell, \cm{Y}_{t-\ell} \rangle+\cm{E}_t,
\end{align}
where $\cm{B}\in\mathbb{R}^{p_1\times\cdots\times p_m\times p_1\times\cdots\times p_m \times L}$ satisfies \eqref{eq:CP} with $N=2m+1$, and $\cm{B}_\ell = \cm{B}_{\cdots \ell} \in\mathbb{R}^{p_1\times\cdots\times p_m\times p_1\times\cdots\times p_m}$. 
The interpretation of \eqref{eq:generalmodel} can be generalized to  model \eqref{eq:TensorAR}, where the factor process  $f_{r,t} = \cm{X}_t\times_{d=1}^{m+1}\bm{\beta}_{r,m+d}$  in  \eqref{eq:generalmodel}
 captures the $r$th input signals of the dynamic system across not only   $m$ different cross-sectional modes, but also the lag mode. That is, the  entries of the loading vector $\bm{\beta}_{r,2m+1} \in\mathbb{R}^{L}$ are relative weights across $L$ lags. They control how the strength of the $r$th cross-sectional dynamic dependence pattern, as characterized by $\bm{\beta}_{r,d}$ for $1\leq d\leq 2m$, varies over time lags.
 
 \begin{remark}[Stationarity condition]\label{remark:stationary}
Let $\bm{y}_t=\vect(\cm{Y}_t)\in\mathbb{R}^{p_y}$ be the vectorization of $\cm{Y}_t$, and similarly $\bm{\varepsilon}_t=\vect(\cm{E}_t)$. 
For $1\leq \ell \leq L$, let $\bm{B}_\ell$ denote the $p_y\times p_y$ matricization of $\cm{B}_\ell$ obtained by combining its first $m$ modes to rows and the other $m$ modes to columns. Then, \eqref{eq:TensorAR} can be written in the vector form as
$\bm{y}_t=\sum_{\ell=1}^L\bm{B}_{(\ell)} \bm{y}_{t-\ell}+\bm{\varepsilon}_t$. By \eqref{eq:CP} it can be shown that $\bm{B}_{(\ell)}= \sum_{r=1}^R\omega_r  \bm{\beta}_{r,m} \otimes \cdots \otimes \bm{\beta}_{r,1} \beta_{r, 2m+1}^{(\ell)} (\bm{\beta}_{r,2m}\otimes\cdots\otimes \bm{\beta}_{r,m+1})^\top$, where $\beta_{r, 2m+1}^{(\ell)}$ is the $\ell$th entry of $\bm{\beta}_{r,2m+1}\in\mathbb{R}^{L}$ for $1\leq \ell\leq L$, given that $L>1$. Thus, by the stationarity condition of the vector autoregressive model \citep{Luetkepohl2005}, the TAR($L$) model in \eqref{eq:TensorAR} is strictly stationarity if and only if $\rho(\bm{\underline{B}})<1$, where $\rho(\bm{\underline{B}})$ is the spectral radius of
\[
\bm{\underline{B}}
=\begin{pmatrix}
\sum_{r=1}^R\omega_r  \bm{\beta}_{r,m} \otimes \cdots \otimes \bm{\beta}_{r,1}  (\bm{\beta}_{r,2m+1}\otimes \bm{\beta}_{r,2m}\otimes\cdots\otimes \bm{\beta}_{r,m+1})^\top\\
	\bm{I}_{(L-1)p_y} \quad	\bm{0}_{(L-1)p_y\times p_y}
\end{pmatrix} \in\mathbb{R}^{ L p_y\times L p_y},
\]
with $\bm{I}_{(L-1)p_y}$ being the identity matrix of size $(L-1)p_y$. Note that the top block of $\bm{\underline{B}}$ is $\sum_{r=1}^R\omega_r  \bm{\beta}_{r,m} \otimes \cdots \otimes \bm{\beta}_{r,1}  (\bm{\beta}_{r,2m+1}\otimes \bm{\beta}_{r,2m}\otimes\cdots\otimes \bm{\beta}_{r,m+1})^\top=(\bm{B}_{(1)}, \dots, \bm{B}_{(L)})$. In the special case where $L=1$, the vector $\bm{\beta}_{r,2m+1}$ is simply omitted as the $(2m+1)$-th mode of $\cm{B}$ is collapsed.
\end{remark}

\begin{remark}[Special case with $m=1$]\label{remark:var}
When $m=1$, $\cm{Y}_t = \vect(\cm{Y}_t) =\bm{y}_t \in\mathbb{R}^{p_y}$ is a vector. In this case,  \eqref{eq:TensorAR} reduces to a vector-valued time series model, $\bm{y}_t=\sum_{\ell=1}^L\bm{B}_{\ell} \bm{y}_{t-\ell}+\bm{\varepsilon}_t$, with  coefficient matrices $\bm{B}_\ell=\sum_{r=1}^R\omega_r\bm{\beta}_{r,1}\bm{\beta}_{r,2}^\top\beta_{r,3}^{(\ell)}$; i.e., in \eqref{eq:TensorAR},  $\cm{X}_t=\bm{X}_t=(\bm{y}_{t-1},\ldots,\bm{y}_{t-L})$ is a $p_y\times L$ matrix, and $\cm{B}$ is a $p_y\times p_y\times L$ tensor that satisfies \eqref{eq:CP} with $N=3$ and is formed by stacking $\bm{B}_{\ell}$ for $1\leq \ell\leq L$. Thus, a   by-product of the proposed TAR model in the case of $m=1$ is a new high-dimensional vector autoregressive (VAR) model based on the CP decomposition. Note that this is different from the vector form of the TAR model for a general $m$-th order tensor time series $\cm{Y}_t$ for $m>1$; see Remark  \ref{remark:stationary}.  
\end{remark}

\section{High-dimensional estimation}\label{sec:est}
\subsection{CP low-rank estimator}\label{subsec:est1}
Suppose that an observed sequence $\{(\cm{Y}_t, \cm{X}_t)\}_{t=1}^T$ of length $T$ follows the tensor stochastic regression in  \eqref{eq:generalmodel}.
 
While the ambient dimension  of $\cm{B}$ is $p_yq_x$, where $p_y=\prod_{d=1}^m p_d$ and $q_x= \prod_{d=1}^n q_d$, by exploiting its CP low-rank structure, the required sample size for its consistent estimation can be reduced to  $T \gg \sum_{d=1}^m p_d +  \sum_{d=1}^n q_d$. Specifically,  we propose the CP low-rank estimator for $\cm{B}$ as follows:
\[
\widetilde{\cm{B}} =[\![\widetilde{\bm{\omega}};\widetilde{\bm{B}}^{(1)},\dots, \widetilde{\bm{B}}^{(N)}]\!] 
= \underset{\bm{\omega}, \bm{B}^{(1)}, \dots, \bm{B}^{(N)}}{ \argmin }
\sum_{t=1}^T\Big\lVert\cm{Y}_t-  \sum_{r=1}^R\omega_r\bm{\beta}_{r,1}\circ\cdots\circ\bm{\beta}_{r,m}\big( \cm{X}_t\times_{d=1}^{n}\bm{\beta}_{r,m+d}  \big)  \Big\rVert_\Fr^2
\]
subject to $\lVert\bm{\beta}_{r,d}\rVert_2=1$ for  $1\leq r\leq R$ and $1\leq d\leq N$, where $\bm{\omega}=(\omega_r)_{1\leq r\leq R}$, and  $\bm{B}^{(d)}=(\bm{\beta}_{1, d},\ldots, \bm{\beta}_{R, d})$.  

Denote by $\cm{B}^*$ the true value of $\cm{B}$. Let  $\bm{y}_t=\vect(\cm{Y}_t)  \in\mathbb{R}^{p_y}$ and $\bm{x}_t=\vect(\cm{X}_t) \in\mathbb{R}^{q_x}$.  Without loss of generality, we assume that $\mathbb{E}(\bm{x}_t)=\bm{0}$, i.e., the data are centered. We also make the following assumptions.

\begin{assumption}[Identifiability]\label{assum:unique}
	The coefficient tensor $\cm{B}^*\in\mathbb{R}^{p_1\times\cdots\times p_m\times q_1\times\cdots\times q_n}$ has a unique CP decomposition $\cm{B}^* =\sum_{r=1}^R\omega_r^*\bm{\beta}_{r,1}^*\circ\cdots\circ\bm{\beta}_{r,N}^*$, up to a permutation, where  $\omega_r^*\in(0,\infty)$ and $\|\bm{\beta}_{r,d}^*\|_2=1$  for $1\leq r\leq R$ and $1\leq d\leq N$.
\end{assumption}

\begin{assumption}[Innovation process]\label{assum:error}
	(i) Let $\bm{\varepsilon}_t=\vect(\cm{E}_t)=\bm{\Sigma}_{\varepsilon}^{1/2}\bm{\xi}_t$, where $\bm{\Sigma}_{\varepsilon}$ is a positive definite covariance matrix, $\{\bm{\xi}_t\}$ is a sequence of $i.i.d.$ random vectors with zero mean and $\var(\bm{\xi}_t)=\bm{I}_{p_y}$, and the coordinates $\xi_{it}$ are mutually independent and $\sigma^2$-sub-Gaussian. (ii) Moreover, $\bm{\xi}_t$ is independent of the $\sigma$-field $\cm{F}_t:=\sigma\{\bm{x}_t, \bm{\xi}_{t-1}, \bm{x}_{t-1}, \bm{\xi}_{t-2}, \dots\}$ for all $t\in\mathbb{Z}$.
\end{assumption}

\begin{assumption}[Response and predictor processes]\label{assum:stationary}
	(i) The process $\{(\cm{Y}_t, \cm{X}_t), t\in\mathbb{Z}\}$ is stationary. (ii) The process $\{\cm{X}_t\}$ has an infinite-order moving average form, $\bm{x}_t =\vect(\cm{X}_t)=\sum_{j=1}^{\infty}\bm{\Phi}_j^x \bm{\epsilon}_{t-j}$, where $\bm{\epsilon}_t\in\mathbb{R}^{q_\epsilon}$,  $\bm{\Phi}_j^x\in\mathbb{R}^{q_x\times q_\epsilon}$, $\sum_{j=1}^{\infty}\|\bm{\Phi}_j^x\|_\Fr <\infty$, and $q_\epsilon\leq q_x$. Moreover, $\{\bm{\epsilon}_t\}$ is a sequence of $i.i.d.$ random vectors with zero mean and $\var(\bm{\epsilon}_t)=\bm{I}_{q_\epsilon}$, and the  coordinates $\epsilon_{it}$ are mutually independent and sub-Gaussian.
\end{assumption}

Assumption \ref{assum:unique} requires the uniqueness of the CP decomposition; see a sufficient condition given in Section \ref{subsec:notation} and a discussion on its mildness in Remark \ref{remark:CPunique}.  The sub-Gaussian condition in Assumption \ref{assum:error}(i) is standard in the literature \citep{wang2022high, wang2024high, zheng2025}. Assumption \ref{assum:error}(ii)  implies that the innovation  $\cm{E}_t$ is independent of $\cm{F}_t$, i.e., information up to time $t$. This condition enables us to establish martingale concentration results under the general random design $\cm{X}_t$.  The stationarity of the time series, i.e., Assumption \ref{assum:stationary}(i), is standard. The infinite-order moving average form in Assumption \ref{assum:stationary}(ii) allows  $\{\bm{x}_t\}$ to be a general linear process. The latter is used to establish the restricted strong convexity of the loss function; see the proofs of main theorems  in the supplementary file.

Let  $\bm{\Phi}_x(z) = \sum_{j=1}^{\infty}\bm{\Phi}_j^x z^j$,  for $z\in\mathbb{C}$.
Denote 
$\kappa_1 = \min_{|z|=1}\lambda_{\min}(\bm{\Phi}_x^{\HH}(z)\bm{\Phi}_x(z))$ and $\kappa_2 = \max_{|z|=1}\lambda_{\max}(\bm{\Phi}_x^{\HH}(z)\bm{\Phi}_x(z))$,
where   $\bm{\Phi}_x^{\HH}(z)$ is the conjugate transpose of $\bm{\Phi}_x(z)$. In particular, Assumption \ref{assum:stationary}(ii) implies $\kappa_1\leq \lambda_{\min}(\bm{\Sigma}_x) \leq \lambda_{\max}(\bm{\Sigma}_x) \leq \kappa_2$, where $\bm{\Sigma}_x = \mathbb{E}(\bm{x}_t\bm{x}_t^\top)$ is the coveriance matrix of $\bm{x}_t$; see more details in Lemma \ref{lemma:eigen} in the supplementary file.  Note that $\kappa_1$ and $\kappa_2$ are allowed to vary with the dimension. 

The following theorem presents the estimation and prediction error bounds for the CP low-rank estimator.

\begin{theorem}[CP low-rank estimator]\label{thm:nonsparsetr}
	Suppose that the observed  response and predictor series $\{(\cm{Y}_t, \cm{X}_t)\}_{t=1}^T$ follow model \eqref{eq:generalmodel}. Under Assumptions \ref{assum:unique}--\ref{assum:stationary}, if $\kappa_1>0$ and  $T \gtrsim (\kappa_2/\kappa_1)^2  R(\sum_{d=1}^m p_d +\sum_{d=1}^{n} q_d) \log N$, then
	\[
	\|\widetilde{\cm{B}} - \cm{B}^*\|_\Fr \lesssim   \sqrt{\frac{\sigma^2  \kappa_2 \lambda_{\max}(\bm{\Sigma}_{\varepsilon}) R(\sum_{d=1}^m p_d +\sum_{d=1}^{n} q_d) \log N }{\kappa_1^2 T} },\quad\text{and}
	\]
	\[
	\frac{1}{T}\sum_{t=1}^T\lVert\langle \widetilde{\cm{B}} - \cm{B}^*, \cm{X}_{t}\rangle\rVert_\Fr^2
	\lesssim  \frac{\sigma^2  \kappa_2 \lambda_{\max}(\bm{\Sigma}_{\varepsilon}) R(\sum_{d=1}^m p_d +\sum_{d=1}^{n} q_d) \log N}{\kappa_1 T} 
	\] 
	with probability at least $1-5  e^{-2 R(\sum_{d=1}^N p_d+1) \log\{6(N+1)\}}$.
\end{theorem}

By Theorem \ref{thm:nonsparsetr}, when the rank $R$ and the constants $\kappa_1, \kappa_2$ and $\lambda_{\max}(\bm{\Sigma}_{\varepsilon})$ are fixed, the estimation error rate is $\|\widetilde{\cm{B}} - \cm{B}^*\|_\Fr=  O_p(\sqrt{(\sum_{d=1}^m p_d +  \sum_{d=1}^n q_d)\log N/T})$, and the  prediction error rate is $T^{-1}\sum_{t=1}^T\lVert\langle \widetilde{\cm{B}} - \cm{B}^*, \cm{X}_{t}\rangle\rVert_\Fr^2= O_p((\sum_{d=1}^m p_d +  \sum_{d=1}^n q_d) \log N/T)$. The required sample size  $T$ is at least $O(\sum_{d=1}^m p_d +  \sum_{d=1}^n q_d)$, a substantial reduction from $O(p_y q_x)$.

\subsection{Sparse CP low-rank estimator}\label{subsec:est2}

In the ultra-high-dimensional setup where the ambient dimension of $\cm{B}$ can grow at an exponential rate of the sample size, i.e., $T\gtrsim \log(p_y q_x)$, we propose the sparse CP low-rank estimator as follows,
\begin{align*}
\widehat{\cm{B}} =[\![\widehat{\bm{\omega}};\widehat{\bm{B}}^{(1)},\dots, \widehat{\bm{B}}^{(N)}]\!] 
= \underset{\bm{\omega}, \bm{B}^{(1)}, \dots, \bm{B}^{(N)}}{ \argmin } &
\sum_{t=1}^T\Big\lVert\cm{Y}_t-  \sum_{r=1}^R\omega_r\bm{\beta}_{r,1}\circ\cdots\circ\bm{\beta}_{r,m}\big( \cm{X}_t\times_{d=1}^{n}\bm{\beta}_{r,m+d}  \big)  \Big\rVert_\Fr^2
\end{align*}
subject to $\lVert\bm{\beta}_{r,d}\rVert_2=1$ and $\lVert\bm{\beta}_{r,d}\rVert_0\leq s_d$ for $1\leq r\leq R$ and $1\leq d\leq N$,  where  $s_d$ is a predetermined positive integer that bounds the sparsity level of $\bm{\beta}_{r,d}$ for each mode $d$. The $\ell_0$-constraints correspond to the hard-thresholding method for enforcing sparsity. It is favorable than the Lasso-type penalized approach as the latter always leads to a biased estimator although it is sparse. In fact, the Lasso or soft-thresholding method is typically preferred for ensuring the convexity of the loss function. However, since the  above loss function  is nonconvex due to the low-rankness, it is not necessary to insist on
using the Lasso-type penalty. In addition,  the hard-thresholding method often provides greater numerical stability, as results tend to be more sensitive to the Lasso regularization parameter than to the sparsity upper bounds; see \cite{sun2017store, Huang2025} and the references therein.

\begin{assumption}[Sparsity]\label{assum:sparse}
For $1\leq r\leq R$ and $1\leq d\leq N$, $\|\bm{\beta}_{r,d}^*\|_0\leq s_d$, where the sparsity levels $s_1,\dots, s_N$ are less than $p_1,\dots, p_m, q_1,\dots, q_n$, respectively.
\end{assumption}

If $(s_1,\dots, s_N)\geq (p_1,\dots, p_m, q_1,\dots, q_n)$ holds elementwise, the sparsity constraints will become inactive, and $\widehat{\cm{B}}$ will reduce to $\widetilde{\cm{B}}$  in Section \ref{subsec:est1}. To avoid uninformative results, we impose the upper bounds on $s_1,\dots, s_N$ in Assumption \ref{assum:sparse}.

Similar to Theorem \ref{thm:nonsparsetr}, we provide nonasymptotic results for the sparse CP low-rank estimator as follows.

\begin{theorem}[Sparse CP low-rank estimator]\label{thm:sparsetr}
Suppose that the observed response and predictor series $\{(\cm{Y}_t, \cm{X}_t)\}_{t=1}^T$ follow model \eqref{eq:generalmodel}.
Under Assumptions \ref{assum:unique}--\ref{assum:sparse}, if $\kappa_1>0$ and $T\gtrsim (\kappa_2/\kappa_1)^2 Rs  \min\{\log(p_y q_x), \log[21e p_y q_x/(2Rs)]\}$, then
\[
\|\widehat{\cm{B}} - \cm{B}^*\|_\Fr \lesssim   \sqrt{\frac{\sigma^2  \kappa_2 \lambda_{\max}(\bm{\Sigma}_{\varepsilon}) Rs \min\{\log(p_yq_x), \log[ep_yq_x/(2Rs)]\} }{\kappa_1^2 T} }, \quad\text{and}
\]
\[
\frac{1}{T}\sum_{t=1}^T\lVert\langle \widehat{\cm{B}} - \cm{B}^*, \cm{X}_{t}\rangle\rVert_\Fr^2
\lesssim  \frac{\sigma^2  \kappa_2 \lambda_{\max}(\bm{\Sigma}_{\varepsilon}) Rs \min\{\log(p_yq_x), \log[ep_yq_x/(2Rs)]\} }{\kappa_1 T} 
\] 
with probability at least $1-5e^{- 2Rs\min\{\log(p_y q_x), \log[e p_y q_x/(2Rs)]\}}$, where $s=\prod_{d=1}^{N}s_d$.
\end{theorem}

When   $R, \kappa_1, \kappa_2$, and $\lambda_{\max}(\bm{\Sigma}_{\varepsilon})$ are fixed, Theorem \ref{thm:sparsetr} implies the estimation and prediction error rates $\|\widehat{\cm{B}} - \cm{B}^*\|_\Fr=  O_p(\sqrt{s \log(p_y q_x)/T})$ and $T^{-1}\sum_{t=1}^T\lVert\langle \widehat{\cm{B}} - \cm{B}^*, \cm{X}_{t}\rangle\rVert_\Fr^2= O_p(s\log(p_y q_x)/T)$. They are much sharper than the results for the non-sparse estimator in Theorem \ref{thm:nonsparsetr},  as $\log(p_y q_x) = \sum_{d=1}^n \log(q_d) + \sum_{d=1}^m \log(p_d)$ is much smaller than $\sum_{d=1}^m p_d  + \sum_{d=1}^n q_d$.

\begin{remark}[Sparse estimation without low-rank structures]\label{remark:lasso}
Due to the tensor structure of the data, the primary goal of our model and estimators is to capture and interpret the multi-dimensional structures across various modes of the response and predictors. This requires leveraging the tensor factorization and hence low-rank structures. An alternative dimension reduction approach is to simply vectorize the response and predictor tensors, and perform the sparse regression  $\vect(\cm{Y}_t) = \bm{B} \vect(\cm{X}_t)+\vect(\cm{E}_t)$, where $\bm{B}$ is the $p_y\times q_x$ matricization of the tensor $\cm{B}$ from \eqref{eq:model}, obtained by unfolding its first $m$ modes into rows and the last $n$ modes into columns. A Lasso estimator can then be applied to estimate $\bm{B}$, yielding the estimation error rate $\|\widehat{\bm{B}}_{\textrm{Lasso}} - \bm{B}^*\|_\Fr =O_p(\sqrt{S\log (p_yq_x)/T})$, assuming exact entrywise sparsity, where $S=\|\vect(\bm{B}^*)\|_0$. This rate is comparable to that in Theorem \ref{thm:sparsetr}, since $Rs \asymp S$.
However, this approach does not preserve any low-rank structure in the estimated tensor $\cm{B}$, and consequently, the inherent multi-dimensional information in the data will be lost.
\end{remark}

\subsection{Application to tensor autoregression}\label{subsec:tarest}

In this section, we consider the TAR($L$) model,  where $\cm{X}_t$ comprises $\cm{Y}_{t-1}, \dots, \cm{Y}_{t-L}$,  and  analyze the proposed non-sparse and sparse CP low-rank estimators.

Some notations and assumptions can be simplified  now that  $\bm{x}_t=(\bm{y}_{t-1}^\top, \dots, \bm{y}_{t-L}^\top)^\top$. Specifically, we have $q_d=p_d$ for $1\leq d\leq m$. If $L>1$, then  $n=m+1$, $N=2m+1$, and $q_{n}=L$; otherwise, $n=m$ and $N=2m$. In both cases, $q_x=p_y L$.  Assumption \ref{assum:error}(ii) is no longer needed, since it directly follows from the independence of $\bm{\xi}_t$ required in Assumption \ref{assum:error}(i), and $\cm{F}_t=\sigma\{ \bm{\xi}_{t-1},  \bm{\xi}_{t-2}, \dots\}$. 

Moreover, Assumption \ref{assum:stationary}(i) can be replaced by the stationarity condition provided in  Remark \ref{remark:stationary}. Furthermore, the stationarity implies that  $\{\bm{y}_t\}$ has the infinite-order moving average form, $\bm{y}_t=\bm{\varepsilon}_t+\sum_{j=1}^\infty \bm{\Psi}_j^* \bm{\varepsilon}_{t-j}$, where the coefficient matrices $\bm{\Psi}_j^*\in\mathbb{R}^{p_y\times p_y}$ depend on $\cm{B}^*$. Consequently, under Assumption \ref{assum:error}(i),  we can  verify that Assumption \ref{assum:stationary}(ii)  is automatically fulfilled, where $\bm{\epsilon}_t=\bm{\xi}_t$. Therefore, the corollaries below are stated without Assumption \ref{assum:stationary}. 

In addition, due to the specific definition of $\bm{x}_t$ for the TAR($L$) case and the removal of Assumption \ref{assum:stationary}, we can replace  $\kappa_1$ and $\kappa_2$ with specific quantities related to  $\bm{\Sigma}_\varepsilon$ and $\bm{\Psi}_j^*$'s. Denote $\kappa_1^{\textrm{AR}}=	\lambda_{\min}(\bm{\Sigma}_\varepsilon)\mu_{\min}(\bm{\Psi}_*)$  and  $\kappa_2^{\textrm{AR}}=\lambda_{\max}(\bm{\Sigma}_\varepsilon)\mu_{\max}(\bm{\Psi}_*)$. Here $\mu_{\min}(\bm{\Psi}_*) = \min_{|z|=1}\lambda_{\min}(\bm{\Psi}_*(z)\bm{\Psi}_*^{\HH}(z))$ and $\mu_{\max}(\bm{\Psi}_*) = \max_{|z|=1}\lambda_{\max}(\bm{\Psi}_*(z)\bm{\Psi}_*^{\HH}(z))$, where
$\bm{\Psi}_*(z) = \bm{I}_{p_y}+\sum_{j=1}^{\infty}\bm{\Psi}_j^* z^j$, and $\bm{\Psi}_*^{\HH}(z)$ is the conjugate transpose of $\bm{\Psi}_*(z)$, for  $z\in\mathbb{C}$. By techniques similar to those in \cite{Basu2015}, we can verify that $\kappa_1^{\textrm{AR}}>0$.

Corollaries \ref{cor:nonsparsetrTAR} and \ref{cor:sparsetrTAR} provide results for the estimation of TAR($L$) models analogous to Theorems \ref{thm:nonsparsetr} and \ref{thm:sparsetr}, respectively. 

\begin{corollary}[CP low-rank TAR estimator]\label{cor:nonsparsetrTAR}
	Suppose that the observed time series $\{\cm{Y}_t\}$ follow a stationary TAR($L$) model in the form of \eqref{eq:TensorAR}. Under Assumptions \ref{assum:unique} and \ref{assum:error}(i), if  $T \gtrsim (\kappa_2^\ar L/\kappa_1^\ar)^2  R(\sum_{d=1}^m p_d +L) \log N$, then
	\[
	\|\widetilde{\cm{B}} - \cm{B}^*\|_\Fr \lesssim   \sqrt{\frac{\sigma^2  \kappa_2^\ar \lambda_{\max}(\bm{\Sigma}_{\varepsilon}) R(\sum_{d=1}^m p_d + L) \log N }{(\kappa_1^{\ar})^2 T} }, \quad\text{and}
	\]
	\[
	\frac{1}{T}\sum_{t=1}^T\lVert\langle \widetilde{\cm{B}} - \cm{B}^*, \cm{X}_{t}\rangle\rVert_\Fr^2
	\lesssim  \frac{\sigma^2  \kappa_2^\ar \lambda_{\max}(\bm{\Sigma}_{\varepsilon}) R(\sum_{d=1}^m p_d + L) \log N}{\kappa_1^\ar T} 
	\] 
	with probability at least $1-5  e^{-2 R(2\sum_{d=1}^m p_d+L) \log\{6(N+1)\}}$.
\end{corollary}

\begin{corollary}[Sparse CP low-rank TAR estimator]\label{cor:sparsetrTAR}
Suppose that the observed time series $\{\cm{Y}_t\}$ follow a stationary TAR($L$) model in the form of \eqref{eq:TensorAR}. Under Assumptions \ref{assum:unique}, \ref{assum:error}(i), and \ref{assum:sparse},  if $T\gtrsim (\kappa_2^\ar L/\kappa_1^\ar)^2 Rs  \min\{\log(p_y^2 L), \log[21e p_y^2 L/(2Rs)]\}$, then 
\[
\|\widehat{\cm{B}} - \cm{B}^*\|_\Fr \lesssim   \sqrt{\frac{\sigma^2  \kappa_2^\ar \lambda_{\max}(\bm{\Sigma}_{\varepsilon}) Rs \min\{\log(p_y^2 L), \log[ep_y^2 L/(2Rs)]\} }{(\kappa_1^\ar)^2 T} }, \quad\text{and}
\]
\[
\frac{1}{T}\sum_{t=1}^T\lVert\langle \widehat{\cm{B}} - \cm{B}^*, \cm{X}_{t}\rangle\rVert_\Fr^2
\lesssim  \frac{\sigma^2  \kappa_2^\ar \lambda_{\max}(\bm{\Sigma}_{\varepsilon}) Rs \min\{\log(p_y^2 L), \log[ep_y^2 L/(2Rs)]\} }{\kappa_1^\ar T} 
\] 
with probability at least $1-5e^{- 2Rs\min\{\log(p_y^2 L), \log[e p_y^2 L/(2Rs)]\}}$, where $s=\prod_{d=1}^{N}s_d$.
\end{corollary}

We can compare the estimation error rate in Corollary \ref{cor:nonsparsetrTAR} with that of the non-convex estimator, denoted $\widetilde{\cm{B}}_{\textrm{Tucker}}$, for the Tucker low-rank TAR model in \cite{wang2024high}, where the coefficient tensor $\cm{B}^*$ is assumed to have Tucker ranks $r_1, \dots, r_{N}$. The latter has the estimation error rate $\|\widetilde{\cm{B}}_{\textrm{Tucker}} - \cm{B}^*\|_\Fr=O_p(\sqrt{(\sum_{d=1}^m  r_d p_d + r_N L + \prod_{d=1}^{N} r_d)/T})$. By  Remark \ref{remark:CPtucker} in the supplementary file, $\cm{B}^*$ can be rewritten in a CP low-rank form with CP-rank $\prod_{d=1}^{N} r_d$. In this case, Corollary \ref{cor:nonsparsetrTAR} yields $\|\widetilde{\cm{B}} - \cm{B}^*\|_\Fr = O_p(\sqrt{\prod_{d=1}^{N} r_d(\sum_{d=1}^{m}p_d+L)\log N/T})$, which is less sharp due to the multiplicative factor $\prod_{d=1}^{N} r_d$. Thus, when $\cm{B}^*$ is Tucker low-rank, the proposed CP low-rank estimator $\widetilde{\cm{B}}$ remains consistent, albeit less efficient than the Tucker low-rank estimator. By contrast, suppose that $\cm{B}^*$ is CP low-rank, as assumed in Corollary  \ref{cor:nonsparsetrTAR}. As discussed in Remark  \ref{remark:CPtucker} in the supplementary file, it has the Tucker decomposition form with Tucker ranks $r_1=\cdots=r_N=R$. In this case, $\|\widetilde{\cm{B}}_{\textrm{Tucker}} - \cm{B}^*\|_\Fr=O_p(\sqrt{[R(\sum_{d=1}^m  p_d +  L) + R^N]/T})$, whereas
 $\|\widetilde{\cm{B}} - \cm{B}^*\|_\Fr = O_p(\sqrt{R(\sum_{d=1}^{m}p_d+L)\log N/T})$, which can be much sharper than the former as it does not involve $R^N$. Additionally, it is  worth noting that the TAR model in  \cite{wang2024high} does not allow sparse factor matrices, while $\widehat{\cm{B}}$ in Corollary \ref{cor:sparsetrTAR} accommodates such cases. Consequently, even if $\cm{B}^*$ is Tucker low-rank, as long as its factor matrices are sparse, $\widehat{\cm{B}}$  generally outperforms $\widetilde{\cm{B}}_{\textrm{Tucker}}$.

\begin{remark}[Sparse estimation in the special case with $m=1$]\label{remark:varest}
Consider the special case $m=1$ discussed in Remark \ref{remark:var}. Here, the proposed TAR model reduces to a CP low-rank counterpart of the high-dimensional VAR model for $\bm{y}_t \in \mathbb{R}^{p_y}$ in \cite{wang2022high}, which assumes the coefficient tensor $\cm{B}^* \in \mathbb{R}^{p_y \times p_y \times L}$ admits a sparse Tucker decomposition with ranks $(r_1, r_2, r_3)$ and $\ell_1$-penalized SHORR estimator achieving error rate $O_p(r\sqrt{s\log(p_y^2 L)/T})$, where $r = r_1 r_2 r_3$ and $s = s_1 s_2 s_3$. By Remark~\ref{remark:CPtucker} in the supplement, such $\cm{B}^*$ also satisfies our sparse CP low-rank assumption with rank $r$, yielding a sharper rate via Corollary \ref{cor:sparsetrTAR}:
\[
\|\widehat{\cm{B}} - \cm{B}^*\|_\Fr = O_p\left(\sqrt{r s \min\{\log(p_y^2 L), \log[21e p_y^2 L/(2rs)]\}/T}\right).
\]
Conversely, if the CP assumption holds with rank $R$, then SHORR has rate $O_p(R^3\sqrt{s\log(p_y^2 L)/T})$, while our estimator achieves a sharper rate with only $\sqrt{R}$ dependence. The improved efficiency of our method arises from the interplay between CP and Tucker structures and the use of $\ell_0$-constraints rather than $\ell_1$-penalties.
\end{remark}

\section{Simulation experiments}\label{sec:sim}
This section presents two simulation experiments, which respectively verify the consistency of the estimators in Section \ref{sec:est} and compare their estimation efficiency with the Tucker low-rank estimator in the context of TAR models. Both the proposed algorithm and an additional experiment demonstrating its computational efficiency are included in the supplementary file, in Sections \ref{sec:algo} and \ref{sec:sim3}, respectively.

The following two data generating processes (DGPs)  are considered:
\begin{itemize}
	\item DGP 1 (Tensor stochastic regression): $\cm{Y}_t=\langle\cm{B},\cm{X}_{t}\rangle +\cm{E}_t$, where $
\cm{Y}_t\in\mathbb{R}^{p_1\times p_2\times p_3}$, and  $\cm{X}_t=(\cm{X}_{i,j,k,t})\in\mathbb{R}^{q_1\times q_2\times q_3}$  is an exogenous predictor whose entries are generated independently from a stationary AR(1) process $\cm{X}_{i,j,k,t}=\alpha_{i,j,k} \cm{X}_{i,j,k,t-1}+\epsilon_{i,j,k,t}$, with $\epsilon_{i,j,k,t}\overset{i.i.d.}{\sim} N(0,1)$.
	
	\item DGP 2 (Tensor autoregression): $\cm{Y}_t=\langle\cm{B},\cm{Y}_{t-1}\rangle +\cm{E}_t$,  where $
	\cm{Y}_t\in\mathbb{R}^{p_1\times p_2\times p_3}$.
\end{itemize}
For both DGPs, $\vect(\cm{E}_t) \overset{i.i.d.}{\sim} N(\bm{0},\bm{I})$, and  we set $p_d=q_d:=p$ for $1\leq d\leq 3$. For DGP 1, the coefficients $\alpha_{i,j,k}$ are drawn from $\text{Unif}(-1,1)$. For both DGPs, we use the same coefficient tensor  $\cm{B}=\sum_{r=1}^R\omega_r\bm{\beta}_{r,1}\circ\bm{\beta}_{r,2}\circ\bm{\beta}_{r,3}\circ\bm{\beta}_{r,4}\circ\bm{\beta}_{r,5}\circ\bm{\beta}_{k,6}$, which is generated as follows for both non-sparse and sparse cases. First, all nonzero entries of $\bm{\beta}_{r,d}$'s are drawn from $\text{Unif}(-1,1)$, which are then normalized such that $\|\bm{\beta}_{r,d}\|_2=1$. Here, when $\bm{\beta}_{r,d}$'s are sparse, we let their sparsity level $s_d=\lVert \bm{\beta}_{r,d}\rVert_0:=s_0$ be identical for $1\leq d\leq 6$. Next, we set $\omega_r=0.6+ \mathbb{I}_{\{R>1\}}0.2(r-1)/(R-1)$ for $1\leq r\leq R$, where $\mathbb{I}_{\{R>1\}}=1$ if $R>1$ and zero otherwise. Then we rescale $\omega_r$'s such that $\rho(\bm{\underline{B}})=0.8$, where $\bm{\underline{B}}$ is defined as in Remark \ref{remark:stationary}. This ensures the stationarity of DGP 2. All results are reported based on averages over 500 replications for each setting.
 
\begin{figure}[!t]
	\centering
	\includegraphics[width=0.75\textwidth]{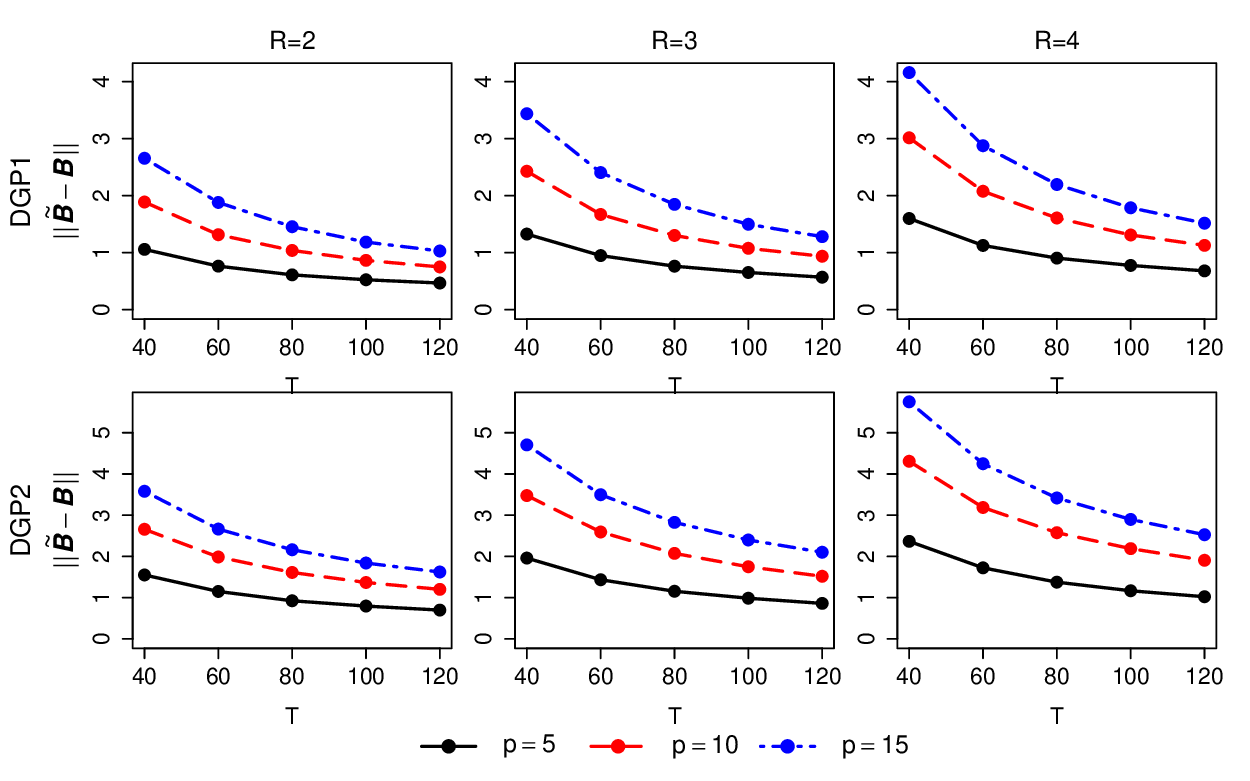}
	\caption{Estimation error $\|\widetilde{\cm{B}}-\cm{B}\|_{\Fr}$ against $T$ in the non-sparse case for different $p$ and $R$.}
	\label{Fig:sim01}
\end{figure}
\begin{figure}[!t]
	\centering
	\includegraphics[width=0.75\textwidth]{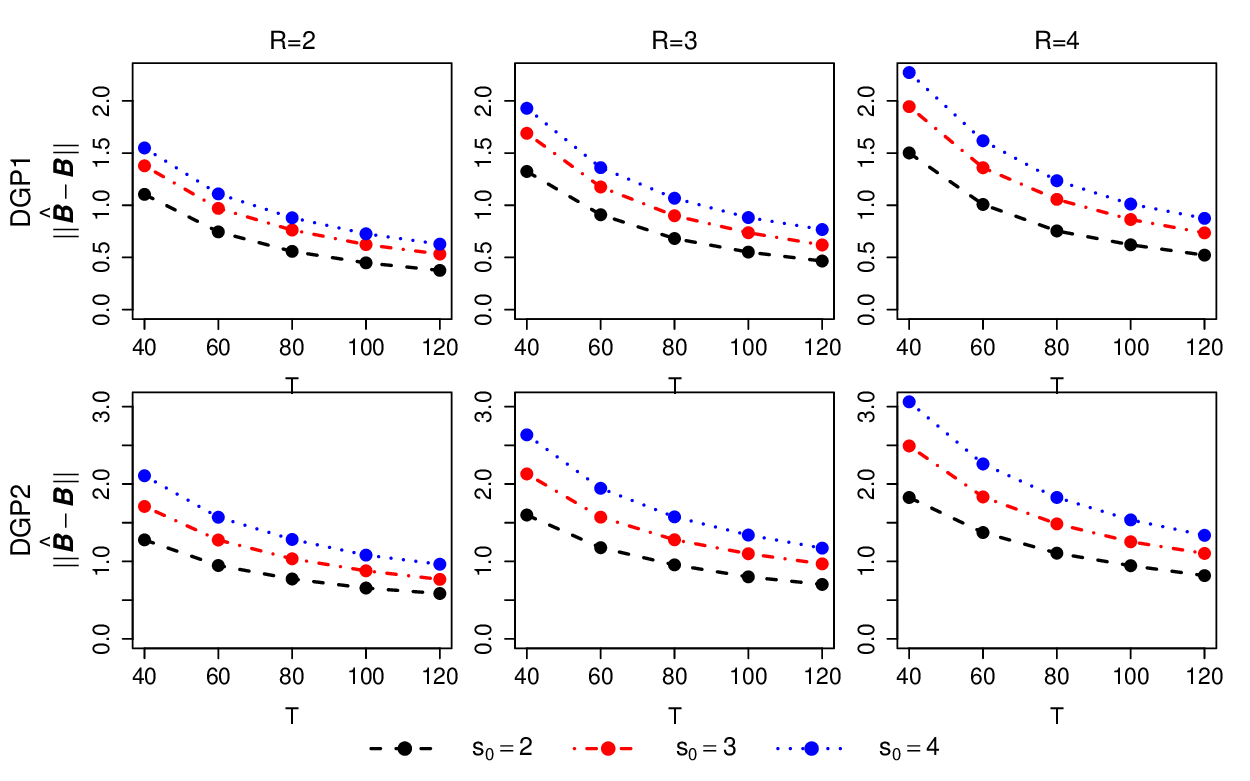}
	\caption{Estimation error $\|\widehat{\cm{B}}-\cm{B}\|_{\Fr}$ against $T$ in the sparse case for different $s_0$ and $R$.}
	\label{Fig:sim02}
\end{figure}

The first experiment aims to verify the consistency of the proposed estimators for the coefficient tensor $\cm{B}$ in both non-sparse and sparse cases. 
In the non-sparse case, we vary $p\in\{5,10,15\}$. In the sparse case, we vary $s_0\in\{2,3,4\}$ while fixing $p=10$. For all settings, we consider $T\in[40,120]$ and $R\in\{2,3,4\}$. Figure \ref{Fig:sim01} displays the estimation error  $\|\widetilde{\cm{B}}-\cm{B}\|_{\Fr}$
against $T$ for the non-sparse case. For both DGPs, we can observe that the estimation error decreases as $T$ increases, and increases as either $p$ or $R$ increases. Similarly, for  the sparse case, Figure \ref{Fig:sim02} shows that the estimation error  $\|\widehat{\cm{B}}-\cm{B}\|_{\Fr}$ decreases as $T$ increases, and increases as $s_0$ or $R$ increases. In sum, the consistency of the proposed estimators is verified in all cases,  aligning with the results presented in Theorems \ref{thm:nonsparsetr} and \ref{thm:sparsetr}. 

\begin{figure}[!t]
	\centering
	\includegraphics[width=0.8\textwidth]{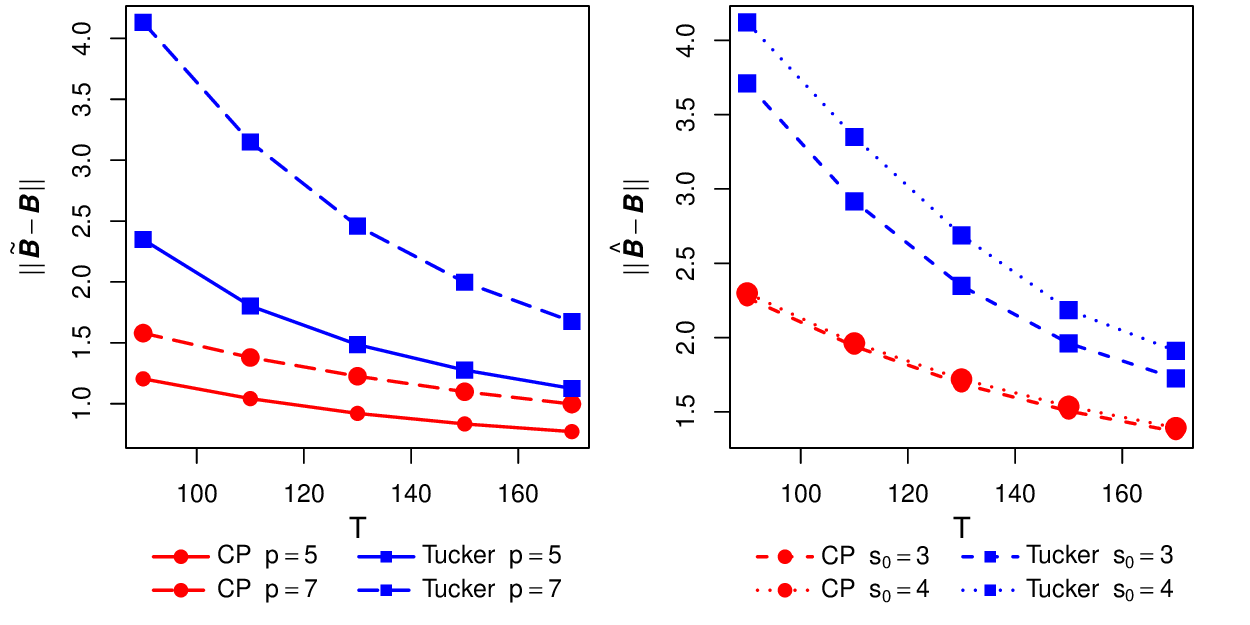}
	\caption{Estimation errors of the proposed estimators  in the non-sparse case ($\widetilde{\cm{B}}$, left panel) and the sparse case ($\widehat{\cm{B}}$, right panel), compared with the non-convex Tucker low-rank estimator for TAR models.}
	\label{Fig:sim2}
\end{figure}

In the second experiment, we compare the proposed estimators to the nonconvex Tucker low-rank estimator $\widetilde{\cm{B}}_{\textrm{Tucker}}$ in \cite{wang2024high} for TAR models. The data are generated from DGP 2 with a non-sparse or sparse coefficient tensor of CP rank  $R=4$. As noted in Remark \ref{remark:CPtucker} in the supplementary file, under this setting, $\cm{B}$ also admits a Tucker decomposition with Tucker ranks $r_d=4$ for $1\leq d\leq 6$. The left panel of  Figure \ref{Fig:sim2} displays the estimation errors  $\|\widetilde{\cm{B}}-\cm{B}\|_{\Fr}$ and  $\|\widetilde{\cm{B}}_{\text{Tucker}}-\cm{B}\|_{\Fr}$ against $T\in[90,170]$ for the  non-sparse case, where we vary $p\in\{5,7\}$. It can be seen that $\widetilde{\cm{B}}$ yields lower estimation errors than $\widetilde{\cm{B}}_{\text{Tucker}}$ in all settings, and the latter also appears to be more sensitive to the increase in $p$.
The right panel of  Figure \ref{Fig:sim2} displays the results for comparing $\widetilde{\cm{B}}_{\textrm{Tucker}}$ with  $\widehat{\cm{B}}$ for the sparse case, where we vary  $s_0\in\{3,4\}$ while fixing $p=10$. The advantage of $\widehat{\cm{B}}$ is clear, as $\widetilde{\cm{B}}_{\textrm{Tucker}}$ cannot take into account the sparsity structure.

\section{Empirical analysis}\label{sec:real}

\subsection{Mixed-frequency macroeconomic data}\label{subsec:macro}
We apply the proposed CP low-rank tensor stochastic regression to a mixed-frequency dataset consisting of all macroeconomic indicators  available since 1967 from the FRED-QD and FRED-MD databases at \url{https://www.stlouisfed.org/research/economists/mccracken/fred-databases}. The former includes $p_y=179$ quarterly variables  classified into eight main categories: (Q1) National Income and Product Accounts (NIPA), (Q2) Industrial Production, (Q3) Employment and Unemployment, (Q4) Housing, (Q5) Prices, (Q6) Interest Rates, (Q7) Money and Credit, and (Q8) others. The latter includes $112$ monthly variables  falling into seven categories including (M1) Output and Income, (M2) Trade (i.e., Consumption, Orders, and Inventories), (M3) Labor Market, (M4) Housing, (M5) Money and Credit, (M6) Interest and Exchange Rates, and (M7) Prices. We transformed all variables to stationarity according to \cite{McCracken2021} and \cite{McCracken2016}, and then standardized each series to have zero mean and unit variance. 

We use quarterly data as the response series, $\bm{y}_t\in\mathbb{R}^{179}$, where $t=1,\dots, T$ represents the quarter index. We arrange the 112 monthly variables   into a matrix $\bm{X}_t\in\mathbb{R}^{112\times3}$, where 3 corresponds to the number of months in each quarter. 
We apply the proposed model to predict the low-frequency outcomes using the high-frequency variables, i.e., $\bm{y}_t=\langle \cm{B}, \bm{X}_t\rangle +\bm{\varepsilon}_t$, where $\cm{B}\in\mathbb{R}^{179\times 112\times3}$ is estimated by the  CP low-rank estimator $\widetilde{\cm{B}}$. We use the data from Q1-1967 to Q4-2021 (i.e., $T=T_{\text{train}}+T_{\text{val}}=220$) to fit the model. The rank  $R=40$  is selected using time series cross-validation.
The validation set spans from Q1-2021 to Q4-2021 (i.e., $T_{\text{val}}=4$), and the selection is based on minimizing the average forecast error, $T_{\text{val}}^{-1}\sum_{i=1}^{T_{\text{val}}}\|\widehat{\bm{y}}_{T_\text{train}+i}-\bm{y}_{T_\text{train}+i}\|_2$, where the forecasts are computed using the same rolling forecast procedure applied to the test set, as we describe in Section \ref{sup:subsec21} of the supplementary file.

\begin{figure}[t]
	\centering 
	\includegraphics[width = \textwidth]{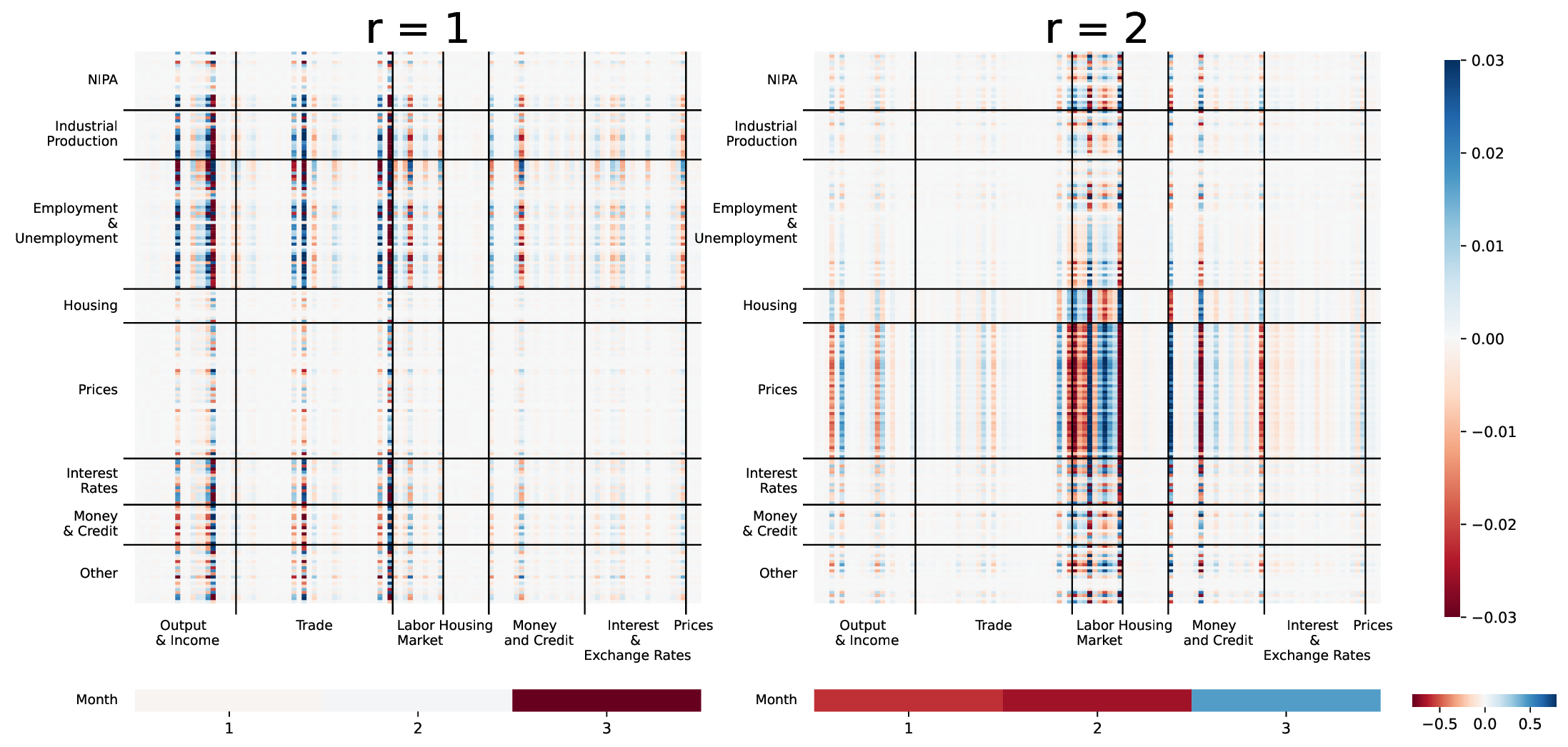}
	\caption{Heatmaps for estimates of $\bm{\beta}_{r,1}\circ\bm{\beta}_{r,2}$ (matrices at top) and $\bm{\beta}_{r,3}^\top$ (row vectors at bottom) for macroeconomic data, for $r=1$ (left panel) and 2 (right panel), where $r$ indexes $\omega_r$ in descending order. The rows and columns of $\bm{\beta}_{r,1}\circ \bm{\beta}_{r,2}$ correspond to quarterly responses and monthly predictors, respectively, and  $\bm{\beta}_{r,3}$ represents the temporal aggregation pattern.}
	\label{Fig:macro}
\end{figure}

The estimates of  $\omega_r$ for $1\leq r\leq 40$, sorted in descending order, are displayed in Figure \ref{Fig:macro1} in the supplementary file. It shows that the top two values, 95.1 and 74.1, are substantially larger than the others. Thus, we next take a closer look at the most important dynamic patterns revealed by the estimated factor vectors corresponding to $r=1$ and 2.
Specifically, we visualize the estimates of the matrices  $\bm{\beta}_{r,1}\circ\bm{\beta}_{r,2}$ and the vectors $\bm{\beta}_{r,3}^\top$ for $r=1$ and 2, obtained based on decomposition of $\widetilde{\cm{B}}$ in Figure \ref{Fig:macro}. 
As discussed in Section \ref{subsec:model}, $\bm{\beta}_{r,1}\circ \bm{\beta}_{r,2}$ captures the relationship between quarterly responses (rows) and monthly predictors (columns), despite their different sampling frequencies. Meanwhile, $\bm{\beta}_{r,3}\in\mathbb{R}^{3}$ represents the temporal aggregation pattern  that the model learns to summarize monthly observations to the quarterly scale.
Notably, the rows and columns of $\bm{\beta}_{r,1}\circ\bm{\beta}_{r,2}$  reveal interesting grouping patterns across the quarterly response and  monthly predictor variables, respectively. On the one hand, the top left panel of Figure \ref{Fig:macro} (i.e., $\bm{\beta}_{r,1}\circ\bm{\beta}_{r,2}$ for $r=1$) indicates that a small number of monthly variables within the categories (M1) Output and Income and (M2) Trade have particularly strong influence on many quarterly variables across all categories, except for (Q4) Housing, and to a lesser extent, (Q5) Prices. On the other hand, the top right panel of Figure \ref{Fig:macro} (i.e., $\bm{\beta}_{r,1}\circ \bm{\beta}_{r,2}$ for $r=2$) reveals some even more distinct block patterns. First, the monthly predictors within the (M3) Labor Market category generally exert a strong influence on nearly all quarterly variables across every category. 
Notably, the block at the intersection of (M3) Labor Market and (Q5) Prices stands out, suggesting that labor market conditions are a key driver of price movements.
Additionally, a few variables within (M1) Output and Income and (M5) Money and Credit also have a broad effect on the category (Q5) Prices. 

In addition, the bottom panel of Figure \ref{Fig:macro} displays estimates of the row vectors, $\bm{\beta}_{r,3}^\top$ for $r=1$ and 2. The result for $\bm{\beta}_{r,3}^\top$ with $r=1$  suggests that  the cross-sectional grouping pattern  revealed by the corresponding $\bm{\beta}_{r,1}\circ\bm{\beta}_{r,2}$ is only prominent in the last month of each quarter. By contrast,  the pattern revealed by $\bm{\beta}_{r,1}\circ\bm{\beta}_{r,2}$ with  $r=2$  is associated with all months within each quarter, with opposite effects observed between the first two months and the last month. For both $r=1$ and 2, the reason why the effect of the last month in each quarter differs from the other two may be that some indexes, such as interest rates, are updated in real time. Consequently, the month closest to the quarterly record date tends to have a stronger impact on the quarterly index series. This interesting finding demonstrates  the usefulness of the proposed model in uncovering temporal structure across mixed-frequency response and predictor variables. 

Moreover, we  assess the out-of-sample forecast accuracy for the high-dimensional quarterly response series based on a rolling procedure in Section \ref{sup:subsec21} of the supplementary file. It is shown that the proposed method achieves the lowest forecast error (FE) in five out of eight categories and the overall FE. The superior performance of our method over the benchmarks in predicting category (Q5) Prices is particularly noteworthy. Interestingly, the heatmap for $r=2$ in Figure \ref{Fig:macro} shows that (Q5) Prices are strongly associated with (M3) Labor Market. This finding further confirms the advantage of our approach in both prediction and interpretability. Additional  heatmaps to further demonstrate the favorable interpretability of our method compared to the vectorization-based sparse estimation are also provided in Section \ref{sup:subsec21} of the supplementary material.

\subsection{Taiwan air pollution data}\label{subsec:taiwan}

To evaluate the empirical performance of the proposed TAR, this section analyzes the Taiwan air pollution dataset complied by \cite{chen2021tensor} available at  \url{https://github.com/youlinchen/TCCA}, which contains monthly measurements of 7 pollutants across 12 monitoring stations in Taiwan from 2005 to 2017. The pollutants are sulfur dioxide (SO$_2$), carbon
monoxide (CO), ozone (O$_3$), particulate matter (PM$_{10}$), oxides of
nitrogen (NO$_x$), nitric oxide (NO), and nitrogen dioxide (NO$_{2}$).
The 12 monitoring stations can be  divided into northern (Guting, Tucheng, Taoyuan, and Hsinchu), southern (Erlin, Xinying, Xiaogang, and Meinong), and eastern (Yilan, Dongshan, Hualien, and Taitung) regions, with 4 stations in each region; see more details in the supplementary material for \cite{chen2021tensor}.

\begin{figure}[!t]
	\centering 
	\includegraphics[width=1.0\textwidth]{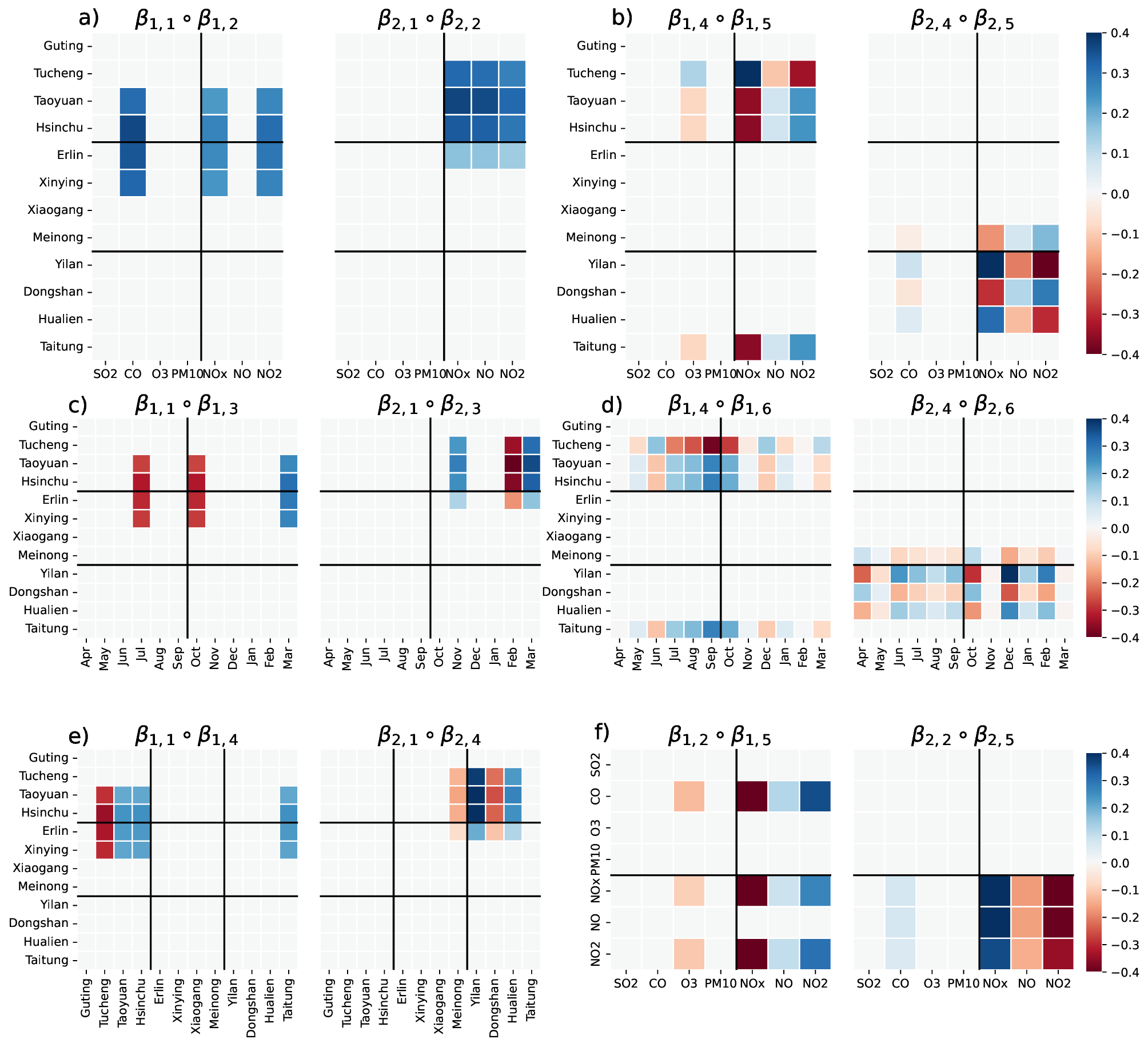}
	\caption{Heatmaps for estimates of $\bm{\beta}_{r,d_1}\circ\bm{\beta}_{r,d_2}$ for Taiwan air pollution data with $r=1$ and 2. The rows and columns of $\bm{\beta}_{r,d_1}\circ \bm{\beta}_{r,d_2}$ correspond to the $d_1$th and $d_2$th modes of $\cm{B}$, respectively.  The axis for monitoring stations is partitioned into three regions (from top to bottom): north, south, and east. The axis for pollutants is divided into two categories: nitrogen pollutants and others. The axis for months is grouped into two seasons: summer (April to September) and winter (October to March).}
	\label{Fig:TW_cross}
\end{figure}

To explore the monthly patterns, we further fold the data by treating month as a new dimension. This leads to a tensor-valued time series $\cm{Y}_t\in\mathbb{R}^{12\times7\times12}$ (monitoring stations $\times$ pollutants $\times $ months), where $t$ is the year index. We fit the proposed CP low-rank TAR of order one,  $\cm{Y}_t=\langle\cm{B},\cm{Y}_{t-1}\rangle +\cm{E}_t$, where $\cm{B}\in\mathbb{R}^{12\times7\times12\times 12\times7\times12}$, as using a higher lag order does not improve the forecasting performance for this dataset. Due to space limitations, we focus on the sparse estimator $\widehat{\cm{B}}$ in the main paper to highlight its greater interpretability; see additional results in the supplementary material. Based on the forecasting performance, we select $R=4$ and  $(s_1,\dots, s_6)=(4,3,3,4,4,12)$. The estimates $\widehat{\omega}_r$ for $1\leq r\leq 4$ in descending order are $26.6, 17.7, 6.5$, and  $5.5$. Since $\widehat{\omega}_1$ and $\widehat{\omega}_2$ carry much more weight in the coefficient tensor, similar to Section \ref{subsec:macro}, we highlight the results corresponding to  them.
Specifically, Figure \ref{Fig:TW_cross} displays heatmaps of various  $\bm{\beta}_{r,d_1}\circ \bm{\beta}_{r,d_2}$  across different mode combinations, where  $1\leq d_1\neq d_2\leq 6$, for $r=1$ and 2.  These visualizations illustrate the interactive relationships within the dynamic system, specifically, response-response (i.e., static relationships between outputs), predictor-predictor (i.e., static relationships between inputs), and response-predictor (i.e., output-input dynamics), as follows: 
\begin{itemize}
\item  \textit{Between  monitoring stations and pollutants:}  $\bm{\beta}_{r,1}\circ \bm{\beta}_{r,2}$ (monitoring stations response-pollutants response) and $\bm{\beta}_{r,4}\circ \bm{\beta}_{r,5}$ (monitoring stations predictor-pollutants predictor),  shown in panels (a) and (b) of Figure \ref{Fig:TW_cross}, respectively.
\item \textit{Between monitoring stations and months:} $\bm{\beta}_{r,1}\circ \bm{\beta}_{r,3}$ (monitoring stations response-months response) and $\bm{\beta}_{r,4}\circ \bm{\beta}_{r,6}$ (monitoring stations predictor-months predictor), shown in  panels (c) and (d) of Figure \ref{Fig:TW_cross}, respectively.
\item \textit{Output-input dynamics:} $\bm{\beta}_{r,1}\circ \bm{\beta}_{r,4}$ (dynamics for monitoring stations) and $\bm{\beta}_{r,2}\circ \bm{\beta}_{r,5}$ (dynamics for pollutants), shown in  panels (e) and (f) of Figure \ref{Fig:TW_cross}, respectively.
\end{itemize}

We summarize the main findings as follows. Panel (a) shows strong geographical clustering, indicating that the same pollutant is more likely to appear in nearby locations. Notably, nitrogen pollutants tend to concentrate in  the northern region, as indicated by the block in the top-right corner of the plot for $r = 2$. The dichotomy between nitrogen pollutants and most   other pollutants is also clear in panel (b), which reflects the grouping patterns in predictive signals. The heatmaps for $r = 1$ and $r = 2$ capture predictive signals originating from different regions and both exhibit geographical clustering.
Panel (c) indicates that a strong association between pollution in the northern region and the winter months  is reflected in the current year's responses. On the other hand, panel (d) shows that pollution levels across different regions from nearly all months of the previous year are predictive of the current year.
Panel (e) reveals that pollution in the northern region in the current year is primarily attributable to pollution from the north and east in the previous year. Finally, panel (f) suggests that nitrogen pollutants have a strong year-over-year influence on one another, but a much weaker influence on other pollutants, except for CO. In short, our findings highlight the multidimensional interpretability enabled by the proposed CP low-rank TAR model, an attractive feature not achievable with a Tucker low-rank structure.

In Section \ref{sup:subsec22} of the supplementary material, we further consider several alternative estimation methods for comparison of interpretability and forecast accuracy, including the nonconvex Tucker low-rank estimator for TAR models proposed by \cite{wang2024high}, and some benchmark estimators for high-dimensional VAR models based on vectorized data. The competitive forecasting performance of the proposed methods is also confirmed for this dataset.

\section{Conclusion and discussion}\label{sec:conclude}
This paper proposes high-dimensional  tensor stochastic regression and TAR models via CP decomposition, which include vector,  matrix, or tensor responses and predictors as special cases. It further leads to new vector, matrix, and tensor autoregressive (AR) models of general lag orders via CP decomposition. 
Unlike existing AR models based on Tucker decomposition, CP decomposition  enables the exploration of interactive relationship across various dimensions, as well as input-output dynamic mechanisms. 
Additionally, this general framework provides an efficient and interpretable approach to mixed-frequency regression with high-dimensional responses.
For the proposed models, CP low-rank and sparse CP low-rank estimators are introduced, for which non-asymptotic estimation error bounds are established, and an efficient alternating minimization algorithm is developed.  Simulation studies and empirical analyses of mixed-frequency macroeconomic data and spatio-temporal air pollution data demonstrate the advantages of the proposed methods in terms of both efficiency and interpretability.

To the best of our knowledge, this is the first work to exploit CP decomposition in general tensor stochastic regression. There are several important questions  that require future research. First, we use time series cross-validation to select the rank and sparsity parameters due to its practical simplicity and favorable performance in general model settings. It is worth exploring consistent estimation methods for the tuning parameters, such as the high-dimensional information criteria \citep[e.g.,][]{zheng2025} or a ridge-type ratio estimator \citep{wang2022high}. Second, while we study the estimation error bounds, the convergence analysis of the algorithm, similar to that in  \cite{sun2017store}, has not been established. The main challenge arises from the stochastic nature of the design matrix, whereas the former assumes a fixed design matrix.  We leave this challenging problem for future investigation.

\putbib[CPAR]
\end{bibunit}

\newpage
\renewcommand{\thesection}{S\arabic{section}}
\renewcommand{\thesubsection}{S\arabic{section}.\arabic{subsection}}
\renewcommand{\theequation}{S\arabic{equation}}
\renewcommand{\thetable}{S\arabic{table}}
\renewcommand{\thefigure}{S\arabic{figure}}
\renewcommand{\theremark}{S\arabic{remark}}
\renewcommand{\thelemma}{S.\arabic{lemma}}
\setcounter{lemma}{0}
\setcounter{section}{0}
\setcounter{table}{0}
\setcounter{remark}{0}
\setcounter{figure}{0}

\begin{bibunit}[apalike]
	\vspace*{10pt}	
	\begin{center}
		{\Large \bf Supplement for ``Tensor Stochastic Regression for High-dimensional Time Series via CP Decomposition''}
	\end{center}
	\vspace{10pt}
	
\begin{abstract}
This supplementary file contains five sections. Section \ref{sec:tucker} discusses the connections and differences between the proposed model and the Tucker low-rank TAR model.
Section \ref{sec:algo} introduces an efficient alternating minimization algorithm for implementing the proposed estimation.  Section \ref{sec:sim3} provides a simulation experiment to demonstrate the computational efficiency of the proposed algorithm. Section \ref{sup:sec2} presents additional empirical results for the mixed-frequency macroeconomic  data and the Taiwan air pollution data. Section \ref{sup:sec1} provides technical proofs for all  theorems and corollaries in the main paper. 
\end{abstract}

\section{Comparison with Tucker low-rank models}\label{sec:tucker}

The proposed TAR model and its VAR special case (see Remark \ref{remark:var} in Section \ref{subsec:model} of the main paper) are reminiscent of the high-dimensional TAR and VAR models introduced in  \cite{wang2022high} and \cite{wang2024high}, respectively. However, the latter are based on the Tucker decomposition of the coefficient tensor $\cm{B}$, rather than the CP decomposition. 

If an $N$th-order tensor $\cm{B}\in\mathbb{R}^{p_1\times\cdots\times p_N}$ has Tucker ranks $(r_1, \dots, r_N)$, then it admits the Tucker decomposition,
\[
\cm{B} = \cm{G}\times_{d=1}^N\bm{U}_d,\]
where $\cm{G}\in\mathbb{R}^{r_1\times\cdots\times r_N}$ is called the core tensor, and $\bm{U}_d\in\mathbb{R}^{p_d\times r_d}$ for $1\leq d \leq N$ are called  factor matrices \citep{Kolda2009}. 

Our motivation for considering the CP decomposition instead of the above Tucker decomposition  is twofold:
\begin{itemize}
\item[(i)] The identifiability of the Tucker decomposition requires much more complex constraints than the CP decomposition \citep{Kolda2009,Han2024}.
\item[(ii)] The CP decomposition  provides valuable insights into   interactions across different cross-sectional or temporal modes, a capability that the Tucker decomposition lacks.
\end{itemize}

First, unlike the CP decomposition,  which easily achieves uniqueness in practice (see Remark  \ref{remark:CPunique} in Section \ref{subsec:notation} of the main paper), the Tucker decomposition suffers from  rotational indeterminacy. Namely, for any nonsingular matrices $\bm{O}_d$, it holds $\cm{G}\times_{d=1}^N \bm{U}_d = (\cm{G}\times_{d=1}^N \bm{O}_d)\times_{i=d}^N(\bm{O}_d^{-1}\bm{U}_d)$. A common approach to addressing  this problem is to consider the higher-order singular value decomposition (HOSVD) \citep{de2000multilinear}. This requires orthogonality constraints on $\bm{U}_d$'s and $\cm{G}$; in particular, $\bm{U}_d^\top \bm{U}_d =\bm{I}_{r_d}$. However, these constraints lead to difficulties in not only algorithmic development but also  the interpretation of factor loadings, especially when the small sample size necessitates further sparsity assumptions on $\bm{U}_d$. 
On the one hand, it is challenging to simultaneously ensure both orthogonality  and sparsity of $\bm{U}_d$'s in the optimization algorithm \citep{wang2024high}. On the other hand, it may not be reasonable to believe that the orthonormal matrices $\bm{U}_d$'s are truly sparse, even though they could be sparse up to a rotation. Thus, the estimates of $\bm{U}_d$'s may not be interpretable. 

Second, in time series modeling, the use of Tucker decomposition is especially motivated from the factor modeling perspective, as $\bm{U}_d$'s naturally generalize the loading matrix in vector-valued factor models to the tensor-valued setting; see \cite{wang2022high} and \cite{wang2024high} for detailed discussions on the connection with factor models \citep{Bai2016}. Thus,  based on each loading matrix $\bm{U}_d$ in the Tucker decomposition, the factor structure for each  mode can be interpreted. 
Nonetheless, individually interpreting the loading matrices  $\bm{U}_d$'s  cannot offer insights into interactions across different cross-sectional or temporal modes. In contrast,  as discussed in Section \ref{subsec:model} of the main paper, the CP decomposition enables us to conveniently explore interactive patterns for any pair of modes $(d_1, d_2)$ by visualizing the matrix $\bm{\beta}_{r,d_1}\circ \bm{\beta}_{r,d_2}$.

\begin{remark}[Connection between CP and Tucker decompositions]\label{remark:CPtucker}
Suppose  $\cm{B}\in\mathbb{R}^{p_1\times\cdots\times p_N}$ has the Tucker decomposition, 
$\cm{B} = \cm{G}\times_{d=1}^N\bm{U}_d$, where $\bm{U}_d=(\bm{u}_{1,d},\dots, \bm{u}_{r_d, d})\in\mathbb{R}^{p_d\times r_d}$, and $\cm{G}\in\mathbb{R}^{r_1\times\cdots\times r_N}$.  Then it can be written in the form of the CP decomposition, $\cm{B} =\sum_{i_1=1}^{r_1} \cdots \sum_{i_N=1}^{r_N}g_{i_1,\dots, i_N}\bm{u}_{i_1,1}\circ\cdots\circ\bm{u}_{i_N,N}:=\sum_{\bm{i}}\omega_{\bm{i}} \bm{\beta}_{{\bm{i}},1}\circ\cdots\circ\bm{\beta}_{{\bm{i}},N}$, which is the sum of $\prod_{d=1}^N r_d$ rank-one tensors, where $\bm{i}:=(i_1,\dots, i_N)$,  $\omega_{\bm{i}}:=g_{i_1,\dots, i_N}$ is the $(i_1,\dots, i_N)$-th entry of $\cm{G}$, and $\bm{\beta}_{\bm{i},d}:=\bm{u}_{i_d,d}$, for $1\leq i_d\leq r_d$ and $1\leq d\leq N$. Nonetheless, note that the actual rank $R$ of $\cm{B}$ is usually smaller than $\prod_{d=1}^N r_d$, i.e., a more parsimonious CP decomposition may exist. Conversely, if $\cm{B}$ has the CP decomposition in \eqref{eq:CP} with rank $R$ (see Section \ref{subsec:notation} of the main paper), then it can be written in the Tucker decomposition form where the core tensor $\cm{G}$ is superdiagonal (i.e., $g_{i_1,\dots, i_N}=0$ unless $i_1=\cdots=i_N$), and the Tucker ranks, $r_1=\cdots=r_N=R$, cannot be further reduced. 
\end{remark}


\section{Alternating minimization algorithm}\label{sec:algo}
In this section, we develop an efficient alternating minimization algorithm for the proposed CP low-rank estimators in  Sections \ref{subsec:est1} and \ref{subsec:est2} of the main paper by adapting  the  alternating update  algorithm  in \cite{sun2017store}.
At each iteration, our algorithm consists of two main steps. In Step 1, we fix the parameters $\{\bm{\beta}_{r,m+d}: 1\leq d\leq n, 1\leq r\leq R\}$ and update the parameters $\{\omega_r, \bm{\beta}_{r,d}: 1\leq d\leq m, 1\leq r\leq R\}$.  Note that the first group of parameters corresponds to loading vectors associated with predictor tensor $\cm{X}_t\in\mathbb{R}^{q_1\times \cdots\times q_n}$, while the second group includes those associated with the response tensor $\cm{Y}_t\in\mathbb{R}^{p_1\times \cdots\times p_m}$.  Then in Step 2, we fix the latter and update the former.

\cite{sun2017store} focused on the sparse CP low-rank  regression for $i.i.d.$ data with a tensor-valued response  and a vector-valued predictor, i.e., $n=1$. In addition, the entrywise sparsity was imposed only on the loading vectors corresponding to the response, i.e.,  $\{\bm{\beta}_{r,d}: 1\leq d\leq m, 1\leq r\leq R\}$. We extend their approach in two directions: (i) We allow the predictor to be an $n$-th order tensor for $n\geq 1$; and (ii) the loading vectors corresponding to the predictor, i.e.,  $\{\bm{\beta}_{r,m+d}: 1\leq d\leq n, 1\leq r\leq R\}$, can also be sparse.

Denote $f_{r,t}=\cm{X}_t\times_{d=1}^{n}\bm{\beta}_{r,m+d}$ and $\cm{R}_{r,t}=\cm{Y}_t-\sum_{r^\prime\neq r}\omega_{r^\prime}f_{r^\prime, t}\bm{\beta}_{r^\prime,1}\circ\cdots\circ\bm{\beta}_{r^\prime,m}$. Since the  parameters $\{\bm{\beta}_{r,m+d}: 1\leq d\leq n, 1\leq r\leq R\}$ are fixed in Step 1, $f_{r,t}$ and $\cm{R}_{r,t}$ for $1\leq r\leq R$ are also fixed. The minimization of the loss function with respect to $\{\omega_r, \bm{\beta}_{r,d}: 1\leq d\leq m, 1\leq r\leq R\}$ is equivalent to the following minimization that is conducted separately for each $1\leq r\leq R$:
\begin{equation}\label{eq:p1lossfunc}
\underset{\omega_r, \bm{\beta}_{r,1},\ldots,\bm{\beta}_{r,m} }{\min}\frac{1}{T}\sum_{t=1}^T f_{r,t}^2\Big\lVert \cm{S}_{r,t} - \omega_{r}\bm{\beta}_{r,1}\circ\cdots\circ\bm{\beta}_{r,m}\Big\rVert_\Fr^2,
\end{equation}
where  $\cm{S}_{r,t}=f_{r,t}^{-1}\cm{R}_{r,t}$.
The constraints for \eqref{eq:p1lossfunc}  are $\lVert \bm{\beta}_{r,d} \rVert_2=1$ for the CP low-rank estimation in Section \ref{subsec:est1}, or  $\lVert \bm{\beta}_{r,d} \rVert_2=1$  and  $\lVert\bm{\beta}_{r,d}\rVert_0\leq s_d$ for the sparse CP low-rank estimation in Section \ref{subsec:est2}, where $1\leq d\leq m$. Note that \eqref{eq:p1lossfunc} is equivalent to the rank-one tensor decomposition,
\[
\underset{\omega_r, \bm{\beta}_{r,1},\ldots,\bm{\beta}_{r,m} }{\min} \left \lVert  \frac{\bar{\cm{R}}_r}{\sum_{t=1}^T f_{r,t}^2} - \omega_{r}\bm{\beta}_{r,1}\circ\cdots\circ\bm{\beta}_{r,m}\right\rVert_\Fr^2,
\]
where $\bar{\cm{R}}_r = \sum_{t=1}^T f_{r,t}^2 \cm{R}_{r,t}$. Using alternating minimization with respect to each $1\leq d\leq m$, under the constraint  $\lVert \bm{\beta}_{r,d} \rVert_2=1$, the solution is given by $\widetilde{\bm{\beta}}_{r,d}=\textrm{Norm}(\bar{\cm{R}}_r\times_{d^\prime=1, d^\prime\neq d}^{m}\bm{\beta}_{r,d^\prime})$, and then $\widetilde{\omega}_{r}=\bar{\cm{R}}_{r}\times_{d=1}^{m}\widetilde{\bm{\beta}}_{r,d}$. If the constraints  $\lVert\bm{\beta}_{r,d}\rVert_0\leq s_d$ are further imposed, the sparse rank-one tensor decomposition can be solved by a truncation-based procedure \citep{sun2017provable, sun2017store}: the solution to the alternating minimization is  $\widehat{\bm{\beta}}_{r,d} = \textrm{Norm}(\textrm{Truncate}(\widetilde{\bm{\beta}}_{r,d}, s_d))=\textrm{Norm}(\textrm{Truncate}(\bar{\cm{R}}_r\times_{d^\prime=1, d^\prime\neq d}^{m}\bm{\beta}_{r,d^\prime}))$, for $1\leq d\leq m$, and then $\widehat{\omega}_{r}=\bar{\cm{R}}_{r}\times_{d=1}^{m}\widehat{\bm{\beta}}_{r,d}$.

\begin{algorithm}[t]
\caption{Sparse CP low-rank tensor stochastic regression}
\label{alg:spralg}
\textbf{Input:}  Rank $R$, regularization parameters $s_d$, initialization $\omega_{r}^{(0)}$, $\bm{\beta}_{r,d}^{(0)}$, for $1\leq d\leq N$,  $1\leq r\leq R$.\\
\textbf{repeat} $i=0,1,2,\dots$\\\vspace{2mm}
\hspace*{3mm}Step 1: \textbf{for} $r\in\{1,\dots,R\}$ \textbf{do}\\\vspace{2mm}
\hspace*{25mm}\textbf{for} $d\in\{1,\ldots, m\}$ \textbf{do} \\\vspace{2mm}
\hspace*{30mm}$\widecheck{\bm{\beta}}_{r,d}^{(i+1)}=\bar{\cm{R}}_{r,t}^{(i)}\times_{d^\prime=1, d^\prime\neq d}^{m}\bm{\beta}_{r,d^\prime}^{(i)}$\\\vspace{2mm}
\hspace*{30mm}$\bm{\beta}_{r,d}^{(i+1)}=\textrm{Norm}(\textrm{Truncate}(\widecheck{\bm{\beta}}_{r,d}^{(i+1)}, s_d))$\\\vspace{2mm}
\hspace*{25mm}\textbf{end for} \\\vspace{2mm}
\hspace*{25mm}$\omega_{r}^{(i+1)}=\bar{\cm{R}}_{r,t}^{(i)}\times_{d=1}^{m}\bm{\beta}_{r,d}^{(i+1)}$\\\vspace{2mm}
\hspace*{19mm}\textbf{end for} \\\vspace{2mm}
\hspace*{3mm} Step 2: \textbf{for} $r\in\{1,\dots,R\}$ \textbf{do}\\\vspace{2mm}
\hspace*{25mm}\textbf{for} $d\in\{1,\ldots, n\}$ \textbf{do}\\\vspace{2mm}
\hspace*{30mm}$\widecheck{\bm{\beta}}_{r,m+d}^{(i+1)}=\big(\sum_{t=1}^T\bm{z}_{r,d,t}^{(i)}\bm{z}_{r,d,t}^{(i)\top} \big)^{-1} \lVert\cm{A}_{r}^{(i)} \rVert_\Fr^{-2} \sum_{t=1}^T\langle\cm{R}_{r,t}^{(i)}, \cm{A}_{r}^{(i)}\rangle\bm{z}_{r,d,t}^{(i)}$\\\vspace{2mm}
\hspace*{30mm}$\bm{\beta}_{r,m+d}^{(i+1)}=\text{Norm}(\textrm{Truncate}(\widecheck{\bm{\beta}}_{r,m+d}^{(i+1)}, s_d))$\\\vspace{2mm}
\hspace*{25mm}\textbf{end for} \\\vspace{2mm}
\hspace*{19mm}\textbf{end for} \\\vspace{2mm}
\textbf{until convergence}
\end{algorithm}

In Step 2, as the parameters $\{\omega_r, \bm{\beta}_{r,d}: 1\leq d\leq m, 1\leq r\leq R\}$ are fixed, due to the  bi-convex structure of the loss function, the alternating  update of $\bm{\beta}_{r,m+d}$  for $1\leq d\leq n$   reduces to an ordinary least squares regression, which has a closed-form solution as follows:
\begin{align*}
\bm{\beta}_{r,m+d} = \left (\sum_{t=1}^T\bm{z}_{r,d,t}\bm{z}_{r,d,t}^\top\right)^{-1}\frac{\sum_{t=1}^T\langle\cm{R}_{r,t}, \cm{A}_{r}\rangle\bm{z}_{r,d,t}}{\lVert\cm{A}_{r} \rVert_\Fr^2},
\end{align*}
where $\cm{A}_{r}=\omega_{r}\bm{\beta}_{r,1}\circ\cdots\circ\bm{\beta}_{r,m}$ and  $\bm{z}_{r,d,t}=\cm{X}_t\times_{d^\prime=1, d^\prime\neq d}^{n}\bm{\beta}_{r,m+d^\prime}$; see also \cite{sun2017store}. Then, for $1\leq d\leq n$, we let $\widetilde{\bm{\beta}}_{r,m+d}=\textrm{Norm}(\bm{\beta}_{r,m+d})$ for the CP low-rank estimation in Section \ref{subsec:est1}, and $\widehat{\bm{\beta}}_{r,m+d} = \textrm{Norm}(\textrm{Truncate}(\widetilde{\bm{\beta}}_{r,m+d}, s_d))$ for the sparse CP low-rank estimation in Section \ref{subsec:est2}.
Note that the equivalent response in Step 2 is a scalar, $\langle\cm{R}_{r,t}, \cm{A}_{r}\rangle$ for $t=1,\ldots,T$. This significantly reduces the computational burden compared to a tensor-valued response problem.
We present the algorithm for sparse CP low-rank estimation in Algorithm \ref{alg:spralg}.  Here $\bar{\cm{R}}_{r,t}^{(i)}$, $\bm{z}_{r,d,t}^{(i)}$ and $\cm{A}_{r}^{(i)}$ are all computed using parameters obtained from the $i$th iteration. For the non-sparse case, the same algorithm can be used by simply omitting the truncation operations in lines 6 and 13 of Algorithm \ref{alg:spralg}.

\begin{remark}[Computational complexities]\label{remark:compute}
Algorithm \ref{alg:spralg} has a per-iteration computational complexity of $O(R(mTp_y+nTq_x+\sum_{d=1}^{n}q_d^3))\asymp O(R[T(p_y+q_x)+q_{\max}^3)])$, if $m$ and $n$ are fixed, where $q_{\max}=\max_{1\leq d\leq n}q_d$. 
As a comparison, consider the Lasso estimator for sparse regression with vectorized response and predictor, which uses the $\ell_1$-regularization and does not assume any low-rank structure; see Remark \ref{remark:lasso} in Section \ref{subsec:est2} of the main paper. The  commonly used iterative shrinkage-thresholding algorithm (ISTA)  \citep{Beck2009} for the Lasso has a per-iteration computational complexity of $O(Tp_yq_x)$, which is generally much higher than Algorithm \ref{alg:spralg}, since $p_y=\prod_{d=1}^m p_d$ and $q_x=\prod_{d=1}^n q_d$ can rapidly increase when $\cm{Y}_t$ and $\cm{X}_t$ are third- or higher-order tensors, where $R, m$ and $n$ are much smaller.   The advantage of Algorithm \ref{alg:spralg} is two-fold: In Step 1, the rank-one tensor decomposition is highly efficient, and in Step 2, rather than having $\cm{X}_t$ interact with all dimensions of $\cm{Y}_t$, we only need to solve a regression problem with a scalar response. On the other hand, if we consider the TAR($L$) model, the non-convex Tucker low-rank estimator introduced in \cite{wang2024high} serves as a baseline. For TAR($L$) models, the per-iteration computational complexity of our algorithm can be simplified to
$O(R  (T p_y L+p_{\max}^3))$, whereas that for  the gradient descent algorithm in \cite{wang2024high}  is $O((T+r_{\max})p_y^2 L+r_{\max}^3)$, if $m$ is fixed, where $p_{\max}=\max_{1\leq d\leq m} p_d$, $r_{\max}=\max_{1\leq d\leq 2m+1} r_d$, and $r_d$ for $1\leq d\leq 2m+1$ are the Tucker ranks. Since the latter rate scales quadratically rather than linearly with  $p_y$, the algorithm for the Tucker low-rank estimator in \cite{wang2024high} is generally much slower than our method for TAR models; this is also confirmed by  our simulation study in Section \ref{sec:sim3}.
\end{remark}

\begin{remark}[Initialization]
Algorithm \ref{alg:spralg} requires predetermined values of $R$ and $s_d$'s. Note that the estimation theory in Section \ref{sec:est} of the main paper does not require these quantities to be minimal;  rather, it suffices that the true rank and sparsity levels are bounded by them. In practice, we may consider a range of $R$ and $s_d$'s, and select them via time-series cross-validation; i.e., we partition the sample into training and validation sets while preserving the temporal order, and select these parameters based on the out-of-sample forecast accuracy on the validation set; see, e.g., \cite{Wilms2023}. In addition, for the initial values $\omega_{r}^{(0)}$ and $\bm{\beta}_{r,d}^{(0)}$'s, we can first obtain an initial estimate of $\cm{B}$ via a convex optimization, and then conduct a  CP decomposition for $\cm{B}^{(0)}$ for any given $R$. Due to the high dimensionality of most practical problems, we recommend using the Lasso estimator for sparse regression with vectorized response and predictor to obtain $\cm{B}^{(0)}$.
\end{remark}

\section{Additional simulation experiment}\label{sec:sim3}

\begin{figure}[!t]
\centering 
\includegraphics[width = \textwidth]{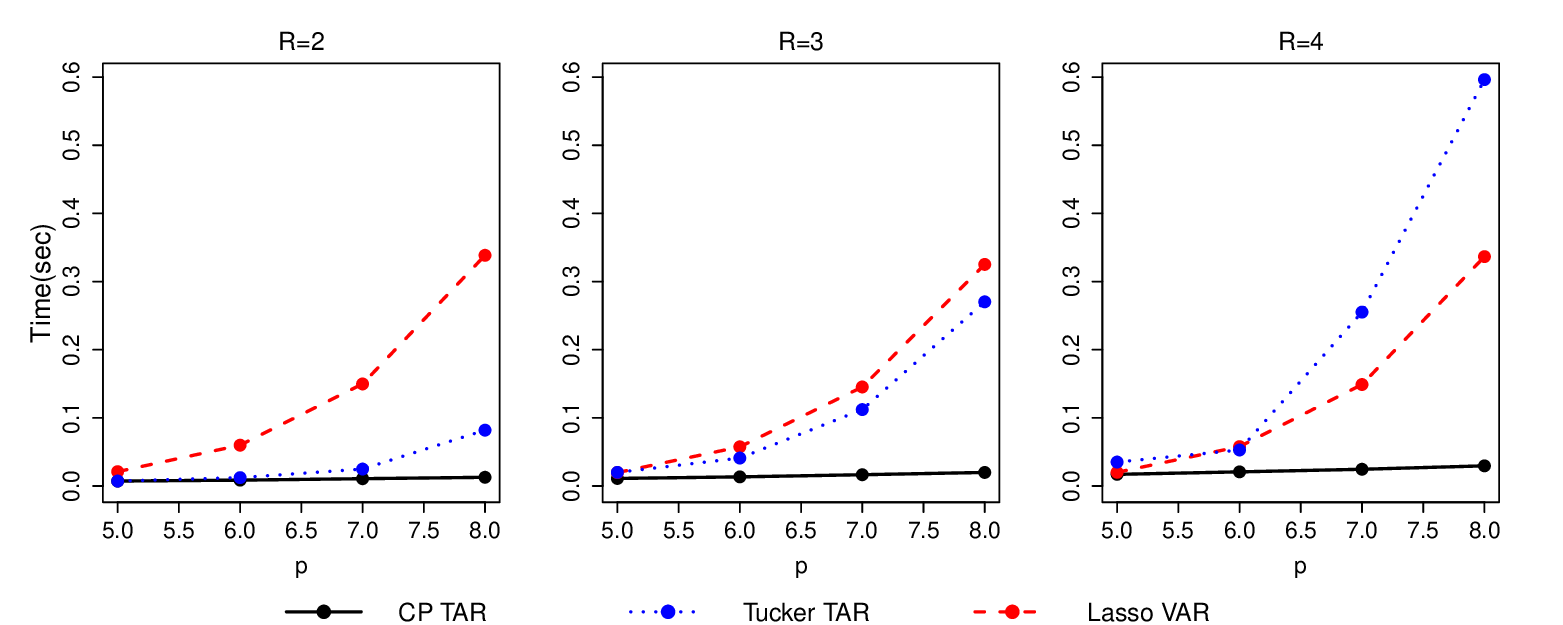}
\caption{Average computation time (in seconds) until convergence for the proposed sparse CP low-rank estimator (CP TAR) and nonconvex Tucker low-rank estimator (Tucker TAR) applied to the TAR model, along with the ISTA algorithm for the Lasso estimator of the VAR model with vectorized data (Lasso VAR).}
\label{Fig:sim3}
\end{figure}

We present an additional  experiment to assess the computation efficiency of the proposed algorithm. We generate data from DGP 2, the TAR(1) model with $T=100$, $s_0=4$, $R\in\{2,3,4\}$, and $p\in\{5,6,7,8\}$, and compare the computation time of  (i) Algorithm \ref{alg:spralg} for the proposed sparse CP low-rank estimator against two benchmark methods: (ii) the nonconvex Tucker low-rank estimator $\widetilde{\cm{B}}_{\textrm{Tucker}}$  in \cite{wang2024high} with Tucker ranks $r_1=\cdots= r_6=R$, which is implemented by a gradient descent algorithm, and (iii) the Lasso estimator for the sparse VAR model \citep{Basu2015} based on vectorized data, $\vect(\cm{Y}_t)=\bm{B}\vect(\cm{Y}_{t-1})+\bm{\varepsilon}_t$, which is implemented via  ISTA; see relevant discussions in Remark  \ref{remark:lasso} in Section \ref{subsec:est2} of the main paper, as well as Remarks \ref{remark:CPtucker} and  \ref{remark:compute} in supplement.

Figure  \ref{Fig:sim3} displays the average computation time required for each algorithm to converge.  It can be observed that, compared to the other two methods, the computation time of the proposed algorithm is much less sensitive to increases in $p$. In particular, the computation time of the Tucker low-rank estimator increases more dramatically with $p$, especially for larger values of $R$. As implied by  Remark \ref{remark:compute}, under the setting of this experiment, the per-iteration computational complexities  are (i) $O(R T p^3)$ for the proposed CP low-rank estimator, (ii) $O((T+R)p^6+R^3)$ for the Tucker low-rank estimator,  and (iii) $O(Tp^6)$ for the Lasso. Thus, the results in Figure  \ref{Fig:sim3} confirm the computational advantage of the proposed algorithm, as its  cost remains low as $R$ and $p$ increase.

\section{Additional empirical results} \label{sup:sec2}

\subsection{Mixed-frequency macroeconomic data}\label{sup:subsec21}

This section provides additional results for the mixed-frequency macroeconomic data in Section \ref{subsec:macro} of the main paper. Figure \ref{Fig:macro1} displays  the estimates of  $\omega_r$ for $1\leq r\leq 40$, sorted in descending order. We can observe that the top two values, 95.1 and 74.1, are substantially larger than the others. This suggests that the most important dynamic patterns are revealed by the estimated factor vectors corresponding to $r=1$ and 2.

\begin{figure}[t]
\centering 
\includegraphics[width =0.4 \textwidth]{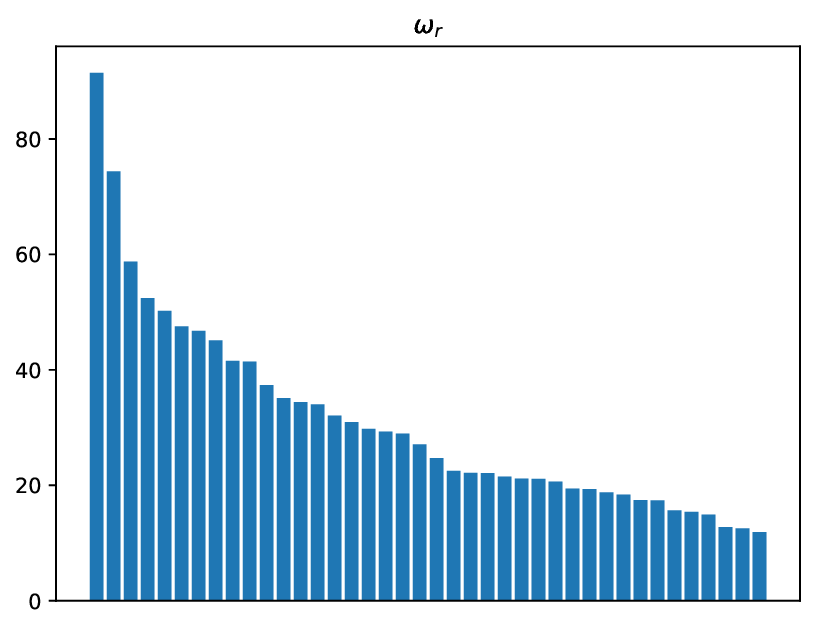}
\caption{Estimates of $\omega_r$ for $1\leq r\leq 40$, arranged in descending order, for macroeconomic data.}
\label{Fig:macro1}
\end{figure}

Moreover, we  assess the out-of-sample forecast accuracy for the high-dimensional quarterly response series based on a rolling procedure, with the test set spanning from Q1-2022 to Q4-2022 (i.e., $T_{\text{test}}=4$): first, we obtain the one-step-ahead forecast for Q1-2022 based on the   model fitted with training data from Q1-1957 to Q4-2021; next, we refit the model by adding the actual data for Q1-2022 into the training set, and then compute the forecast for Q2-2022; this procedure is repeated until we reach the end of the  test set. We compute the average  forecast error (FE) for all quarterly response variables over the test period, i.e., $T_{\text{test}}^{-1}\sum_{i=1}^{T_{\text{test}}}\|\widehat{\bm{y}}_{T+i}-\bm{y}_{T+i}\|_2$, as well as the FE for each of the eight categories of response variables. 
For comparison, we also consider  the  regression based on  vectorizing the data:
$\bm{y}_t=\langle \bm{B}, \vect(\bm{X}_t)\rangle +\bm{\varepsilon}_t$, where  $\vect(\bm{X}_t)\in\mathbb{R}^{336}$  collapses  the predictors and months into a single dimension, and $\bm{B}\in\mathbb{R}^{179\times 336}$ is the mode-one matricization of the tensor $\cm{B}$. While many dimension reduction methods for $\bm{B}$ are available, we consider the two most representative ones as benchmarks: $\bm{B}$ is either entrywise sparse or low-rank, and is estimated  via the Lasso \citep{Tibshirani1996}  or low-rank matrix factorization \citep{Chi2019}, respectively. It is worth noting that, unlike the tensor $\cm{B}\in\mathbb{R}^{179\times 112\times3}$, the vectorization approach makes it  difficult to clearly interpret the grouping patterns across the 112 monthly variables (or across the three months) from   $\bm{B}\in\mathbb{R}^{179\times 336}$, as they are confounded by the months.

\begin{table}[t]
\centering
\begin{tabularx}{\linewidth}{lZZZZZZZZc}\toprule
	& Q1 & Q2 & Q3 & Q4 & Q5  &  Q6 &  Q7  & Q8 & Overall\\\midrule
	Vectorization +  Lasso & 3.64 & \B 2.12 & 5.67 & 2.72 & 9.76 & 4.79 & 3.27 & 4.27 & 14.51\\
	Vectorization +  Low-rank & 3.61 & 2.27 & 4.18 & \B 2.34 & 6.36 & 4.60 & 4.09 & \B 3.42 & 11.69\\
	Proposed  CP &	\B 3.27 & 2.27 & \B 3.72 & 2.84 & \B 5.83 & \B 3.35 & \B 3.16 & 3.66 & \B 10.57\\\bottomrule
\end{tabularx}
\caption{Out-of-sample forecast errors  across eight categories of quarterly response variables for macroeconomic data. The smallest value in each column is highlighted in bold.}
\label{Tab:errorMacro}
\end{table}

Table \ref{Tab:errorMacro} reports the FEs based on the proposed CP low-rank method and the two benchmarks mentioned above, for all quarterly outcomes and for each category. The proposed method achieves the lowest FE in five out of eight categories and the overall FE. The superior performance of our method over the benchmarks in predicting category (Q5) Prices is particularly noteworthy. Interestingly, the heatmap for $r=2$ in Figure \ref{Fig:macro} shows that (Q5) Prices are strongly associated with (M3) Labor Market. This finding further confirms the advantage of our approach in both prediction and interpretability.

Additionally, we provide  heatmaps to further demonstrate the favorable interpretability of our method compared to the vectorization-based sparse estimation: $\bm{y}_t=\langle \bm{B}, \vect(\bm{X}_t)\rangle +\bm{\varepsilon}_t$, where  $\vect(\bm{X}_t)\in\mathbb{R}^{336}$  collapses  the predictors and months into a single dimension, and $\bm{B}\in\mathbb{R}^{179\times 336}$ is the mode-one matricization of the tensor $\cm{B}\in\mathbb{R}^{179\times 112\times3}$. 
The left panel of Figure \ref{Fig:MacroMatrix_Lasso} shows the Lasso estimate of $\bm{B}$ based on the vectorization approach. As shown, the matrix contains numerous weak and unstructured signals, and the collapse of the predictor and month dimensions into a single column dimension results in a cluttered and dense structure that is virtually impossible to interpret. In contrast, our tensor-based approach preserves the multidimensional structure of the data, allowing for a more
meaningful investigation of features across each mode. 
The CP decomposition enables the extraction of dominant signal patterns from the parameter matrix, resulting in a clearer and more interpretable representation. This is also reflected by the much more structured patterns in the  matricization obtained by flattening the estimate $\widetilde{\cm{B}}$ based on  the proposed CP low-rank estimation, as shown in the right panel of Figure \ref{Fig:MacroMatrix_Lasso}.

\begin{figure}[t]
\centering
\begin{tabular}{ll}
	\small{(a) Vectorization + Lasso} & \small{(b) Proposed CP}\\
	\includegraphics[width=0.48\linewidth]{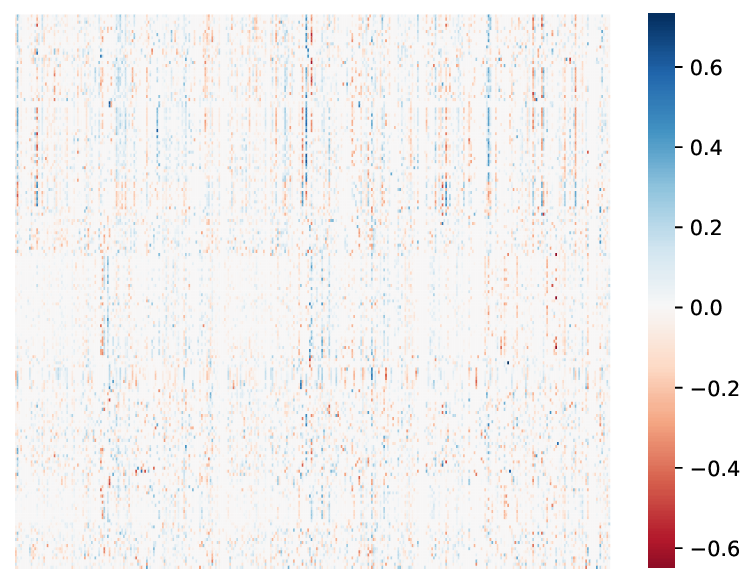} & \includegraphics[width=0.48\linewidth]{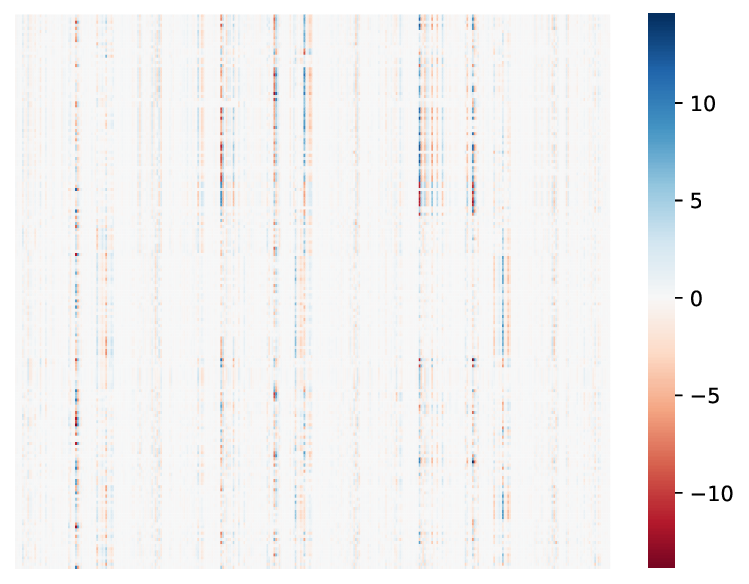}
\end{tabular}
\caption{Heatmaps of $\bm{B}$ estimated via Lasso regression based on the vectorization approach (left) and the corresponding matrix obtained by flattening the estimate $\widetilde{\cm{B}}$ based on  the proposed CP low-rank estimation (right), for mixed-frequency macroeconomic data.}
\label{Fig:MacroMatrix_Lasso}
\end{figure}

\subsection{Taiwan air pollution data}\label{sup:subsec22}

For the Taiwan air pollution data in Section \ref{subsec:taiwan} of the main paper, we further consider several alternative  methods for comparison of interpretability and forecast accuracy, including 
\begin{itemize}
\item VAR model based on vectorized data $\bm{y}_t=\vect(\cm{Y}_t)$:
$\bm{y}_t=\bm{B} \bm{y}_{t-1}+\bm{\varepsilon}_t$, where    $\bm{B}$ is the $p_y\times p_y$ matricization of $\cm{B}$, and $\bm{\varepsilon}_t=\vect(\cm{E}_t)$. We consider two estimators, the Lasso \citep{Tibshirani1996} for entrywise sparse matrix $\bm{B}$, and the low-rank estimator for $\bm{B}$ estimated via  low-rank matrix factorization \citep{Chi2019}.

\item Tucker low-rank TAR model: the nonconvex Tucker low-rank estimator proposed by \cite{wang2024high} for the TAR model,  $\cm{Y}_t=\langle\cm{B},\cm{Y}_{t-1}\rangle +\cm{E}_t$,  where  $\cm{B}= \cm{G}\times_{d=1}^N\bm{U}_d$, with  $\cm{G}\in\mathbb{R}^{r_1\times\cdots\times r_6}$ , $\bm{U}_d\in\mathbb{R}^{p_d\times r_d}$ for $1\leq d \leq 6$, and $(p_1,\dots, p_6)=(12,7,12,12,7,12)$. The selected Tucker ranks are $(r_1,\dots, r_6)=(3,2,5,1,2,3)$. 
\end{itemize}
We consider both the proposed CP low-rank estimator $\widetilde{\cm{B}}$ and sparse CP low-rank estimator $\widehat{\cm{B}}$.  For the Lasso estimator, we examine two penalty levels: one chosen to match the number of nonzero parameters in the fitted sparse CP low-rank TAR model, and another selected to optimize forecast performance, which results in a dense estimate; these are referred to as Lasso (sparse) and Lasso (dense) in Table  \ref{Tab:errorTW}. We conduct the rolling forecast procedure as in Section \ref{sup:subsec21}, using the data in 2017 as the test set. 

\begin{table}[!t]
\begin{tabularx}{0.95\linewidth}{lrrrYYY}
	\toprule
	& \multicolumn{3}{c}{VAR} & \multicolumn{3}{c}{TAR} \\
	\cmidrule(lr){2-4} \cmidrule(lr){5-7}
	Station      & Lasso (dense) & Lasso (sparse) & Low-rank & Tucker     & \multicolumn{1}{c}{CP}        & Sparse CP \\ \midrule
	Guting       & 8.97         & 8.98        & 9.26     & 8.81          & \textbf{8.33} & 9.00      \\
	Tucheng      & 10.25        & 10.29       & 10.59    & 9.29          & \textbf{7.99} & 10.27     \\
	Taoyuan      & 9.00         & 9.02        & 9.29     & \textbf{8.03} & 8.27          & 9.22      \\
	Hsinchu      & 11.27        & 11.27       & 11.02    & 9.91          & \textbf{9.05} & 11.34     \\
	Erlin        & 11.68        & 11.72       & 11.14    & 10.85         & \textbf{9.42} & 11.94     \\
	Xinying      & 10.88        & 10.88       & 10.62    & 11.10         & \textbf{9.93} & 11.15     \\
	Xiaogang     & 10.79        & 10.81       & 10.35    & \textbf{9.16} & 11.36         & 10.81     \\
	Meinong      & 10.35        & 10.37       & 9.76     & \textbf{8.15} & 10.05         & 10.37     \\
	Yilan        & 10.75        & 10.74       & 10.85    & 9.65          & \textbf{7.93} & 10.67     \\
	Dongshan     & 11.41        & 11.42       & 11.08    & 10.50         & \textbf{9.30} & 11.42     \\
	Hualien      & 9.94         & 9.93        & 9.56     & 10.98         & \textbf{8.63} & 9.93      \\
	Taitung      & 10.10        & 10.09       & 9.65     & 10.26         & \textbf{9.17} & 10.09     \\
	All          & 36.31        & 36.35       & 35.64    & 33.87         & \textbf{31.76} & 36.56    \\ \bottomrule
\end{tabularx}%
\caption{Out-of-sample forecast errors across 12 monitoring stations of response variables for Taiwan air pollution data.    The smallest value in each column is highlighted in bold.}
\label{Tab:errorTW}
\end{table}

Table \ref{Tab:errorTW} shows that the proposed CP low-rank TAR model (in the non-sparse case) achieves the smallest overall forecast error (FE) and yields the lowest FE in 9 out of 12 monitoring stations. The larger FEs from the Lasso estimators and the sparse version of our model suggest that the coefficient tensor may not be sparse.
Meanwhile, the Tucker low-rank estimator produces the second smallest overall FE and achieves the lowest FE in 3 out of 12 stations. These results corroborate the presence of an inherent low-rank structure in the data.  The superior forecast accuracy of our model compared to the Tucker low-rank TAR model may be attributed to overparameterization in the latter. Specifically, the Tucker model is equivalent to a CP low-rank TAR model with an effective rank of  $R=\prod_{d=1}^6r_d=180$, whereas our non-sparse model uses a much smaller CP rank of just $R=11$.

\begin{figure}[!t]
\centering
\includegraphics[width=0.8\textwidth]{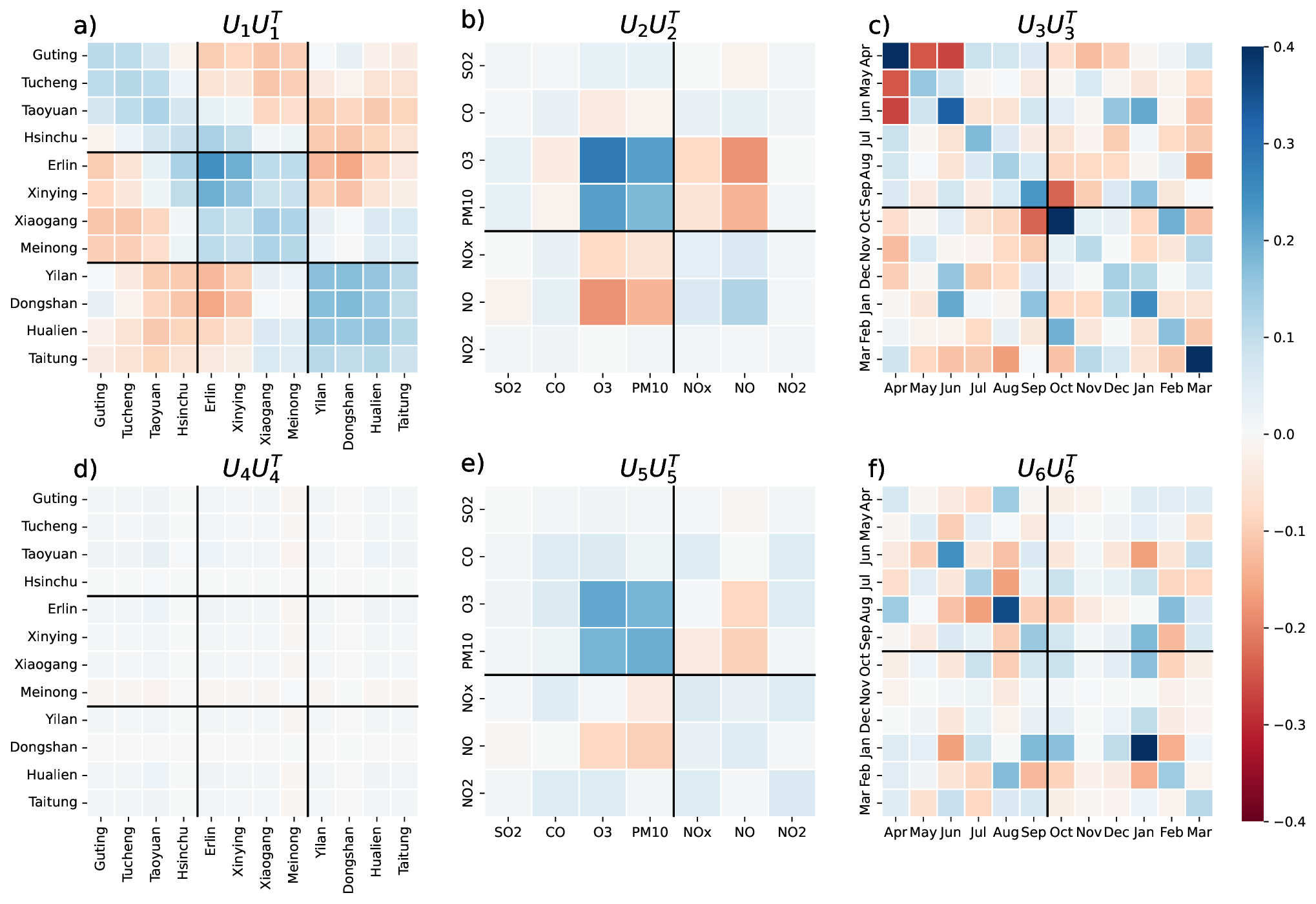}
\caption{Projection matrices from the fitted Tucker low-rank TAR model. $\boldsymbol{U}_1\boldsymbol{U}_1^\top$ corresponds to predictor loadings for monitoring stations; $\boldsymbol{U}_2\boldsymbol{U}_2^\top$ for pollutants; and $\boldsymbol{U}_3\boldsymbol{U}_3^\top$ for  months. Similarly, $\boldsymbol{U}_4\boldsymbol{U}_4^\top$, $\boldsymbol{U}_5\boldsymbol{U}_5^\top$, and $\boldsymbol{U}_6\boldsymbol{U}_6^\top$ represent the response loadings for monitoring stations, pollutants, and months, respectively.}
\label{Fig:TWnc}
\end{figure}

Figure \ref{Fig:TWnc} presents the projection matrices $\bm{U}_d \bm{U}_d^\top$ for $1 \leq d \leq 6$, obtained from the fitted Tucker low-rank TAR model. Overall, the patterns are generally difficult to interpret. For instance, $\bm{U}_4\bm{U}_4^\top$, which corresponds to the predictor loadings for the monitoring stations, has uniformly small entries and lacks clear structure. The same mode (i.e., the fourth mode of $\cm{B}$) corresponds to $\bm{\beta}_{r,4}$ in the proposed model. Unlike the uninformative projection in $\bm{U}_4\bm{U}_4^\top$, the loading vector $\bm{\beta}_{r,4}$ reveals distinct geographical groupings; see the rows of panels (b) and (d), and the columns of panel (e) in Figure \ref{Fig:TW_cross} in Section \ref{subsec:taiwan} of the main paper. On the other hand, $\bm{U}_1\bm{U}_1^\top$ represents the response loadings for the monitoring stations. 
While its diagonal blocks reveal three distinct groups that align with the three geographical regions, compared to an analysis of $\bm{\beta}_{r,1}$ for different $r$, it falls short in highlighting the areas most affected by air pollution; see the rows of panels (a), (c), and (e) in Figure \ref{Fig:TW_cross} in Section \ref{subsec:taiwan} of the main paper. In addition, the information from  $\bm{U}_2\bm{U}_2^\top$ and  $\bm{U}_5\bm{U}_5^\top$ is also limited; these can be compared to $\bm{\beta}_{r,2}$ and $\bm{\beta}_{r,5}$, respectively. From both projection matrices, we observe two important blocks, one for O$_3$ and PM$_{10}$, and the other for NO and NO$_x$. Moreover, the patterns of $\bm{U}_3\bm{U}_3^\top$ and  $\bm{U}_6\bm{U}_6^\top$ are too scattered to interpret clearly.

\begin{figure}[t]
\centering
\begin{tabular}{ll}
	\small{(a) Sparse VAR} & \small{(b) Sparse CP low-rank TAR}\\
	\includegraphics[width=0.48\linewidth]{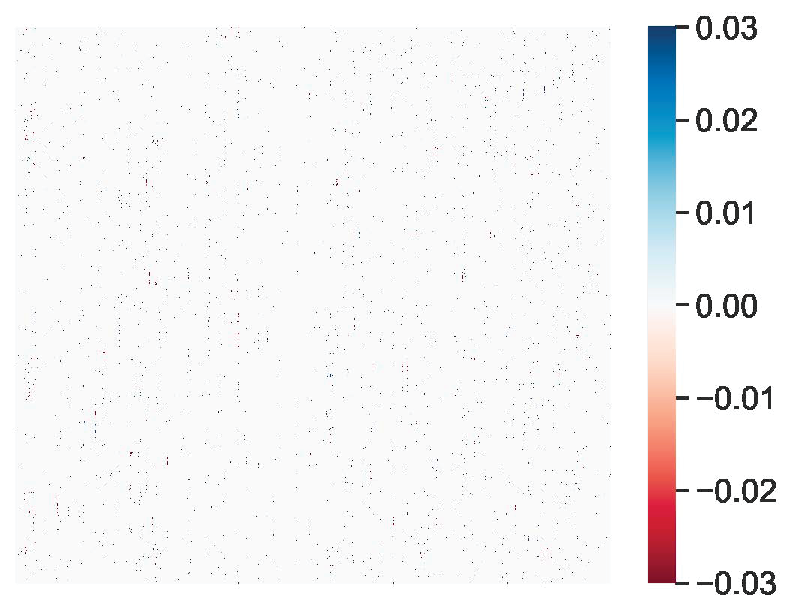} &
	\includegraphics[width=0.48\linewidth]{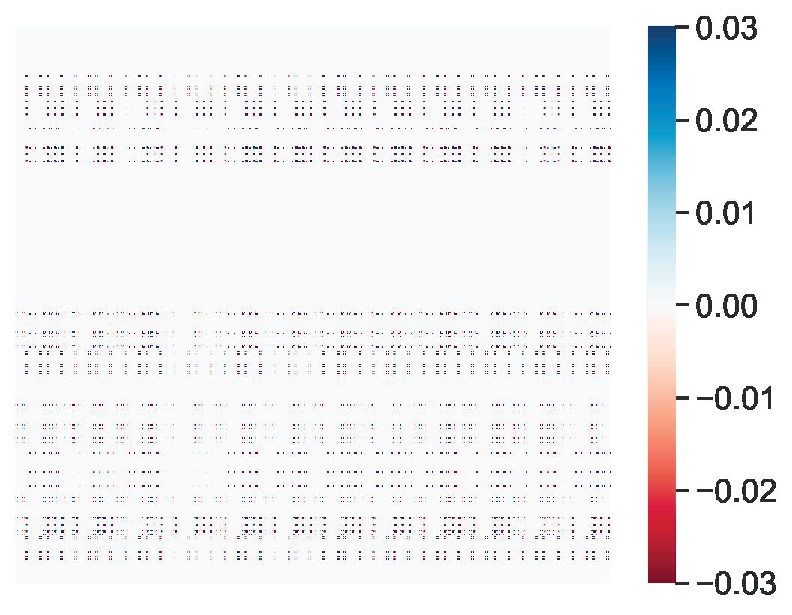}
\end{tabular}
\caption{Heatmap of $\bm{B}$ estimated via Lasso based on the VAR model (left) and the corresponding matrix obtained by flattening the estimate  $\widehat{\cm{B}}$ based on the proposed sparse CP low-rank TAR model, for Taiwan air pollution data.}
\label{Fig:TWlasso}
\end{figure}

Finally, Figure \ref{Fig:TWlasso} presents heatmaps of $\bm{B}$ estimated via Lasso under the VAR model, together with the corresponding matrix obtained by flattening the estimate $\widehat{\cm{B}}$ from the proposed sparse CP low-rank TAR model. The Lasso estimate is clearly difficult to interpret, as the signals appear overly diffuse due to the loss of tensor structure. In contrast, our method incorporates dimension reduction in a structurally informed way, as reflected by the more coherent grouping patterns in the latter.

\section{Technical proofs}
\label{sup:sec1}

\subsection{Notations}
Let
$\pazocal{S}^{p-1}=\{\bm{\beta}\in\mathbb{R}^p: \lVert\bm{\beta}\rVert_2=1\}$ denote the unit sphere in $\mathbb{R}^p$. For any positive integer $K$, denote $[K]=\{1,2, \dots, K\}$. Let $p_y=\prod_{d=1}^m p_d$ and $q_x=\prod_{d=1}^n q_d$.
For a tensor $\cm{B}\in\mathbb{R}^{p_1\times\cdots\times p_m\times q_1\times \cdots \times q_n}$, denote by  $\cm{B}_{[m]}$ the $p_y$-by-$q_x$ matrix  (a.k.a., multi-mode matricization) obtained by collapsing the first $m$ modes of $\cm{B}$ into rows and the last $n$ modes into columns, where the ordering of indices follows the same convention as in tensor vectorization. Denote $\bm{x}_t=\vect(\cm{X}_t)$,  $\bm{y}_t=\vect(\cm{Y}_t)$, and $\bm{\varepsilon}_t=\vect(\cm{E}_t)$.

\subsection{Auxiliary lemmas}
This section provides six auxiliary  lemmas used in the proofs of  theorems and corollaries.

\begin{lemma}\label{lemma:eigen}
Suppose that Assumption \ref{assum:stationary}(ii) holds for $\{\bm{x}_t\}$.
Let   $\bm{\Sigma}_x(\ell)=\mathbb{E}(\bm{x}_t\bm{x}_{t-\ell}^\top)$ for $\ell\in\mathbb{Z}$, and $\bm{\Sigma}_x :=\bm{\Sigma}_x(0)= \mathbb{E}(\bm{x}_t\bm{x}_{t}^\top)$. In addition, let $\underline{\bm{\Sigma}}_{x} = \mathbb{E}(\underline{\bm{x}}_T \underline{\bm{x}}_T^\top)=[\bm{\Sigma}_x(\ell-k)]_{1\leq \ell,k\leq T}$, where $\underline{\bm{x}}_T=(\bm{x}_T^\top, \bm{x}_{T-1}^\top\dots, \bm{x}_1^\top)^\top$.  Then it holds 
\begin{equation}\label{eq:eigen1}
	\kappa_1\leq \lambda_{\min}(\bm{\Sigma}_x) \leq \lambda_{\max}(\bm{\Sigma}_x) \leq \kappa_2 
\end{equation}
and 
\begin{equation}\label{eq:eigen2}
	\kappa_1\leq \lambda_{\min}(\underline{\bm{\Sigma}}_x) \leq \lambda_{\max}(\underline{\bm{\Sigma}}_x) \leq \kappa_2,
\end{equation}
where $\kappa_1 = \min_{|z|=1}\lambda_{\min}(\bm{\Phi}_x^{\HH}(z)\bm{\Phi}_x(z))$ and $\kappa_2 = \max_{|z|=1}\lambda_{\max}(\bm{\Phi}_x^{\HH}(z)\bm{\Phi}_x(z))$, with
$\bm{\Phi}_x^{\HH}(z)$ being the conjugate transpose of $\bm{\Phi}_x(z) = \sum_{j=1}^{\infty}\bm{\Phi}_j^x z^j$,  for $z\in\mathbb{C}$.
\end{lemma}

\begin{proof}[Proof of Lemma \ref{lemma:eigen}]
The results of this lemma are provided by  Lemma S18(i) in the supplementary material for \cite{zheng2025}.
\end{proof}

\begin{lemma}\label{lemma:eigenAR}
Suppose that $\cm{Y}_t=\sum_{\ell=1}^{L}\langle\cm{B}_\ell,\cm{Y}_{t-\ell}\rangle+\cm{E}_t$ is a stationary TAR($L$) model, and Assumption \ref{assum:error}(i) holds for $\{\cm{E}_t\}$.
Let $\bm{x}_t=(\bm{y}_{t-1}^\top, \dots, \bm{y}_{t-L}^\top)^\top$, where $\bm{y}_t=\vect(\cm{Y}_t)$. Then it holds 
\begin{equation}\label{eq:eigenAR1}
	\kappa_1^{\textrm{AR}}\leq \lambda_{\min}(\bm{\Sigma}_x) \leq \lambda_{\max}(\bm{\Sigma}_x) \leq \kappa_2^{\textrm{AR}} 
\end{equation}
and 
\begin{equation}\label{eq:eigenAR2}
	\lambda_{\max}(\underline{\bm{\Sigma}}_x) \leq \kappa_2^{\textrm{AR}} L,
\end{equation}
where $\kappa_1^{\textrm{AR}}=	\lambda_{\min}(\bm{\Sigma}_\varepsilon)\mu_{\min}(\bm{\Psi}_*)$  and  $\kappa_2^{\textrm{AR}}=\lambda_{\max}(\bm{\Sigma}_\varepsilon)\mu_{\max}(\bm{\Psi}_*)$, respectively. Here $\mu_{\min}(\bm{\Psi}_*) = \min_{|z|=1}\lambda_{\min}(\bm{\Psi}_*(z)\bm{\Psi}_*^{\HH}(z))$ and $\mu_{\max}(\bm{\Psi}_*) = \max_{|z|=1}\lambda_{\max}(\bm{\Psi}_*(z)\bm{\Psi}_*^{\HH}(z))$, where
$\bm{\Psi}_*(z) = \sum_{j=0}^{\infty}\bm{\Psi}_j^* z^j$, with $\bm{\Psi}_0^*=\bm{I}$, and $\bm{\Psi}_*^{\HH}(z)$ is the conjugate transpose of $\bm{\Psi}_*(z)$, for  $z\in\mathbb{C}$.
\end{lemma}

\begin{proof}[Proof of Lemma \ref{lemma:eigenAR}]
Denote
$\bm{\Sigma}_y(\ell)=\mathbb{E}(\bm{y}_t\bm{y}_{t-\ell}^\top)$ for $\ell\in\mathbb{Z}$, and  $\underline{\bm{\Sigma}}_y = \mathbb{E}(\underline{\bm{y}}_L \underline{\bm{y}}_L^\top)=[\bm{\Sigma}_y(\ell-k)]_{1\leq \ell,k\leq L}$, where $\underline{\bm{y}}_L=(\bm{y}_L^\top, \bm{y}_{L-1}^\top\dots, \bm{y}_1^\top)^\top$. Thus, $\bm{\Sigma}_x=\underline{\bm{\Sigma}}_y$.  It can be verified that  Assumption \ref{assum:stationary}(ii)  holds for the process $\bm{y}_{t-1}=\sum_{j=1}^{\infty}\bm{\Phi}_j^x \bm{\epsilon}_{t-j}=\bm{\Phi}_x(B) \bm{\epsilon}_{t}= \sum_{j=1}^{\infty}\bm{\Psi}_{j-1}^* \bm{\Sigma}_{\varepsilon}^{1/2} \bm{\xi}_{t-j}=\bm{\Psi}_*(B)B \bm{\Sigma}_{\varepsilon}^{1/2} \bm{\xi}_t$, by setting $\bm{\epsilon}_{t}=\bm{\xi}_t$ and the MA characteristic polynomial $\bm{\Phi}_x(z)=z\bm{\Psi}_*(z)\bm{\Sigma}_{\varepsilon}^{1/2}$, i.e., $\bm{\Phi}_j^x=\bm{\Psi}_{j-1}^*\bm{\Sigma}_{\varepsilon}^{1/2}$ for $j\geq1$. Moreover, note that \eqref{eq:eigen2}  is applicable to any positive integer $T$. As a result, we can apply Lemma \ref{lemma:eigen} by treating  $\bm{y}_{t-1}$ as the process $\bm{x}_t$ therein.  Then we can show that \eqref{eq:eigen2} holds for $\underline{\bm{\Sigma}}_y$ with the bounds $\kappa_i$ replaced by $\kappa_i^{\textrm{AR}}$ for $i=1,2$. Thus, \eqref{eq:eigenAR1} is proved.

To prove  \eqref{eq:eigenAR2}, the key is to write $\bm{x}_t=(\bm{y}_{t-1}^\top, \dots, \bm{y}_{t-L}^\top)^\top = (B^0, B^1,  \dots, B^{L-1})^\top \bm{y}_{t-1}=\bm{\Phi}_x(B) \bm{\epsilon}_{t}$, by setting $\bm{\epsilon}_{t}=\bm{\xi}_t$ and the MA characteristic polynomial 
\[
\bm{\Phi}_x(z)=(z \bm{I}, z^2 \bm{I}, \dots, z^L\bm{I})^\top \bm{\Psi}_*(z)\bm{\Sigma}_{\varepsilon}^{1/2}
\]
Now we  apply Lemma \ref{lemma:eigen} by regarding $(\bm{y}_{t-1}^\top, \dots, \bm{y}_{t-L}^\top)^\top$ as the process $\bm{x}_t$ therein. This leads to 
\begin{equation*}
	\lambda_{\max}(\underline{\bm{\Sigma}}_x) \leq \max_{|z|=1}\lambda_{\max}(\bm{\Phi}_x^{\HH}(z)\bm{\Phi}_x(z)) \leq \lambda_{\max}(\bm{\Sigma}_\varepsilon)\mu_{\max}(\bm{\Psi}_*) L = \kappa_2^{\textrm{AR}} L,
\end{equation*}
i.e., \eqref{eq:eigenAR2}. The proof of this lemma is complete.
\end{proof}

\begin{lemma}[Hanson-Wright inequality for stationary time series]\label{lemma:hansonw}
Suppose that $\bm{x}_t =\sum_{j=1}^{\infty}\bm{\Phi}_j^x \bm{\epsilon}_{t-j}$, where $\bm{\epsilon}_t\in\mathbb{R}^{q_\epsilon}$,  $\bm{\Phi}_j^x\in\mathbb{R}^{q_x\times q_\epsilon}$, $\sum_{j=1}^{\infty}\|\bm{\Phi}_j^x\|_\Fr <\infty$, and $q_\epsilon\leq q_x$. Moreover, suppose that $\{\bm{\epsilon}_t\}$ is a sequence of $i.i.d.$ random vectors with zero mean and $\textrm{var}(\bm{\epsilon}_t)=\bm{I}_{q_\epsilon}$, and the  coordinates $\epsilon_{it}$ are mutually independent and $\sigma_x^2$-sub-Gaussian.  
Then, for any $\bm{M}\in\mathbb{R}^{K\times q_x}$ and any $\eta>0$,  where $K$ is a positive integer, it holds
\begin{equation*}
	\mathbb{P}\left \{\left |\frac{1}{T}\sum_{t=1}^{T}\|\bm{M}\bm{x}_t\|_2^2 - \mathbb{E}\left (\|\bm{M}\bm{x}_t\|_2^2\right )\right | \geq \eta \sigma_x^2 \lambda_{\max}(\underline{\bm{\Sigma}}_x)\|\bm{M}\|_{\Fr}^2\right \} \leq 2e^{-c_{\HW}\min(\eta, \eta^2)T},
\end{equation*}
where $c_{\HW}>0$ is an absolute constant. In particular, if  $\bm{x}_t=(\bm{y}_{t-1}^\top, \dots, \bm{y}_{t-L}^\top)^\top$, where $\bm{y}_t=\vect(\cm{Y}_t)$ is from a TAR($L$) process $\{\cm{Y}_t\}$ satisfying the conditions of Lemma \ref{lemma:eigenAR},  then the above inequality holds with $\sigma_x^2=\sigma^2$.
\end{lemma}

\begin{proof}[Proof of Lemma \ref{lemma:hansonw}]
This lemma follows immediately from Lemma S16 in the supplementary material of \cite{zheng2025}.
\end{proof}

\begin{lemma}[Martingale concentration bound]\label{lemma:martgl} 
Suppose that Assumptions \ref{assum:error} and \ref{assum:stationary}(ii)  hold for $\{\bm{\varepsilon}_t\}$ and $\{\bm{x}_t\}$. 	Then, for any $\bm{M}\in\mathbb{R}^{K\times q_x}$ and any $a,b>0$,  where $K$ is a positive integer, it holds
\begin{equation*}
\mathbb{P}\left \{ \sum_{t=1}^{T}\langle\bm{M}\bm{x}_t, \bm{\varepsilon}_t \rangle\geq a, \; \sum_{t=1}^{T}\lVert \bm{M}\bm{x}_t \rVert_2^2 \leq b\right \} \leq  \exp\left \{-\frac{a^2}{2\sigma^2 \lambda_{\max}(\bm{\Sigma}_{\varepsilon})b}\right \}.
\end{equation*}
\end{lemma}

\begin{proof}[Proof of Lemma \ref{lemma:martgl}]
This lemma follows  immediately  from  Lemma S.9  in \cite{sarma2024}.
\end{proof}

\begin{lemma}[Covering number and discretization for CP low-rank tensors]\label{lem:4}
For a given rank $R$, let $\bm{\Pi}(R)=\{\cm{T}\in\mathbb{R}^{p_1\times \cdots\times p_N}:\lVert\cm{T}\rVert_\Fr = 1, \mathrm{rank}(\cm{T})\leq R\}$. For any $0<\epsilon<1$, let $\mathcal{N}(R,\epsilon)$ be a minimal $\epsilon$-net of $\bm{\Pi}(R)$ in the Frobenius norm.
\begin{itemize}
\item [(i)] The cardinality of $\mathcal{N}(R,\epsilon)$ satisfies $\log|\mathcal{N}(R,\epsilon)| \leq R(\sum_{d=1}^N p_d+1) \log\{3(N+1)/\epsilon\}$.
\item [(ii)] For any $\cm{T}\in\mathbb{R}^{p_1\times \cdots\times p_N}$, let $\cm{T}_{[m]}$ be the $p_x$-by-$q_y$ multi-mode matricization obtained by collapsing the first $m$ modes of $\cm{T}$ into rows and the last $n$ modes into columns, where $N=m+n$, $p_x=\prod_{d=1}^{m}p_d$, and $q_y=\prod_{d=m+1}^{N}p_d$. For any $\cm{Z}\in \mathbb{R}^{p_1\times \cdots\times p_N}$ and $\bm{X}\in \mathbb{R}^{q_y\times T}$, it holds \[
\sup_{\cmt{T}\in\bm{\Pi}(R)}\langle \cm{T}, \cm{Z} \rangle \leq (1-\sqrt{2}\epsilon)^{-1} \max_{\cmt{T}\in \mathcal{N}(R,\epsilon)} \langle \cm{T}, \cm{Z} \rangle\]
and
\[
\sup_{\cmt{T}\in\bm{\Pi}(R)}\lVert \cm{T}_{[m]} \bm{X}\rVert_\Fr\leq (1-\sqrt{2}\epsilon)^{-1} \max_{\cmt{T}\in \mathcal{N}(R,\epsilon)}\lVert \cm{T}_{[m]} \bm{X}\rVert_\Fr.\]
\end{itemize}
\end{lemma}

\begin{proof}[Proof of Lemma \ref{lem:4}]
Claim (i) can be proved by a method similar to that for Lemma 6 in \cite{raskutti2011minimax}. Specifically,  any $\cm{T}\in \bm{\Pi}(R)$ has the CP decomposition  $\cm{T}=\sum_{r=1}^R\omega_r\boldsymbol{\beta}_{r,1}\circ\cdots\circ\boldsymbol{\beta}_{r,N}$, where $\boldsymbol{\omega}=(\omega_1,\ldots,\omega_R)^\top \in \pazocal{S}^{R-1}$, and $\bm{\beta}_{r,d}\in \pazocal{S}^{p_d-1}$ for  $1\leq r\leq R$ and $1\leq d\leq N$. Thus, the problem of covering $ \bm{\Pi}(R)$ can be converted to those of covering $\pazocal{S}^{R-1}$ and $\pazocal{S}^{p_d-1}$'s. Denote by $\bar{\pazocal{S}}^{p-1}(\epsilon_1)$  a minimal $\epsilon_1$-net of $\pazocal{S}^{p-1}$ in the Euclidean norm for any positive integer $p$ and $\epsilon_1>0$. Note that $|\bar{\pazocal{S}}^{p-1}(\epsilon_1)|\leq (3/\epsilon_1)^p$. Let $\bar{\bm{\Pi}}(R,\epsilon_1)=\{\bar{\cm{T}}=\sum_{r=1}^R\bar{\omega}_r\bar{\boldsymbol{\beta}}_{r,1}\circ\cdots\circ\bar{\boldsymbol{\beta}}_{r,N}:\bar{\bm{\omega}}\in\bar{\pazocal{S}}^{R-1}(\epsilon_1), \bar{\boldsymbol{\beta}}_{r,d}\in \bar{\pazocal{S}}^{p_d-1}(\epsilon_1), 1\leq r\leq R, 1\leq d\leq N\}$. As a result, $|\bar{\bm{\Pi}}(R,\epsilon_1)| \leq (3/\epsilon_1)^{R(\sum_{d=1}^N p_d+1)}$.

Moreover, for any $\cm{T}\in\bm{\Pi}(R)$, there exists $\bar{\cm{T}}\in \bar{\bm{\Pi}}(R,\epsilon)$ such that $\lVert\boldsymbol{\omega}-\bar{\boldsymbol{\omega}}\rVert_2^2 \leq \epsilon_1$, and $\lVert\boldsymbol{\beta}_{r,d}-\bar{\boldsymbol{\beta}}_{r,d}\rVert_2^2 \leq \epsilon_1$ for  $1\leq r\leq R$ and $1\leq d\leq N$.  Then we have 
\begin{align*}
\lVert \cm{T}-\bar{\cm{T}}\rVert_F^2\leq & \sum_{r=1}^R\lVert\omega_r\boldsymbol{\beta}_{r,1}\circ\cdots\circ\boldsymbol{\beta}_{r,N}- \bar{\omega}_r\bar{\boldsymbol{\beta}}_{r,1}\circ\cdots\circ\bar{\boldsymbol{\beta}}_{r,N}\rVert_\Fr^2 \\
\leq & \sum_{r=1}^R\Bigg\{\Big\lVert\lvert \omega_r-\bar{\omega}_r\rvert\boldsymbol{\beta}_{r,1}\circ\cdots\circ\boldsymbol{\beta}_{r,N}\Big\rVert_\Fr^2+\Big\lVert\bar{\omega}_r\lVert\boldsymbol{\beta}_{r,1}-\bar{\boldsymbol{\beta}}_{r,1}\rVert_2\circ\cdots\circ\boldsymbol{\beta}_{r,N}\Big\rVert_\Fr^2+\cdots\\
&\qquad +\Big\lVert\bar{\omega}_r\bar{\boldsymbol{\beta}}_{r,1}\circ\cdots\circ\lVert\boldsymbol{\beta}_{r,N}-\bar{\boldsymbol{\beta}}_{r,N}\rVert_2\Big\rVert_\Fr^2 \Bigg\}\\
\leq & \epsilon_1+  N \epsilon_1\sum_{r=1}^R \bar{\omega}_r^2 = \epsilon,
\end{align*}
by setting $\epsilon_1=\epsilon/(N+1)$. This shows that $\bar{\bm{\Pi}}(R,\epsilon_1)$ is an $\epsilon$-net of $\bm{\Pi}(R)$, and thus $\log|\mathcal{N}(R,\epsilon)| \leq \log |\bar{\bm{\Pi}}(R,\epsilon_1)| \leq R(\sum_{d=1}^N p_d+1) \log \{3(N+1)/\epsilon\}$. 

Next we prove claim (ii). For any $\cm{T}=\omega_r\boldsymbol{\beta}_{r,1}\circ\cdots\circ\boldsymbol{\beta}_{r,N}\in\bm{\Pi}(R)$,  there exists $\bar{\cm{T}}= \bar{\omega}_r\bar{\boldsymbol{\beta}}_{r,1}\circ\cdots\circ\bar{\boldsymbol{\beta}}_{r,N}\in \bar{\bm{\Pi}}(R,\epsilon)$ such that $\lVert\cm{T}-\bar{\cm{T}}\rVert_\Fr^2 \leq \epsilon$. Moreover, we can partition $\cm{T}-\bar{\cm{T}}=
\cm{T}_1+\cm{T}_2$ such that  $\cm{T}_i/\|\cm{T}_i\|_\Fr\in \bm{\Pi}(R)$ for $i=1,2$. As a result, $\|\cm{T}-\bar{\cm{T}}\|_\Fr^2=
\|\cm{T}_1\|_\Fr^2+\|\cm{T}_2\|_\Fr^2$. Then, by the Cauchy-Schwarz inequality, $\|\cm{T}_1\|_\Fr+\|\cm{T}_2\|_\Fr \leq \sqrt{2}\|\cm{T}-\bar{\cm{T}}\|_\Fr\leq \sqrt{2}\epsilon$. Since 
\[
\langle \cm{T}, \cm{Z} \rangle  = \langle \bar{\cm{T}}, \cm{Z} \rangle + \sum_{i=1}^{2}\|\cm{T}_i\|_\Fr  \langle \cm{T}_i/\|\cm{T}_i\|_\Fr, \cm{Z} \rangle,
\]
we have 
\[
\gamma:=\sup_{\cmt{T}\in\bm{\Pi}(R)}\langle \cm{T}, \cm{Z} \rangle \leq \max_{\cmt{T}\in \mathcal{N}(R,\epsilon)} \langle \cm{T}, \cm{Z} \rangle + \sqrt{2}\epsilon \gamma.
\]
Rearranging the terms, we obtain the first result in claim (ii). The second result can be readily proved along the same lines. The proof of this lemma is complete.
\end{proof}

\begin{lemma}[Covering number and discretization for $\ell_0$-balls]\label{lemma:cover1}
For a given sparsity level $K$, let $\cm{K}(K)=\{\bm{\Delta}\in\mathbb{R}^{p_1\times \cdots\times p_m\times q_1\times\cdots\times q_n}:\lVert\bm{\Delta}\rVert_\Fr = 1, \lVert\vect(\bm{\Delta})\rVert_0\leq K\}$. For any $0<\epsilon<1$, let $\pazocal{N}(K,\epsilon)$ be a minimal $\epsilon$-net of $\cm{K}(K)$ in the Frobenius norm.
\begin{itemize}
\item [(i)] The cardinality of $\pazocal{N}(K,\epsilon)$ satisfies $\log|\pazocal{N}(K,\epsilon)| \leq {p_yq_x\choose K} + K\log(1/\epsilon)$.
\item [(ii)] For any $\cm{Z}\in \mathbb{R}^{p_1\times \cdots\times p_m\times q_1\times\cdots\times q_n}$, it holds \[
\sup_{\bm{\Delta}\in\cmt{K}(K)}\langle \bm{\Delta}, \cm{Z} \rangle \leq (1-\sqrt{2}\epsilon)^{-1} \max_{\bm{\Delta}\in \pazocal{N}(K,\epsilon)} \langle \bm{\Delta}, \cm{Z} \rangle.\]
\end{itemize}
\end{lemma}

\begin{proof}[Proof of Lemma \ref{lemma:cover1}]
Claim (i) is implied directly from the proof of Lemma 6 in \cite{raskutti2011minimax}. We prove claim (ii) as follows. For any $\bm{\Delta}\in\cm{K}(2K)$, there exists $\bm{\Delta}^\prime\in\pazocal{N}(2K,\epsilon)\subset \cm{K}(2K)$ such that $\|\bm{\Delta}^\prime-\bm{\Delta}\|_\Fr \leq \epsilon$. Note that $\|\vect(\bm{\Delta}^\prime-\bm{\Delta})\|_0\leq 2K$. Thus, we can partition $\bm{\Delta}^\prime-\bm{\Delta}=
\bm{\Delta}_1+\bm{\Delta}_2$ such that   $\vect(\bm{\Delta}_i)$ for $i=1,2$ have non-overlapping supports with $\|\vect(\bm{\Delta}_i)\|_0\leq K$. As a result, $\|\bm{\Delta}^\prime-\bm{\Delta}\|_\Fr^2=
\|\bm{\Delta}_1\|_\Fr^2+\|\bm{\Delta}_2\|_\Fr^2$. Then, by the Cauchy-Schwarz inequality, $\|\bm{\Delta}_1\|_\Fr+\|\bm{\Delta}_2\|_\Fr \leq \sqrt{2}\|\bm{\Delta}^\prime-\bm{\Delta}\|_\Fr\leq \sqrt{2}\epsilon$. Since 
\[
\langle \bm{\Delta}, \cm{Z} \rangle  = \langle \bm{\Delta}^\prime, \cm{Z} \rangle + \sum_{i=1}^{2}\|\bm{\Delta}_i\|_\Fr  \langle \bm{\Delta}_i/\|\bm{\Delta}_i\|_\Fr, \cm{Z} \rangle,
\]
we have 
\[
\gamma:=\sup_{\bm{\Delta}\in\cmt{K}(K)}\langle \bm{\Delta}, \cm{Z} \rangle \leq \max_{\bm{\Delta}\in \pazocal{N}(K,\epsilon)} \langle \bm{\Delta}, \cm{Z} \rangle + \sqrt{2}\epsilon \gamma.
\]
Rearranging the terms, the proof of claim (ii) is complete.
\end{proof}

\subsection{Proofs of Theorems \ref{thm:nonsparsetr} and \ref{thm:sparsetr}}

\begin{proof}[Proof of Theorem \ref{thm:nonsparsetr}]
Let $\bm{\Pi}(R)=\{\cm{B}\in\mathbb{R}^{p_1\times \cdots  p_N}:\lVert\cm{B}\rVert_\Fr\leq  1, \mathrm{rank}(\cm{B})\leq R\}$, where $p_{m+d}:=q_d$ for $1\leq d\leq n$, and $N=m+n$. Denote $\bm{\Delta}=\cm{B}-\cm{B}^*$ and $\widetilde{\bm{\Delta}}=\widetilde{\cm{B}}-\cm{B}^*$. Then $\bm{\Delta}/\|\bm{\Delta}\|_\Fr, \widetilde{\bm{\Delta}}/\|\widetilde{\bm{\Delta}}\|_\Fr \in \bm{\Pi}(2R)$. Since the loss function is minimized at $\widetilde{\cm{B}}$ over the parameter space $\bm{\Gamma}(R, \bm{s})$, we have $	\sum_{t=1}^T\lVert\cm{Y}_{t}-\langle\widetilde{\cm{B}}, \cm{X}_{t}\rangle\rVert_\Fr^2\leq \sum_{t=1}^T\lVert\cm{Y}_{t}-\langle \cm{B}^*, \cm{X}_{t}\rangle\rVert_\Fr^2$, which implies
\begin{equation}\label{eq:optim1}
\frac{1}{T}\sum_{t=1}^T\lVert\langle\widetilde{\bm{\Delta}}, \cm{X}_{t}\rangle\rVert_\Fr^2\leq \frac{2}{T}\sum_{t=1}^T\big\langle \widetilde{\bm{\Delta}}, \cm{E}_{t}\circ\cm{X}_{t}\big\rangle.
\end{equation} 

First consider the left side of \eqref{eq:optim1}. Let $\bm{\Delta}_{[m]}$ be the multi-mode matricization of $\bm{\Delta}$ obtained by collapsing its first $m$ modes into rows, and then further define the vectorization $\bm{\delta}_{[m]}=\vect(\bm{\Delta}_{[m]})$.
Similarly we can define $\widetilde{\bm{\Delta}}_{[m]}$ and $\widetilde{\bm{\delta}}_{[m]}=\vect(\widetilde{\bm{\Delta}}_{[m]})$. Note that
\begin{equation}\label{eq:optimleft1}
\sum_{t=1}^T\lVert\langle \bm{\Delta}, \cm{X}_{t}\rangle\rVert_\Fr^2 = \sum_{t=1}^T\lVert \bm{\Delta}_{[m]} \bm{x}_{t}\rVert_2^2 = \lVert \bm{\Delta}_{[m]} \bm{X}\rVert_\Fr^2.
\end{equation}
where $\bm{X}=(\bm{x_1},\dots, \bm{x}_T)$. Denote by $\sigma_x^2$ the sub-Gaussian parameter of  $\epsilon_{it}$ in Assumption \ref{assum:stationary}(ii).
Applying Lemma  \ref{lemma:hansonw} with  $\bm{M}=\bm{\Delta}_{[m]}$ and $\eta= \kappa_1/(2\sigma_x^2\kappa_2)$, by Lemmas \ref{lemma:eigen} and \ref{lemma:hansonw}, we have the pointwise bound for any $\bm{\Delta}\in \bm{\Pi}(2R)$:
\begin{equation}\label{eq:hw1}
\mathbb{P}\left \{\left |\frac{1}{T}\sum_{t=1}^{T}\|\bm{\Delta}_{[m]}\bm{x}_t\|_2^2 - \mathbb{E}\left (\|\bm{\Delta}_{[m]}\bm{x}_t\|_2^2\right )\right | \geq \kappa_1/2 \right \} \leq 2e^{-c_1 \kappa_1^2 T/\kappa_2^2},
\end{equation}
where $c_1=c_{\HW}\min\{1/(2\sigma_x^2), 1/(4\sigma_x^4)\}$. By Lemma \ref{lemma:eigen}, if $\lVert\bm{\Delta}_{[m]}\rVert_\Fr = 1$, then $\kappa_2\geq \lambda_{\max}(\bm{\Sigma}_x)\geq\mathbb{E} (\|\bm{\Delta}_{[m]}\bm{x}_t\|_2^2 ) = \mathbb{E} (\|(\bm{x}_t^\top \otimes \bm{I}_{p_y})\bm{\delta}_{[m]}\|_2^2 ) = \bm{\delta}_{[m]}^\top (\bm{\Sigma}_x\otimes \bm{I}_{p_y}) \bm{\delta}_{[m]} \geq \lambda_{\min}(\bm{\Sigma}_x) \geq \kappa_1$. This, together with \eqref{eq:optimleft1}, further leads to the following pointwise bound for any $\bm{\Delta}\in \bm{\Pi}(2R)$:
\begin{equation}\label{eq:point1}
\mathbb{P}\left \{  \kappa_1/2 \leq \frac{1}{T}	\lVert \bm{\Delta}_{[m]} \bm{X}\rVert_\Fr^2 \leq 3\kappa_2/2 \right \} \leq 2e^{-c_1 \kappa_1^2 T/\kappa_2^2}.
\end{equation}
Next we strengthen this result to union bounds that hold for all $\bm{\Delta}\in \bm{\Pi}(2R)$. Let $\mathcal{N}(2R, \epsilon)$ be a minimal $\epsilon$-net of $\bm{\Pi}(2R)$. Define the event
\[
\mathcal{E}(\epsilon) = \left \{ \forall \bm{\Delta}\in \mathcal{N}(2R, \epsilon):  \sqrt{ \kappa_1/2 } \leq \frac{1}{\sqrt{T}}	\lVert \bm{\Delta}_{[m]} \bm{X}\rVert_\Fr \leq \sqrt{3\kappa_2/2} \right \}.
\]
On the one hand, by Lemma  \ref{lem:4}(ii), it holds 
\begin{equation}\label{eq:side1}
\mathcal{E}(\epsilon) \subset   \left \{ \max_{\bm{\Delta}\in \mathcal{N}(2R, \epsilon)}    \frac{1}{\sqrt{T}}\lVert \bm{\Delta}_{[m]} \bm{X}\rVert_\Fr \leq \sqrt{3\kappa_2/2} \right \}
\subset \left \{ \sup_{\bm{\Delta}\in \bm{\Pi}(2R)}    \frac{1}{\sqrt{T}} \lVert \bm{\Delta}_{[m]} \bm{X}\rVert_\Fr \leq \frac{\sqrt{3\kappa_2/2}}{1-\sqrt{2}\epsilon} \right \}.
\end{equation}
On the other hand, for any $\bm{\Delta}\in \bm{\Pi}(2R)$ and its corresponding $\bar{\bm{\Delta}}\in \mathcal{N}(2R, \epsilon)$, by arguments similar to those for the proof of Lemma \ref{lem:4}(ii), we can show that 
\begin{align*}
\frac{1}{\sqrt{T}}\lVert \bm{\Delta}_{[m]} \bm{X}\rVert_\Fr &\geq \frac{1}{\sqrt{T}}	\lVert \bar{\bm{\Delta}}_{[m]} \bm{X}\rVert_\Fr - \frac{1}{\sqrt{T}}	\lVert (\bar{\bm{\Delta}}_{[m]} - \bm{\Delta}_{[m]} )\bm{X} \rVert_\Fr \\
& \geq \min_{\bar{\bm{\Delta}}\in \mathcal{N}(2R, \epsilon)}  \frac{1}{\sqrt{T}}\lVert \bar{\bm{\Delta}}_{[m]} \bm{X}\rVert_\Fr - \frac{1}{\sqrt{T}} \sum_{i=1}^{2}	\lVert (\bm{\Delta}_i)_{[m]} \bm{X}\rVert_\Fr\\
& \geq \min_{\bar{\bm{\Delta}}\in \mathcal{N}(2R, \epsilon)}  \frac{1}{\sqrt{T}}	\lVert \bar{\bm{\Delta}}_{[m]} \bm{X}\rVert_\Fr - \sqrt{2}\epsilon \sup_{\bm{\Delta}\in \bm{\Pi}(2R)}\frac{1}{\sqrt{T}}\lVert \bm{\Delta}_{[m]} \bm{X}\rVert_\Fr.
\end{align*}
Taking the infimum over $\bm{\Pi}(2R)$ and in view of \eqref{eq:side1}, we can show that on the event $\mathcal{E}(\epsilon)$,
\[
\inf_{\bm{\Delta}\in \bm{\Pi}(2R)} \frac{1}{\sqrt{T}} \lVert \bm{\Delta}_{[m]} \bm{X}\rVert_\Fr \geq \sqrt{\kappa_1/2} - \sqrt{2}\epsilon \cdot  \frac{\sqrt{3\kappa_2/2}}{1-\sqrt{2}\epsilon} \geq \sqrt{\kappa_1/8},
\]
if we choose $\epsilon=(5\sqrt{2})^{-1}$. Thus,
\begin{equation}\label{eq:side2}
\mathcal{E}(\epsilon) \subset   \left \{ \inf_{\bm{\Delta}\in \bm{\Pi}(2R)}    \frac{1}{\sqrt{T}}	 \lVert \bm{\Delta}_{[m]} \bm{X}\rVert_\Fr \geq  \sqrt{\kappa_1/8} \right \}.
\end{equation}
Moreover, by \eqref{eq:point1} and the covering number in Lemma \ref{lem:4}(i), we have 
\begin{align}\label{eq:cover3}
\mathbb{P}(\mathcal{E}^\complement(\epsilon)) &\leq e^{2R(\sum_{d=1}^N p_d+1) \log\{3(N+1)/\epsilon\}} \max_{\bm{\Delta}\in \mathcal{N}(2R, \epsilon)}	\mathbb{P}\left \{  \kappa_1/2 \leq \frac{1}{T}	 \lVert \bm{\Delta}_{[m]} \bm{X}\rVert_\Fr^2 \leq 3\kappa_2/2 \right \} \notag\\
&\leq 2 e^{-2R(\sum_{d=1}^N p_d+1) \log\{3(N+1)/\epsilon\}},
\end{align}
given that  $T\geq 4 c_1^{-1} (\kappa_2/\kappa_1)^2 R(\sum_{d=1}^N p_d+1) \log\{3(N+1)/\epsilon\}$. Combining \eqref{eq:side1}--\eqref{eq:cover3} and the above choice of  $\epsilon=(5\sqrt{2})^{-1}$, we have
\begin{equation}\label{eq:union1}
\sqrt{\kappa_1/8} \leq \inf_{\bm{\Delta}\in \bm{\Pi}(2R)}    \frac{1}{\sqrt{T}}	 \lVert \bm{\Delta}_{[m]} \bm{X}\rVert_\Fr \leq \sup_{\bm{\Delta}\in \bm{\Pi}(2R)}    \frac{1}{\sqrt{T}}	 \lVert \bm{\Delta}_{[m]} \bm{X}\rVert_\Fr \leq  \sqrt{3\kappa_2},
\end{equation}
with probability at least $1-2 e^{-2R(\sum_{d=1}^N p_d+1) \log\{15(N+1)\}}$. In light of \eqref{eq:optimleft1}, this implies 
\begin{equation}\label{eq:lower1}
\mathbb{P}(\mathcal{E}_1^\complement)\leq  2 e^{-2R(\sum_{d=1}^N p_d+1) \log\{15(N+1)\}},
\end{equation}
where $\mathcal{E}_1^\complement$ is the complement of the event
\[
\mathcal{E}_1 = \left \{ \frac{1}{T}\sum_{t=1}^T\lVert\langle\widetilde{\bm{\Delta}}, \cm{X}_{t}\rangle\rVert_\Fr^2 \geq \frac{\kappa_1}{8}  \|\widetilde{\bm{\Delta}}\|_\Fr^2 \right \}.
\]

Next we consider the right side of \eqref{eq:optim1}. By Lemma \ref{lem:4}(ii), we have
\begin{align}\label{eq:cover4}
\frac{2}{T}\sum_{t=1}^T\big\langle \widetilde{\bm{\Delta}}, \cm{E}_{t}\circ\cm{X}_{t}\big\rangle 
&\leq \frac{2}{T}\|\widetilde{\bm{\Delta}}\|_\Fr \sup_{\bm{\Delta}\in \bm{\Pi}(2R)} \big\langle \widetilde{\bm{\Delta}}, \sum_{t=1}^T\cm{E}_{t}\circ\cm{X}_{t}\big\rangle \notag\\
&\leq \frac{8}{T}\|\widetilde{\bm{\Delta}}\|_\Fr
\max_{\bm{\Delta}\in \mathcal{N}(2R,1/2)} \langle \bm{\Delta}, \sum_{t=1}^T\cm{E}_{t}\circ\cm{X}_{t} \rangle,
\end{align}
where $\mathcal{N}(2R,1/2)\subset \bm{\Pi}(2R)$ is the minimal $1/2$-net of $\bm{\Pi}(2R)$ in the Frobenius norm. By arguments similar to those detailed in the proof of Theorem \ref{thm:sparsetr}, for any $\bm{\Delta}\in \bm{\Pi}(2R)$ and  $a>0$, we can show that
\begin{align*}
\mathbb{P}\left \{ \sum_{t=1}^T \langle \bm{\Delta}_{[m]}\bm{x}_{t}, \bm{\varepsilon}_{t} \rangle \geq a \right \}\leq  \exp\left \{-\frac{a^2}{3 \sigma^2 \kappa_2 \lambda_{\max}(\bm{\Sigma}_{\varepsilon})T}\right \}  +  2e^{-c_1  T}.
\end{align*}
As a result,  by setting $a= 4\sqrt{\sigma^2  \kappa_2 \lambda_{\max}(\bm{\Sigma}_{\varepsilon})T R(\sum_{d=1}^N p_d+1) \log\{6(N+1)\} }$, we have
\begin{align*}
\mathbb{P}\left \{ \max_{\bm{\Delta}\in \mathcal{N}(2R,1/2)} \langle \bm{\Delta}, \sum_{t=1}^T\cm{E}_{t}\circ\cm{X}_{t} \rangle \geq a \right \}
&\leq |\mathcal{N}(2R,1/2)|   \max_{\bm{\Delta}\in \mathcal{N}(2R,1/2)} \mathbb{P}\left \{ \langle \bm{\Delta}, \sum_{t=1}^T\cm{E}_{t}\circ\cm{X}_{t} \rangle \geq a \right \}\\
&\leq 3e^{- 2 R(\sum_{d=1}^N p_d+1) \log\{6(N+1)\}},
\end{align*}
provided that $T \geq  4c_1^{-1} R(\sum_{d=1}^N p_d+1) \log\{6(N+1)\}$.
This, together with \eqref{eq:cover4}, implies 
\begin{align}\label{eq:cover5}
\mathbb{P}(\mathcal{E}_2^\complement)\leq  3e^{-2 R(\sum_{d=1}^N p_d+1) \log\{6(N+1)\}},
\end{align}
where $\mathcal{E}_2^\complement$ is the complement of the event
\[
\mathcal{E}_2 = \left \{ \frac{2}{T}\sum_{t=1}^T\big\langle \widetilde{\bm{\Delta}}, \cm{E}_{t}\circ\cm{X}_{t}\big\rangle   \leq  32\|\widetilde{\bm{\Delta}}\|_\Fr \sqrt{\frac{\sigma^2  \kappa_2 \lambda_{\max}(\bm{\Sigma}_{\varepsilon})R(\sum_{d=1}^N p_d+1) \log\{6(N+1)\} }{T} }  \right \}.
\]
In view of  \eqref{eq:optim1}, on the event $\mathcal{E}_1 \cap \mathcal{E}_2$, we have
\[
\|\widetilde{\bm{\Delta}}\|_\Fr \lesssim   \sqrt{\frac{\sigma^2  \kappa_2 \lambda_{\max}(\bm{\Sigma}_{\varepsilon}) R (\sum_{d=1}^N p_d+1) \log\{6(N+1)\}  }{\kappa_1^2 T} },
\]
and consequently,
\[
\frac{1}{T}\sum_{t=1}^T\lVert\langle\widetilde{\bm{\Delta}}, \cm{X}_{t}\rangle\rVert_\Fr^2
\lesssim  \frac{\sigma^2  \kappa_2 \lambda_{\max}(\bm{\Sigma}_{\varepsilon}) R(\sum_{d=1}^N p_d+1) \log\{6(N+1)\} }{\kappa_1 T} 
\] 
By \eqref{eq:lower1} and \eqref{eq:cover5}, we have $\mathbb{P}(\mathcal{E}_1 \cap \mathcal{E}_2) \geq 1-5  e^{-2 R(\sum_{d=1}^N p_d+1) \log\{6(N+1)\}}$. The proof of this theorem is complete.
\end{proof}

\begin{proof}[Proof of Theorem \ref{thm:sparsetr}]
Given  $R$ and $\bm{s}=(s_1,\dots, s_N)\leq (p_1,\dots, p_m, q_1,\dots, q_n)$, where the inequality is elementwise, denote the parameter space of $\cm{B}$ by
\[
\bm{\Gamma}(R, \bm{s}) = \Big\{\cm{B}= \sum_{r=1}^R\omega_r\bm{\beta}_{r,1}\circ\cdots\circ\bm{\beta}_{r,N}: \omega_{r}\in\mathbb{R}, \bm{\beta}_{r,d}\in\pazocal{S}^{p_d-1}, \lVert\bm{\beta}_{r,d}\rVert_0\leq s_d,  r\in[R], d\in[N] \Big\},
\]
where $p_{m+d}:=q_d$ for $1\leq d\leq n$.
Then $\cm{B}^*, \widehat{\cm{B}} \in \bm{\Gamma}(R, \bm{s})$. Note that for any $\cm{B}\in \bm{\Gamma}(R, \bm{s})$, it holds $\|\vect(\cm{B})\|_0 \leq Rs$, where $s=\prod_{d=1}^{N}s_d$.

Denote $\bm{\Delta}=\cm{B}-\cm{B}^*$ and $\widehat{\bm{\Delta}}=\widehat{\cm{B}}-\cm{B}^*$. Since the loss function is minimized at $\widehat{\cm{B}}$ over the parameter space $\bm{\Gamma}(R, \bm{s})$, we have $	\sum_{t=1}^T\lVert\cm{Y}_{t}-\langle\widehat{\cm{B}}, \cm{X}_{t}\rangle\rVert_\Fr^2\leq \sum_{t=1}^T\lVert\cm{Y}_{t}-\langle \cm{B}^*, \cm{X}_{t}\rangle\rVert_\Fr^2$, which implies
\begin{equation}\label{eq:optim}
\frac{1}{T}\sum_{t=1}^T\lVert\langle\widehat{\bm{\Delta}}, \cm{X}_{t}\rangle\rVert_\Fr^2\leq \frac{2}{T}\sum_{t=1}^T\big\langle \widehat{\bm{\Delta}}, \cm{E}_{t}\circ\cm{X}_{t}\big\rangle.
\end{equation} 

First consider the left side of \eqref{eq:optim}. Let $\bm{\Delta}_{[m]}$ be the multi-mode matricization of $\bm{\Delta}$ obtained by collapsing its first $m$ modes into rows.
Similarly we can define $\widehat{\bm{\Delta}}_{[m]}$. Note that
\begin{equation}\label{eq:optimleft}
\lVert\langle\widehat{\bm{\Delta}}, \cm{X}_{t}\rangle\rVert_\Fr^2 = \lVert \widehat{\bm{\Delta}}_{[m]} \bm{x}_{t}\rVert_2^2.
\end{equation}

Let $\cm{K}(2Rs)=\{\bm{\Delta}\in\mathbb{R}^{p_1\times \cdots\times p_m\times q_1\times\cdots\times q_n}:\lVert\bm{\Delta}\rVert_\Fr = 1, \lVert\vect(\bm{\Delta})\rVert_0\leq 2Rs\}$. Note that for any $\cm{B}\in \bm{\Gamma}(R, \bm{s})$, we have $\bm{\Delta} / \| \bm{\Delta} \|_\Fr\in \cm{K}(2Rs)$. In particular, $\widehat{\bm{\Delta}} / \| \widehat{\bm{\Delta}} \|_\Fr\in \cm{K}(2Rs)$.  Denote by $\sigma_x^2$ the sub-Gaussian parameter of  $\epsilon_{it}$ in Assumption \ref{assum:stationary}(ii).
Applying Lemma  \ref{lemma:hansonw} with  $\bm{M}=\bm{\Delta}_{[m]}$ and $\eta= \kappa_1/(2\sigma_x^2\kappa_2)$, by Lemmas \ref{lemma:eigen} and \ref{lemma:hansonw}, we have the pointwise bound for any $\bm{\Delta}\in \cm{K}(2Rs)$:
\begin{equation}\label{eq:hw2}
\mathbb{P}\left \{\left |\frac{1}{T}\sum_{t=1}^{T}\|\bm{\Delta}_{[m]}\bm{x}_t\|_2^2 - \mathbb{E}\left (\|\bm{\Delta}_{[m]}\bm{x}_t\|_2^2\right )\right | \geq \kappa_1/2 \right \} \leq 2e^{-c_1 \kappa_1^2 T/\kappa_2^2},
\end{equation}
where $c_1=c_{\HW}\min\{1/(2\sigma_x^2), 1/(4\sigma_x^4)\}$. By arguments similar to the proof of Lemma F.2 in \cite{Basu2015}, this can be strengthened to a union bound over  $\bm{\Delta}\in \cm{K}(2Rs)$ as follows:
\begin{align*}
&\mathbb{P}\left \{\sup_{\bm{\Delta}\in \cmt{K}(2Rs)} \left |\frac{1}{T}\sum_{t=1}^{T}\|\bm{\Delta}_{[m]}\bm{x}_t\|_2^2 - \mathbb{E}\left (\|\bm{\Delta}_{[m]}\bm{x}_t\|_2^2\right )\right | \geq \kappa_1/2 \right \} \\
& \hspace{10mm}\leq 2e^{-c_1 \kappa_1^2 T/\kappa_2^2 + 2Rs\min\{\log(p_y q_x), \log[21e p_y q_x/(2Rs)]\}} 
\leq 2e^{- 2Rs\min\{\log(p_y q_x), \log[21e p_y q_x/(2Rs)]\}},
\end{align*}
given that $T\geq 4 c_1^{-1}  Rs (\kappa_2/\kappa_1)^2 \min\{\log(p_y q_x), \log[21e p_y q_x/(2Rs)]\}$. Moreover, by Lemma \ref{lemma:eigen}, if $\lVert\bm{\Delta}_{[m]}\rVert_\Fr =1$, then $\mathbb{E} (\|\bm{\Delta}_{[m]}\bm{x}_t\|_2^2 ) = \mathbb{E} (\|(\bm{x}_t^\top \otimes \bm{I}_{p_y})\bm{\delta}_{[m]}\|_2^2 ) = \bm{\delta}_{[m]}^\top (\bm{\Sigma}_x\otimes \bm{I}_{p_y}) \bm{\delta}_{[m]} \geq \lambda_{\min}(\bm{\Sigma}_x) \geq \kappa_1$. As a result, by the triangle inequality, we can show that
\begin{equation*}
\mathbb{P}\left \{ \inf_{\bm{\Delta}\in \cmt{K}(2Rs)} \frac{1}{T}\sum_{t=1}^{T}\|\bm{\Delta}_{[m]}\bm{x}_t\|_2^2 \leq \kappa_1/2  \right \} \leq  2e^{- 2Rs\min\{\log(p_y q_x), \log[21e p_y q_x/(2Rs)]\}},
\end{equation*}
which, together with \eqref{eq:optimleft}, implies
\begin{equation}\label{eq:lower}
\mathbb{P}(\mathcal{E}_1^\complement)\leq  2e^{- 2Rs\min\{\log(p_y q_x), \log[21e p_y q_x/(2Rs)]\}},
\end{equation}
where $\mathcal{E}_1^\complement$ is the complement of the event
\[
\mathcal{E}_1 = \left \{ \frac{1}{T}\sum_{t=1}^T\lVert\langle\widehat{\bm{\Delta}}, \cm{X}_{t}\rangle\rVert_\Fr^2 \geq \frac{\kappa_1}{2}  \|\widehat{\bm{\Delta}}\|_\Fr^2 \right \}.
\]

On the other hand, applying Lemma  \ref{lemma:hansonw} with  $\bm{M}=\bm{\Delta}_{[m]}$ and $\eta= 1/(2\sigma_x^2)$, together with Lemma \ref{lemma:eigen}, we have a  similar pointwise bound for any $\bm{\Delta}\in \cm{K}(2Rs)$:
\begin{equation}\label{eq:hw4}
\mathbb{P}\left \{\left |\frac{1}{T}\sum_{t=1}^{T}\|\bm{\Delta}_{[m]}\bm{x}_t\|_2^2 - \mathbb{E}\left (\|\bm{\Delta}_{[m]}\bm{x}_t\|_2^2\right )\right | \geq \kappa_2/2 \right \} \leq 2e^{-c_1  T}.
\end{equation}
By Lemma \ref{lemma:eigen}, if $\lVert\bm{\Delta}_{[m]}\rVert_\Fr = 1$, then $\mathbb{E} (\|\bm{\Delta}_{[m]}\bm{x}_t\|_2^2 ) = \bm{\delta}_{[m]}^\top (\bm{\Sigma}_x\otimes \bm{I}_{p_y}) \bm{\delta}_{[m]} \leq \lambda_{\max}(\bm{\Sigma}_x) \leq \kappa_2$. Then, by the triangle inequality, for any $\bm{\Delta}\in \cm{K}(2Rs)$, we have the pointwise bound
\begin{equation}\label{eq:upper}
\mathbb{P}\left \{ \frac{1}{T}\sum_{t=1}^{T}\|\bm{\Delta}_{[m]}\bm{x}_t\|_2^2  \geq 3\kappa_2/2 \right \} \leq  2e^{-c_1  T}.
\end{equation}


Next we consider the right side of \eqref{eq:optim}. By Lemma \ref{lemma:cover1}(ii), we have
\begin{align}\label{eq:cover1}
\frac{2}{T}\sum_{t=1}^T\big\langle \widehat{\bm{\Delta}}, \cm{E}_{t}\circ\cm{X}_{t}\big\rangle 
&\leq \frac{2}{T}\|\widehat{\bm{\Delta}}\|_\Fr \sup_{\bm{\Delta}\in \cmt{K}(2Rs)} \big\langle \widehat{\bm{\Delta}}, \sum_{t=1}^T\cm{E}_{t}\circ\cm{X}_{t}\big\rangle \notag\\
&\leq \frac{8}{T}\|\widehat{\bm{\Delta}}\|_\Fr
\max_{\bm{\Delta}\in \pazocal{N}(2Rs,1/2)} \langle \bm{\Delta}, \sum_{t=1}^T\cm{E}_{t}\circ\cm{X}_{t} \rangle,
\end{align}
where $\pazocal{N}(2Rs,1/2)\subset \cm{K}(2Rs)$ is the minimal $1/2$-net of $\cm{K}(2Rs)$ in the Frobenius norm. 

Moreover, by Lemma \ref{lemma:cover1}(i) and the fact that ${p\choose K}\leq \min\{p^K,(ep/K)^K\}$, we have
\[
\log |\pazocal{N}(2Rs,1/2)| \leq 2Rs \min\{\log(p_yq_x), \log[ep_yq_x/(2Rs)]\} + 4Rs;
\]
see also the proof of Lemma F.2 in \cite{Basu2015}.
Note that 
\begin{align*}
\langle \bm{\Delta}, \sum_{t=1}^T\cm{E}_{t}\circ\cm{X}_{t} \rangle = \sum_{t=1}^T \langle \bm{\Delta}_{[m]}, \bm{\varepsilon}_{t}\bm{x}_{t}^\top\rangle = \sum_{t=1}^T \langle \bm{\Delta}_{[m]}\bm{x}_{t}, \bm{\varepsilon}_{t} \rangle.
\end{align*}
For any $\bm{\Delta}\in \cm{K}(2Rs)$ and  $a>0$, by \eqref{eq:upper} and applying Lemma \ref{lemma:martgl} with $\bm{M}=\bm{\Delta}_{[m]}$, we can show that
\begin{align*}
&\mathbb{P}\left \{ \sum_{t=1}^T \langle \bm{\Delta}_{[m]}\bm{x}_{t}, \bm{\varepsilon}_{t} \rangle \geq a \right \}\\ &\hspace{5mm}\leq 
\mathbb{P}\left \{ \sum_{t=1}^{T}\langle\bm{\Delta}_{[m]}\bm{x}_t, \bm{\varepsilon}_t \rangle\geq a, \; \sum_{t=1}^{T}\lVert \bm{\Delta}_{[m]}\bm{x}_t \rVert_2^2 \leq 3\kappa_2 T/2 \right \} + \mathbb{P}\left \{ \sum_{t=1}^{T}\lVert \bm{\Delta}_{[m]}\bm{x}_t \rVert_2^2 \geq 3\kappa_2 T/2\right \} \\
&\hspace{5mm}\leq  \exp\left \{-\frac{a^2}{3 \sigma^2 \kappa_2 \lambda_{\max}(\bm{\Sigma}_{\varepsilon})T}\right \}  +  2e^{-c_1  T}.
\end{align*}
As a result,  by setting $a= 6\sqrt{\sigma^2  \kappa_2 \lambda_{\max}(\bm{\Sigma}_{\varepsilon})T Rs \min\{\log(p_yq_x), \log[ep_yq_x/(2Rs)]\} }$, we have
\begin{align*}
\mathbb{P}\left \{ \max_{\bm{\Delta}\in \pazocal{N}(2Rs,1/2)} \langle \bm{\Delta}, \sum_{t=1}^T\cm{E}_{t}\circ\cm{X}_{t} \rangle \geq a \right \}
&\leq |\pazocal{N}(2Rs,1/2)|   \max_{\bm{\Delta}\in \pazocal{N}(2Rs,1/2)} \mathbb{P}\left \{ \langle \bm{\Delta}, \sum_{t=1}^T\cm{E}_{t}\circ\cm{X}_{t} \rangle \geq a \right \}\\
&\leq 3e^{-6Rs  \min\{\log(p_yq_x), \log[ep_yq_x/(2Rs)]\}},
\end{align*}
given that $T \geq  12 c_1^{-1}  Rs \min\{\log(p_yq_x), \log[ep_yq_x/(2Rs)]\}$. This, together with \eqref{eq:cover1}, implies 
\begin{align}\label{eq:cover2}
\mathbb{P}(\mathcal{E}_2^\complement)\leq  3e^{-6Rs  \min\{\log(p_yq_x), \log[ep_yq_x/(2Rs)]\}},
\end{align}
where $\mathcal{E}_2^\complement$ is the complement of the event
\[
\mathcal{E}_2 = \left \{ \frac{2}{T}\sum_{t=1}^T\big\langle \widehat{\bm{\Delta}}, \cm{E}_{t}\circ\cm{X}_{t}\big\rangle   \leq  48\|\widehat{\bm{\Delta}}\|_\Fr \sqrt{\frac{\sigma^2  \kappa_2 \lambda_{\max}(\bm{\Sigma}_{\varepsilon}) Rs \min\{\log(p_yq_x), \log[ep_yq_x/(2Rs)]\} }{T} }  \right \}.
\]
In view of  \eqref{eq:optim}, on the event $\mathcal{E}_1 \cap \mathcal{E}_2$, we have
\[
\|\widehat{\bm{\Delta}}\|_\Fr \lesssim   \sqrt{\frac{\sigma^2  \kappa_2 \lambda_{\max}(\bm{\Sigma}_{\varepsilon}) Rs \min\{\log(p_yq_x), \log[ep_yq_x/(2Rs)]\} }{\kappa_1^2 T} },
\]
and consequently,
\[
\frac{1}{T}\sum_{t=1}^T\lVert\langle\widehat{\bm{\Delta}}, \cm{X}_{t}\rangle\rVert_\Fr^2
\lesssim  \frac{\sigma^2  \kappa_2 \lambda_{\max}(\bm{\Sigma}_{\varepsilon}) Rs \min\{\log(p_yq_x), \log[ep_yq_x/(2Rs)]\} }{\kappa_1 T} 
\] 
By \eqref{eq:lower} and \eqref{eq:cover2}, we have $\mathbb{P}(\mathcal{E}_1 \cap \mathcal{E}_2) \geq 1-5e^{- 2Rs\min\{\log(p_y q_x), \log[e p_y q_x/(2Rs)]\}}$. The proof is complete.
\end{proof}	

\subsection{Proofs of Corollaries \ref{cor:nonsparsetrTAR} and \ref{cor:sparsetrTAR} }

\begin{proof}[Proof of Corollary \ref{cor:nonsparsetrTAR}]
The proof of Corollary \ref{cor:nonsparsetrTAR} follows closely that of Theorem \ref{cor:nonsparsetrTAR}, with $\bm{x}_t=(\bm{y}_{t-1}^\top, \dots, \bm{y}_{t-L}^\top)^\top$. The first key difference is that instead of \eqref{eq:hw1}, we apply Lemmas \ref{lemma:eigenAR} (as opposed to \ref{lemma:eigen}) and \ref{lemma:hansonw} with $\eta= \kappa_1^{\ar}/(2\sigma^2\kappa_2^{\ar} L)$, which leads to
\begin{equation}\label{eq:hw3}
\mathbb{P}\left \{\left |\frac{1}{T}\sum_{t=1}^{T}\|\bm{\Delta}_{[m]}\bm{x}_t\|_2^2 - \mathbb{E}\left (\|\bm{\Delta}_{[m]}\bm{x}_t\|_2^2\right )\right | \geq \kappa_1^{\ar}/2 \right \} \leq 2e^{-c_1 \kappa_1^2 T/(\kappa_2^{\ar} L)^2},
\end{equation}
where $c_1=c_{\HW}\min\{1/(2\sigma^2), 1/(4\sigma^4)\}$. Then by Lemma \ref{lemma:eigenAR}, if $\lVert\bm{\Delta}_{[m]}\rVert_\Fr = 1$, then $\kappa_2^{\ar}\geq \lambda_{\max}(\bm{\Sigma}_x) \geq \lambda_{\min}(\bm{\Sigma}_x) \geq \kappa_1^{\ar}$. As a result, \eqref{eq:point1} is replaced by
\begin{equation*}
\mathbb{P}\left \{  \kappa_1^{\ar}/2 \leq \frac{1}{T}	\lVert \bm{\Delta}_{[m]} \bm{X}\rVert_\Fr^2 \leq 3\kappa_2^{\ar}/2 \right \} \leq 2e^{-c_1 (\kappa_1^{\ar})^2 T/(\kappa_2^\ar L)^2}.
\end{equation*}
Define the events
\[
\mathcal{E}_1 = \left \{ \frac{1}{T}\sum_{t=1}^T\lVert\langle\widetilde{\bm{\Delta}}, \cm{X}_{t}\rangle\rVert_\Fr^2 \geq \frac{\kappa_1^{\ar}}{8}  \|\widetilde{\bm{\Delta}}\|_\Fr^2 \right \}
\]
and 
\[
\mathcal{E}_2 = \left \{ \frac{2}{T}\sum_{t=1}^T\big\langle \widetilde{\bm{\Delta}}, \cm{E}_{t}\circ\cm{X}_{t}\big\rangle   \leq  32\|\widetilde{\bm{\Delta}}\|_\Fr \sqrt{\frac{\sigma^2  \kappa_2^{\ar} \lambda_{\max}(\bm{\Sigma}_{\varepsilon})R(2\sum_{d=1}^m p_d+L) \log\{6(N+1)\} }{T} }  \right \}.
\]
Analogous to the proof of Theorem \ref{cor:nonsparsetrTAR}, given that $T \gtrsim (\kappa_2^\ar L/\kappa_1^\ar)^2  R(\sum_{d=1}^m p_d +L) \log N$, we can show that $\mathbb{P}(\mathcal{E}_1 \cap \mathcal{E}_2) \geq 1-5  e^{-2 R(2\sum_{d=1}^m p_d+L) \log\{6(N+1)\}}$, and hence the results of this corollary. The proof is complete.
\end{proof}

\begin{proof}[Proof of Corollary \ref{cor:sparsetrTAR}]
The proof of Corollary \ref{cor:sparsetrTAR} is similar to that of Theorem \ref{cor:sparsetrTAR}, with $\bm{x}_t=(\bm{y}_{t-1}^\top, \dots, \bm{y}_{t-L}^\top)^\top$. 	Instead of \eqref{eq:hw2}, we apply Lemmas \ref{lemma:eigenAR} (as opposed to \ref{lemma:eigen}) and \ref{lemma:hansonw} with $\eta= \kappa_1^{\ar}/(2\sigma^2\kappa_2^{\ar} L)$, which leads to
\begin{equation*}
\mathbb{P}\left \{\left |\frac{1}{T}\sum_{t=1}^{T}\|\bm{\Delta}_{[m]}\bm{x}_t\|_2^2 - \mathbb{E}\left (\|\bm{\Delta}_{[m]}\bm{x}_t\|_2^2\right )\right | \geq \kappa_1^\ar/2 \right \} \leq 2e^{-c_1 \kappa_1^2 T/(\kappa_2^\ar L)^2},
\end{equation*}
where $c_1=c_{\HW}\min\{1/(2\sigma^2), 1/(4\sigma^4)\}$.	Another difference is that instead of \eqref{eq:hw4},  by applying Lemma  \ref{lemma:hansonw} with $\eta= 1/(2\sigma^2 L)$, together with Lemma \ref{lemma:eigenAR}, we have the pointwise bound for any $\bm{\Delta}\in \cm{K}(2Rs)$:
\begin{equation*}
\mathbb{P}\left \{\left |\frac{1}{T}\sum_{t=1}^{T}\|\bm{\Delta}_{[m]}\bm{x}_t\|_2^2 - \mathbb{E}\left (\|\bm{\Delta}_{[m]}\bm{x}_t\|_2^2\right )\right | \geq \kappa_2^\ar/2 \right \} \leq 2e^{-c_1  T/L^2}.
\end{equation*}
Define the events
\[
\mathcal{E}_1 = \left \{ \frac{1}{T}\sum_{t=1}^T\lVert\langle\widehat{\bm{\Delta}}, \cm{X}_{t}\rangle\rVert_\Fr^2 \geq \frac{\kappa_1^\ar}{2}  \|\widehat{\bm{\Delta}}\|_\Fr^2 \right \}.
\]
and
\[
\mathcal{E}_2 = \left \{ \frac{2}{T}\sum_{t=1}^T\big\langle \widehat{\bm{\Delta}}, \cm{E}_{t}\circ\cm{X}_{t}\big\rangle   \leq  48\|\widehat{\bm{\Delta}}\|_\Fr \sqrt{\frac{\sigma^2  \kappa_2^\ar \lambda_{\max}(\bm{\Sigma}_{\varepsilon}) Rs \min\{\log(p_y^2 L), \log[e p_y^2 L/(2Rs)]\} }{T} }  \right \}.
\]	
With $T\gtrsim (\kappa_2^\ar L/\kappa_1^\ar)^2 Rs  \min\{\log(p_y^2 L), \log[21e p_y^2 L/(2Rs)]\}$, along the lines of the  proof of Theorem \ref{cor:sparsetrTAR}, we can show that
$\mathbb{P}(\mathcal{E}_1 \cap \mathcal{E}_2) \geq 1-5e^{- 2Rs\min\{\log(p_y^2 L), \log[e p_y^2 L/(2Rs)]\}}$ and  accomplish the  proof of this corollary.
\end{proof}

\putbib[CPAR]	
\end{bibunit}

\begin{thebibliography}{}

\bibitem[Babii et~al., 2024]{Babii2024}
Babii, A., Ghysels, E., and Pan, J. (2024).
\newblock Tensor \textsc{PCA} for factor models.
\newblock {\em Available at SSRN 4791809}.

\bibitem[Babii et~al., 2022]{Babii2022}
Babii, A., Ghysels, E., and Striaukas, J. (2022).
\newblock {Machine Learning Time Series Regressions With an Application to
  Nowcasting}.
\newblock {\em Journal of Business and Economic Statistics}, 40:1094--1106.

\bibitem[Bai and Wang, 2016]{Bai2016}
Bai, J. and Wang, P. (2016).
\newblock Econometric analysis of large factor models.
\newblock {\em Annual Review of Economics}, 8:53--80.

\bibitem[Barigozzi et~al., 2022]{Barigozzi2022}
Barigozzi, M., He, Y., Li, L., and Trapani, L. (2022).
\newblock Statistical inference for large-dimensional tensor factor model by
  iterative projections.
\newblock {\em arXiv preprint arXiv:2206.09800}.

\bibitem[Basu and Michailidis, 2015]{Basu2015}
Basu, S. and Michailidis, G. (2015).
\newblock Regularized estimation in sparse high-dimensional time series models.
\newblock {\em The Annals of Statistics}, 43:1535--1567.

\bibitem[Bi et~al., 2021]{Bi2021}
Bi, X., Tang, X., Yuan, Y., Zhang, Y., and Qu, A. (2021).
\newblock Tensors in statistics.
\newblock {\em Annual Review of Statistics and Its Application}, 8:345--368.

\bibitem[Chan and Qi, 2024]{Chan2024}
Chan, J. C.~C. and Qi, Y. (2024).
\newblock Large bayesian tensor vars with stochastic volatility.
\newblock {\em arXiv preprint arXiv:2409.16132}.

\bibitem[Chang et~al., 2023]{Chang2023}
Chang, J., He, J., Yang, L., and Yao, Q. (2023).
\newblock Modelling matrix time series via a tensor cp-decomposition.
\newblock {\em Journal of the Royal Statistical Society Series B: Statistical
  Methodology}, 85:127--148.

\bibitem[Chen et~al., 2024]{Chen2024}
Chen, B., Han, Y., and Yu, Q. (2024).
\newblock Estimation and inference for \textsc{CP} tensor factor models.
\newblock {\em arXiv preprint arXiv:2406.17278}.

\bibitem[Chen et~al., 2021a]{Chen2021}
Chen, R., Xiao, H., and Yang, D. (2021a).
\newblock Autoregressive models for matrix-valued time series.
\newblock {\em Journal of Econometrics}, 222:539--560.

\bibitem[Chen et~al., 2022]{Chen2022}
Chen, R., Yang, D., and Zhang, C.-H. (2022).
\newblock Factor models for high-dimensional tensor time series.
\newblock {\em Journal of the American Statistical Association}, 117:94--116.

\bibitem[Chen and Lam, 2024]{Chen2024a}
Chen, W. and Lam, C. (2024).
\newblock Rank and factor loadings estimation in time series tensor factor
  model by pre-averaging.
\newblock {\em The Annals of Statistics}, 52:364--391.

\bibitem[Chen et~al., 2021b]{chen2021tensor}
Chen, Y.-L., Kolar, M., and Tsay, R.~S. (2021b).
\newblock Tensor canonical correlation analysis with convergence and
  statistical guarantees.
\newblock {\em Journal of Computational and Graphical Statistics}, 30:728--744.
\newblock Data available at https://github.com/youlinchen/TCCA.

\bibitem[Feng et~al., 2021]{feng2021brain}
Feng, L., Bi, X., and Zhang, H. (2021).
\newblock Brain regions identified as being associated with verbal reasoning
  through the use of imaging regression via internal variation.
\newblock {\em Journal of the American Statistical Association},
  116(533):144--158.

\bibitem[Foroni et~al., 2015]{Foroni2015}
Foroni, C., Marcellino, M., and Schumacher, C. (2015).
\newblock Unrestricted mixed data sampling (midas): Midas regressions with
  unrestricted lag polynomials.
\newblock {\em Journal of the Royal Statistical Society. Series A (Statistics
  in Society)}, 178:57--82.

\bibitem[Foroni and Marcellino, 2013]{Foroni2013}
Foroni, C. and Marcellino, M.~G. (2013).
\newblock A survey of econometric methods for mixed-frequency data.
\newblock {\em Available at SSRN 2268912}.

\bibitem[Ghysels et~al., 2006]{midas2006}
Ghysels, E., Sinko, A., and Valkanov, R. (2006).
\newblock \textsc{MIDAS} regressions: Further results and new directions.
\newblock {\em Econometric Reviews}, 26:53--90.

\bibitem[Han et~al., 2024]{Han2024}
Han, Y., Yang, D., Zhang, C.-H., and Chen, R. (2024).
\newblock {CP factor model for dynamic tensors}.
\newblock {\em Journal of the Royal Statistical Society Series B: Statistical
  Methodology}, 86:1383–1413.

\bibitem[Huang et~al., 2025]{Huang2025}
Huang, F., Lu, K., Zheng, Y., and Li, G. (2025).
\newblock Supervised factor modeling for high-dimensional linear time series.
\newblock {\em Journal of Econometrics}, 249:105995.

\bibitem[Kolda and Bader, 2009]{Kolda2009}
Kolda, T.~G. and Bader, B.~W. (2009).
\newblock Tensor decompositions and applications.
\newblock {\em SIAM Review}, 51:455--500.

\bibitem[Kruskal, 1976]{kruskal1976more}
Kruskal, J.~B. (1976).
\newblock More factors than subjects, tests and treatments: an indeterminacy
  theorem for canonical decomposition and individual differences scaling.
\newblock {\em Psychometrika}, 41:281--293.

\bibitem[Liu et~al., 2021]{Liu2021}
Liu, J., Zhu, C., Long, Z., and Liu, Y. (2021).
\newblock Tensor regression.
\newblock {\em Foundations and Trends® in Machine Learning}, 14:379--565.

\bibitem[Lock, 2018]{Lock2018}
Lock, E.~F. (2018).
\newblock Tensor-on-tensor regression.
\newblock {\em Journal of Computational and Graphical Statistics}, 27:638--647.

\bibitem[Luo and Zhang, 2024]{LZ2024}
Luo, Y. and Zhang, A.~R. (2024).
\newblock Tensor-on-tensor regression: Riemannian optimization,
  over-parameterization, statistical-computational gap, and their interplay.
\newblock {\em Annals of Statistics}.
\newblock to appear.

\bibitem[L\"{u}tkepohl, 2005]{Luetkepohl2005}
L\"{u}tkepohl, H. (2005).
\newblock {\em New Introduction to Multiple Time Series Analysis}.
\newblock Springer Science \& Business Media.

\bibitem[McCracken and Ng, 2016]{McCracken2016}
McCracken, M.~W. and Ng, S. (2016).
\newblock {FRED-MD}: A monthly database for macroeconomic research.
\newblock {\em Journal of Business \& Economic Statistics}, 34:574--589.

\bibitem[McCracken and Ng, 2021]{McCracken2021}
McCracken, M.~W. and Ng, S. (2021).
\newblock {FRED-QD: A quarterly database for macroeconomic research}.
\newblock {\em Federal Reserve Bank of St. Louis Review}, 103:1--44.

\bibitem[Raskutti et~al., 2019]{raskutti2019convex}
Raskutti, G., Yuan, M., and Chen, H. (2019).
\newblock Convex regularization for high-dimensional multi-response tensor
  regression.
\newblock {\em Annals of Statistics}, 47:1554--1584.

\bibitem[Rogers et~al., 2013]{rogers2013multilinear}
Rogers, M., Li, L., and Russell, S.~J. (2013).
\newblock Multilinear dynamical systems for tensor time series.
\newblock {\em Advances in Neural Information Processing Systems}, 26.

\bibitem[Samadi and Billard, 2024]{Samadi2024}
Samadi, S.~Y. and Billard, L. (2024).
\newblock On a matrix-valued autoregressive model.
\newblock {\em Journal of Time Series Analysis}.
\newblock to appear.

\bibitem[Sidiropoulos and Bro, 2000]{Sidiropoulos2000}
Sidiropoulos, N.~D. and Bro, R. (2000).
\newblock On the uniqueness of multilinear decomposition of n-way arrays.
\newblock {\em Journal of Chemometrics}, 14:229--239.

\bibitem[Sun and Li, 2017]{sun2017store}
Sun, W.~W. and Li, L. (2017).
\newblock {STORE: Sparse tensor response regression and neuroimaging analysis}.
\newblock {\em Journal of Machine Learning Research}, 18:1--37.

\bibitem[Ten~Berge and Sidiropoulos, 2002]{ten2002uniqueness}
Ten~Berge, J.~M. and Sidiropoulos, N.~D. (2002).
\newblock On uniqueness in \textsc{CANDECOMP/PARAFAC}.
\newblock {\em Psychometrika}, 67:399--409.

\bibitem[Tsay, 2024]{Tsay2024}
Tsay, R.~S. (2024).
\newblock Matrix-variate time series analysis: A brief review and some new
  developments.
\newblock {\em International Statistical Review}, 92:246--262.

\bibitem[Velu et~al., 1986]{Velu1986}
Velu, R.~P., Reinsel, G.~C., and Wichern, D.~W. (1986).
\newblock Reduced rank models for multiple time series.
\newblock {\em Biometrika}, 73:105--118.

\bibitem[Wang et~al., 2019]{wang2019factor}
Wang, D., Liu, X., and Chen, R. (2019).
\newblock Factor models for matrix-valued high-dimensional time series.
\newblock {\em Journal of Econometrics}, 208:231--248.

\bibitem[Wang et~al., 2024]{wang2024high}
Wang, D., Zheng, Y., and Li, G. (2024).
\newblock High-dimensional low-rank tensor autoregressive time series modeling.
\newblock {\em Journal of Econometrics}, 238(1):105544.

\bibitem[Wang et~al., 2022]{wang2022high}
Wang, D., Zheng, Y., Lian, H., and Li, G. (2022).
\newblock High-dimensional vector autoregressive time series modeling via
  tensor decomposition.
\newblock {\em Journal of the American Statistical Association},
  117(539):1338--1356.

\bibitem[Wang et~al., 2017]{wang2017generalized}
Wang, X., Zhu, H., and Initiative, A. D.~N. (2017).
\newblock Generalized scalar-on-image regression models via total variation.
\newblock {\em Journal of the American Statistical Association},
  112(519):1156--1168.

\bibitem[Zeng and Ng, 2021]{Zeng2021}
Zeng, C. and Ng, M.~K. (2021).
\newblock Incremental \textsc{CP} tensor decomposition by alternating
  minimization method.
\newblock {\em SIAM Journal on Matrix Analysis and Applications}, 42:832--858.

\bibitem[Zheng, 2025]{zheng2025}
Zheng, Y. (2025).
\newblock An interpretable and efficient infinite-order vector autoregressive
  model for high-dimensional time series.
\newblock {\em Journal of the American Statistical Association}, 120:212--225.

\bibitem[Zhou et~al., 2013]{zhou2013tensor}
Zhou, H., Li, L., and Zhu, H. (2013).
\newblock Tensor regression with applications in neuroimaging data analysis.
\newblock {\em Journal of the American Statistical Association},
  108(502):540--552.

\end{thebibliography}


\begin{thebibliography}{}

\bibitem[Bai and Wang, 2016]{Bai2016}
Bai, J. and Wang, P. (2016).
\newblock Econometric analysis of large factor models.
\newblock {\em Annual Review of Economics}, 8:53--80.

\bibitem[Basu and Michailidis, 2015]{Basu2015}
Basu, S. and Michailidis, G. (2015).
\newblock Regularized estimation in sparse high-dimensional time series models.
\newblock {\em The Annals of Statistics}, 43:1535--1567.

\bibitem[Beck and Teboulle, 2009]{Beck2009}
Beck, A. and Teboulle, M. (2009).
\newblock Iterative shrinkage and thresholding for linear inverse problems.
\newblock {\em SIAM Journal on Imaging Sciences}, 2(1):183--202.

\bibitem[Chi et~al., 2019]{Chi2019}
Chi, Y., Lu, Y.~M., and Chen, Y. (2019).
\newblock Nonconvex optimization meets low-rank matrix factorization: An
  overview.
\newblock {\em IEEE Transactions on Signal Processing}, 67:5239--5269.

\bibitem[De~Lathauwer et~al., 2000]{de2000multilinear}
De~Lathauwer, L., De~Moor, B., and Vandewalle, J. (2000).
\newblock A multilinear singular value decomposition.
\newblock {\em SIAM journal on Matrix Analysis and Applications},
  21(4):1253--1278.

\bibitem[Han et~al., 2024]{Han2024}
Han, Y., Yang, D., Zhang, C.-H., and Chen, R. (2024).
\newblock {CP factor model for dynamic tensors}.
\newblock {\em Journal of the Royal Statistical Society Series B: Statistical
  Methodology}, 86:1383–1413.

\bibitem[Huang et~al., 2024]{sarma2024}
Huang, F., Lu, K., and Zheng, Y. (2024).
\newblock \textsc{SARMA: S}calable low-rank high-dimensional autoregressive
  moving averages via tensor decomposition.
\newblock {\em ArXiv preprint arXiv:2405.00626}.

\bibitem[Kolda and Bader, 2009]{Kolda2009}
Kolda, T.~G. and Bader, B.~W. (2009).
\newblock Tensor decompositions and applications.
\newblock {\em SIAM Review}, 51:455--500.

\bibitem[Raskutti et~al., 2011]{raskutti2011minimax}
Raskutti, G., Wainwright, M.~J., and Yu, B. (2011).
\newblock Minimax rates of estimation for high-dimensional linear regression
  over $\ell_q$ -balls.
\newblock {\em IEEE Transactions on Information Theory}, 57(10):6976--6994.

\bibitem[Sun and Li, 2017]{sun2017store}
Sun, W.~W. and Li, L. (2017).
\newblock {STORE: Sparse tensor response regression and neuroimaging analysis}.
\newblock {\em Journal of Machine Learning Research}, 18:1--37.

\bibitem[Sun et~al., 2017]{sun2017provable}
Sun, W.~W., Lu, J., Liu, H., and Cheng, G. (2017).
\newblock Provable sparse tensor decomposition.
\newblock {\em Journal of the Royal Statistical Society Series B: Statistical
  Methodology}, 79(3):899--916.

\bibitem[Tibshirani, 1996]{Tibshirani1996}
Tibshirani, R. (1996).
\newblock Regression shrinkage and selection via the lasso.
\newblock {\em Journal of the Royal Statistical Society: Series B},
  58:267–288.

\bibitem[Wang et~al., 2024]{wang2024high}
Wang, D., Zheng, Y., and Li, G. (2024).
\newblock High-dimensional low-rank tensor autoregressive time series modeling.
\newblock {\em Journal of Econometrics}, 238(1):105544.

\bibitem[Wang et~al., 2022]{wang2022high}
Wang, D., Zheng, Y., Lian, H., and Li, G. (2022).
\newblock High-dimensional vector autoregressive time series modeling via
  tensor decomposition.
\newblock {\em Journal of the American Statistical Association},
  117(539):1338--1356.

\bibitem[Wilms et~al., 2023]{Wilms2023}
Wilms, I., Basu, S., Bien, J., and Matteson, D. (2023).
\newblock Sparse identification and estimation of large-scale vector
  autoregressive moving averages.
\newblock {\em Journal of the American Statistical Association}, 118:571–582.

\bibitem[Zheng, 2025]{zheng2025}
Zheng, Y. (2025).
\newblock An interpretable and efficient infinite-order vector autoregressive
  model for high-dimensional time series.
\newblock {\em Journal of the American Statistical Association}, 120:212--225.

\end{thebibliography}
\end{document}